\documentclass[english,10pt]{amsart}
\usepackage[utf8]{inputenc}
\usepackage{amsfonts}
\usepackage{amssymb, amsmath}
\usepackage{amsthm}
\usepackage{mathtools}
\usepackage{dsfont}
\usepackage{bm}
\usepackage{verbatim}
\usepackage{xcolor}
\usepackage{url}
\usepackage{subcaption}
\usepackage{float}
\usepackage{caption}
\usepackage[pdftex]{hyperref}
\hypersetup{plainpages=false,unicode=true,pdffitwindow=true}
\usepackage{graphicx}
\usepackage{comment}
\usepackage{hyperref}
\usepackage{shuffle}
\usepackage{algorithm}
\usepackage{algpseudocode}

\theoremstyle{plain}
\newtheorem{theorem}{Theorem}[section]
\newtheorem{lemma}[theorem]{Lemma}
\newtheorem{proposition}[theorem]{Proposition}
\newtheorem{corollary}[theorem]{Corollary}

\theoremstyle{definition}
\newtheorem{definition}[theorem]{Definition}

\newtheorem{remark}[theorem]{Remark}

\newtheorem{assumption}[theorem]{Assumption}

\newcommand{\N}{\mathbb{N}}
\newcommand{\R}{\mathbb{R}}
\newcommand{\C}{\mathbb{C}}
\newcommand{\E}{\mathbb{E}}

\newcommand{\lip}{\textup{Lip}}
\newcommand{\eps}{\varepsilon}

\newcommand{\abs}[1]{\left\lvert #1 \right\rvert}

\newcommand{\tol}{\textup{TOL}}
\DeclareMathOperator*{\arginf}{arginf}

\newcommand{\dd}{\mathrm{d}}
\newcommand{\sdd}{\,\dd}

\newcommand{\vertiii}[1]{{\left\vert\kern-0.25ex\left\vert\kern-0.25ex\left\vert #1 
    \right\vert\kern-0.25ex\right\vert\kern-0.25ex\right\vert}}

\usepackage{todonotes}

\numberwithin{equation}{section}

\usepackage[backend=bibtex,style=alphabetic]{biblatex}

\bibliography{literature.bib}

\title[Markovian Approximations]{Weak Markovian Approximations of Rough Heston}

\author{Christian Bayer}
\address{Weierstrass Institute, Mohrenstraße 39, 10117 Berlin, Germany}
\email{christian.bayer@wias-berlin.de}

\author{Simon Breneis}
\address{Weierstrass Institute, Mohrenstraße 39, 10117 Berlin, Germany}
\email{simon.breneis@wias-berlin.de}

\makeatletter
\@namedef{subjclassname@2020}{\textup{2020} Mathematics Subject Classification}
\makeatother

\date{\today}
\subjclass[2020]{91G60, 91G20}
\keywords{Rough Heston model, Markovian approximations, Weak error}

\thanks{C.B. and S.B. gratefully acknowledge the support of the DFG through the IRTG 2544. The authors would also like to thank Eduardo Abi Jaber for encouraging them to consider Hurst parameters $H<0$.}

\begin{document}
\maketitle

\begin{abstract}
    The rough Heston model is a very popular recent model in mathematical finance; however, the lack of Markov and semimartingale properties poses significant challenges in both theory and practice. A way to resolve this problem is to use Markovian approximations of the model. Several previous works have shown that these approximations can be very accurate even when the number of additional factors is very low. Existing error analysis is largely based on the strong error, corresponding to the $L^2$ distance between the kernels. Extending earlier results by [Abi Jaber and El Euch, SIAM Journal on Financial Mathematics 10(2):309--349, 2019], we show that the weak error of the Markovian approximations can be bounded using the $L^1$-error in the kernel approximation for general classes of payoff functions for European style options. Moreover, we give specific Markovian approximations which converge super-polynomially in the number of dimensions, and illustrate their numerical superiority in option pricing compared to previously existing approximations. The new approximations also work for the hyper-rough case $H > -1/2$.
\end{abstract}

\section{Introduction}

Rough volatility \cite{roughvolbook} is a now established paradigm for modeling of equity markets, which provides excellent fits to market data. However, theoretical analysis and numerical approximation typically becomes more challenging as rough volatility models are neither Markov processes nor semimartingales. Apart from the rough Bergomi model \cite{bayer2016pricing}, the workhorse model of rough volatility is the \emph{rough Heston model} \cite{el2019characteristic} given by
\begin{align}
\dd S_t &= \sqrt{V_t} S_t \left(\rho \sdd W_t + \sqrt{1-\rho^2} \sdd B_t\right),\qquad S_0 = S_0,\label{eqn:RHestonStock}\\
V_t &= V_0 + \int_0^t K(t-s) (\theta - \lambda V_s) \sdd s + \int_0^t K(t-s) \nu\sqrt{V_s} \sdd W_s,\label{eqn:RHestonVol}
\end{align}
where $\theta, \lambda, \nu > 0$, $\rho\in[-1, 1]$, where $(B_t,W_t)$ is a two-dimensional standard Brownian motion and where $K$ is the fractional kernel given by
\begin{equation}\label{eqn:DefinitionOfK}
K(t) = \frac{t^{H-1/2}}{\Gamma(H+1/2)},
\end{equation}
with Hurst parameter $H\in(-1/2, 1/2)$. 

The kernel $K$ in \eqref{eqn:RHestonVol} introduces a dependence of the volatility process $V$ on the past, ensuring that the volatility has memory. However, this also means that the process $(S, V)$ is neither a Markov process, nor a semi-martingale. Furthermore, the sample paths of $V$ are only $(H-\eps)$-Hölder continuous, for all $\eps > 0$, where $H\in(0, 1/2)$ is the Hurst parameter in \eqref{eqn:DefinitionOfK}. This leads to significant problems of the rough Heston model both in terms of theoretical analysis, and for simulation of sample paths, and option pricing more generally. Note, however, that the rough Heston model is an affine Volterra process \cite{abi2019affine}, and can be analyzed as such, implying weak existence and uniqueness as well as a semi-explicit formula for the characteristic function -- in terms of a fractional Riccati equation.

One remedy for these challenging properties of the model is to use Markovian approximations of $V$. More precisely, note that the kernel $K$ is completely monotone, and can thus be written as
\begin{equation}\label{eqn:IntegralRepresentationOfK}
    K(t) = \int_0^\infty e^{-xt} \mu(\dd x),\ \mu(\dd x) \coloneqq c_H x^{-H-1/2} \sdd x,\ c_H \coloneqq \frac{1}{\Gamma(H+1/2)\Gamma(1/2-H)},
\end{equation}
for all $t > 0$, where $\Gamma$ is the gamma function.

Taking the special example of fractional Brownian motion (fBm) $W^H$ for a moment, we observe that
\begin{align}
W^H_t &= \int_0^t K(t-s) \sdd W_s = \int_0^t \int_0^\infty e^{-x(t-s)} \mu(\dd x) \sdd W_s\nonumber \\
&= \int_0^\infty \int_0^t e^{-x(t-s)} \sdd W_s \mu(\dd x) =: \int_0^\infty Y_t(x) \mu(\dd x),\label{eqn:fBmExample}
\end{align}
where $W$ is a Brownian motion, and where we used the stochastic Fubini theorem. Note that $Y_t(x)$ is an Ornstein-Uhlenbeck (OU) process driven by $W$ with mean-reversion rate $x$. Indeed, in \cite{carmona1998fractional} it was shown that $(Y_t(x))_{x>0}$ is an infinite-dimensional Markov process. Therefore, fBm is a linear functional of an infinite-dimensional Markov process. 

The representation \eqref{eqn:fBmExample} indicates a natural way of approximating fBm by a finite-dimensional Markov process. We simply discretize the integral over $\mu(\dd x)$ in \eqref{eqn:fBmExample}, to obtain an approximation of $W^H$ as a linear functional of a finite-dimensional OU (and hence Markov) process. Furthermore, it is easy to see that discretizing the integral in \eqref{eqn:fBmExample} is equivalent to discretizing the integral in \eqref{eqn:IntegralRepresentationOfK}, i.e. to approximating $K$ by a discrete sum of exponential functions $K^N$ given by 
\begin{equation}\label{eqn:KNForm}
K^N(t) \coloneqq \sum_{i=1}^N w_i e^{-x_i t},
\end{equation}
for some non-negative nodes $(x_i)_{i=1}^N$ and non-negative weights $(w_i)_{i=1}^N$.

Furthermore, using the approximation $K^N$ of $K$, we can define the approximation $(S^N, V^N)$ of $(S, V)$ by 
\begin{align*}
\dd S^N_t &= \sqrt{V^N_t} S^N_t \left(\rho \sdd W_t + \sqrt{1-\rho^2} \sdd B_t\right),\qquad S_0 = S_0,\\
V^N_t &= V_0 + \int_0^t K^N(t-s) (\theta - \lambda V^N_s) \sdd s + \int_0^t K^N(t-s) \nu\sqrt{V^N_s} \sdd W_s.
\end{align*}

It has been shown, e.g. in \cite{abi2019multifactor} and in  \cite[Proposition 2.1]{alfonsi2021approximation}, that $V^N$ is the solution to an $N$-dimensional stochastic differential equation (SDE). The precise form of this SDE will however not be important for our purposes.

Given the Markovian approximation $(S^N, V^N)$ of $(S, V)$, we are of course interested in proving error bounds and convergence rates as $N\to\infty$. Assuming Lipschitz-continuous coefficients $b$ and $\sigma$, \cite{alfonsi2021approximation} proved for general stochastic Volterra equations with Hurst parameter $H>0$ given by
\begin{equation}\label{eqn:GeneralSVE}
X_t = X_0 + \int_0^t K(t-s) b(X_s) \sdd s + \int_0^t K(t-s) \sigma(X_s) \sdd W_s
\end{equation}
the strong error bound
\begin{equation}\label{eqn:StrongErrorBound}
\E \abs{X_T - X^N_T}^2 \le C\int_0^T \abs{K(t) - K^N(t)}^2 \sdd t,
\end{equation}
for some $C>0$, where $X^N$ is again defined by replacing $K$ with $K^N$ in \eqref{eqn:GeneralSVE}.

However, this bound is not directly applicable to the rough Heston model due to the singularity of the square root $\sqrt{V_t}$ in $0$. Moreover, the $L^2$-error in $K$ can converge very slowly for small Hurst parameters $1/2 \gg H > 0$. Indeed, since $K(t) \approx t^{H-1/2}$, the singularity at $t=0$ of $K(t)^2 \approx t^{2H-1}$ is barely integrable, leading to slow convergence rates. For example, using Gaussian quadrature rules, it was shown in \cite{bayer2023markovian} that one can achieve a convergence rate of
\begin{equation}\label{eqn:L2ErrorIfWeUsedPreviousQuadratureRule}
\left(\int_0^T \abs{K(t) - K^N(t)}^2 \sdd t\right)^{1/2} \le C\exp\left(-1.064 \left(1 + \frac{H}{3/2 - H}\right)^{-1/2}\sqrt{HN}\right).
\end{equation}
While this rate is superpolynomial, it is still very slow for small $H \approx 0$ -- ant not applicable at all in the hyper-rough case $-1/2 < H \le 0$, see \cite{jusselin2020no}. 

At the same time, a strong error bound as in \eqref{eqn:StrongErrorBound} may be too stringent for many practical applications, where we are often satisfied with a weak error bound. Indeed, in \cite[Proposition 4.3]{abi2019multifactor} is was shown for the specific case of the call option in the rough Heston model with some strike price $P$ that we have the much better error bound
\begin{equation}\label{eqn:WeakErrorBound}
\abs{\E(S_T - P)^+ - \E(S^N_T - P)^+} \le C\int_0^T \abs{K(t) - K^N(t)} \sdd t,
\end{equation}
where $x^+ \coloneqq x \lor 0$. Using the $L^1$-error instead of the $L^2$-error in the kernel is a significant improvement, since $K$ has much better integrability properties at $t=0$ than $K^2$. However, the authors of \cite{abi2019multifactor} did not seem to proceed to prove the weak error bound \eqref{eqn:WeakErrorBound} for more general payoff functions. Moreover, they proved \eqref{eqn:WeakErrorBound} only under very strong assumptions on the approximating kernel $K^N$. The first aim of this paper is hence to extend and greatly generalize the results of \cite{abi2019multifactor}. More precisely, in Corollary \ref{cor:FinalCorollaryBoundedPayoff} we prove that for every $H\in(-1/2,1/2)$ and for ``nice'' payoff functions $h:\R_+\to\R$, there exists a constant $C>0$ such that
\begin{equation}\label{eqn:WeakErrorBoundIntroduction}
\abs{\E h(S_T) - \E h(S^N_T)} \le C\int_0^T \abs{K(t) - K^N(t)} \sdd t.
\end{equation}
Additionally, the approximating kernels $K^N$ only have to satisfy the rather mild Assumption \ref{ass:TheAssumption}.

Having proved the error bound \eqref{eqn:WeakErrorBoundIntroduction}, we want to find a good approximation $K^N$ of $K$ minimizing the $L^1$-norm, and prove a convergence rate similar to \eqref{eqn:L2ErrorIfWeUsedPreviousQuadratureRule}. We give two slightly different versions in Theorems \ref{thm:TheL1TheoremGeometric} and \ref{thm:TheL1TheoremNonGeometric}. In particular, in Theorem \ref{thm:TheL1TheoremNonGeometric} we show that we can achieve the convergence rate
\begin{equation}\label{eqn:L1ConvergenceRateIntroduction}
\int_0^T \abs{K(t) - K^N(t)} \sdd t \le C\exp\left(-2.38 \sqrt{(H+1/2)N}\right)
\end{equation}
for any $H > -1/2.$ This is a much faster convergence rate than the one in \eqref{eqn:L2ErrorIfWeUsedPreviousQuadratureRule}, especially for tiny $H>0$, and moreover, it is also valid for negative $H>-1/2$.

The paper is structured as follows. In Section \ref{sec:MainTheorySection}, we prove \eqref{eqn:WeakErrorBoundIntroduction}, i.e. that the weak error of the Markovian approximations of the rough Heston model can be bounded by the $L^1$-error in the kernel. Next, in Section \ref{sec:ErrorRates} we give two separate (but similar) approaches for achieving small $L^1$-errors in the kernel $K^N$. Especially the first approach using geometric Gaussian rules in Section \ref{sec:GeometricGaussianApproximations} resembles the approach in \cite{bayer2023markovian}, except that we refine their results in several ways. The second approach in Section \ref{sec:NonGeometricGaussianApproximations} using non-geometric Gaussian rules is essentially an improvement of the first approach, yielding the rate in \eqref{eqn:L1ConvergenceRateIntroduction}. We decided to give both methods, since the geometric rules are simpler and easier to understand, and since the majority of the proof for the geometric rules can be transferred to the non-geometric rules. Next, in Section \ref{sec:Numerics}, we show that by optimizing the $L^1$-error we can indeed achieve very fast convergence rates and achieve small errors even for small values of the approximating dimension $N$. We also illustrate that minimizing the $L^1$-error leads to much faster convergence rates in the weak errors than minimizing the $L^2$-error. Finally, some technical results and some of the algorithms used in Section \ref{sec:Numerics} are delegated to the appendices.

\section{Weak error of Markovian approximations}\label{sec:MainTheorySection}

Throughout this section, we will work under the following assumption.

\begin{assumption}\label{ass:TheAssumption}
Assume that $K^N$ are chosen such that
\begin{equation*}
e_N \coloneqq \int_0^T \abs{K(t) - K^N(t)} \sdd t \to 0,
\end{equation*}
and that there exists a constant $C>0$ such that for all $t\in[0, T]$ and $N\ge 1$, we have
\begin{equation}\label{eqn:SecondAssumption}
K^N(t) \le C K(t),\qquad \abs{(K^N)'(t)} \le C\abs{K'(t)}.
\end{equation}
\end{assumption}

We remark that sequences of such nodes and weights satisfying Assumption \ref{ass:TheAssumption} have been given in many previous works, among others in \cite{abi2019multifactor, alfonsi2021approximation, bayer2023markovian, harms2019strong} (although \eqref{eqn:SecondAssumption} is usually not verified in these references, it is merely a rather weak and technical assumption). Also, as in \cite[Theorem 3.1]{abi2019multifactor}, we have strong existence and uniqueness of the process $(S^N, V^N)$, and non-negativity of $V^N$.

We now recall some basic definitions and facts from \cite{abi2019multifactor}. First, we define the characteristic function of the log-stock price $\log(S_T)$ as $$\varphi(z) \coloneqq \E\left[\exp\left(z\log(S_T/S_0)\right)\right].$$ Moreover, we note that $\varphi$ can be given in semi-closed form as 
\begin{equation}\label{eqn:CharacteristicFunctionSemiClosedForm}
\varphi(z) = \exp\left(\int_0^T F(z, \psi(T-t,z))g(t)\sdd t\right),
\end{equation}
where $g$ is defined by 
\begin{equation}\label{eqn:DefinitionOfg}
g(t) = V_0 + \theta\int_0^t K(s) \sdd s = V_0 + \theta \frac{t^{H+1/2}}{\Gamma(H + 3/2)},
\end{equation}
and where $\psi$ is the solution to the fractional Riccati equation
\begin{equation}\label{eqn:PsiRiccatiEquation}
\psi(t, z) = \int_0^t K(t-s)F(z, \psi(s, z))\sdd s,
\end{equation}
with
\begin{equation*}
F(z, x) = \frac{1}{2}(z^2 - z) + (\rho\nu z - \lambda)x + \frac{\nu^2}{2}x^2.
\end{equation*}

By replacing the kernel $K$ with $K^N$ and defining $g^N$ and $\psi^N$ as in \eqref{eqn:DefinitionOfg} and \eqref{eqn:PsiRiccatiEquation}, respectively, we obtain a formula similar to \eqref{eqn:CharacteristicFunctionSemiClosedForm} for the characteristic function $\varphi^N$ of the log-stock price $\log(S_T^N)$.

Our ultimate goal will be to prove a bound in terms of $e_N$ for the weak error of approximating $\log(S_T)$ by $\log(S^N_T)$. To this end, we will first prove a local Lipschitz error bound for the characteristic function $\abs{\varphi(z) - \varphi^N(z)}$ in terms of $e_N$ in Section \ref{sec:CharacteristicFunctionErrorBound}. Then, in Section \ref{sec:WeakErrorBoundForTheLogStockPrice}, we will use this error bound on $\varphi$ to prove a weak error bound of $\log(S_T)$.

\subsection{Error bound for the characteristic function}\label{sec:CharacteristicFunctionErrorBound}

The goal of this subsection is to prove the following theorem under Assumption \ref{ass:TheAssumption}.

\begin{theorem}\label{thm:PhiNErrorBound}
For all $T\ge 0$ there exists a constant $C_0$ such that for all $N\ge 0$ and $z=a + bi$ with $a\in[0,1]$ and $b\in\R$ satisfying 
\begin{equation}\label{eqn:StupidCondition}
C_0(1 + b^6)e_N \le 1,
\end{equation}
we have $$\abs{\varphi(z) - \varphi^N(z)} \le 6C_0\varphi(a)(1 + b^6)e_N.$$
\end{theorem}

From here on, the values of the constant $C$ may change from line to line, but it is always independent of $N$. Furthermore, for a complex number $z\in \C$, we write $z = a + bi$ for $a,b\in\R$.

First, under Assumption \ref{ass:TheAssumption}, we have the following two propositions, similar to \cite[Theorem 4.1]{abi2019multifactor} and \cite[Proposition 5.4]{abi2019multifactor}. 
We remark that while \cite{abi2019multifactor} uses much more restrictive assumptions than Assumption \ref{ass:TheAssumption}, it is not difficult to verify that their proves still go through with only minor adaptations.

\begin{proposition}\label{prop:PsiNUniformConvergence}
There exists a constant $C$ such that for all $N\ge 0$ and $z=a+bi$ with $a\in[0,1]$ and $b\in\R$, $$\sup_{t\in[0,T]} \abs{\psi^N(t, z) - \psi(t, z)} \le C(1 + b^4) e_N.$$
\end{proposition}

\begin{proposition}\label{prop:PsiNUniformBound}
There exists a constant $C$ such that for all $N\ge 0$ and $z=a+bi$ with $a\in[0,1]$ and $b\in\R$, $$\sup_{N\in\N}\sup_{t\in[0,T]} \abs{\psi^N(t, z)} \le C(1 + b^2).$$
\end{proposition}

Proposition \ref{prop:PsiNUniformConvergence} already suggests that the Markovian approximations $\varphi^N$ might converge weakly at the same speed as $e_N$, as in Theorem \ref{thm:PhiNErrorBound}. To prove this theorem, we need some intermediary lemmas.

\begin{lemma}\label{lem:FirstPreliminaryLemma}
For all $T\ge 0$ there exists a constant $C_0$ such that for all $N\ge 0$ and $z=a+bi$ with $a\in[0,1]$ and $b\in\R$ satisfying $$C_0(1 + b^6)e_N \le 1,$$ we have
\begin{align*}
\bigg|1 - \exp\bigg(\int_0^T &\left(F(z, \psi^N(T-t, z)) - F(z, \psi(T-t, z))\right)g(t)\sdd t\bigg)\bigg| \le \frac{3}{2}C_0(1 + b^6)e_N.
\end{align*}
\end{lemma}

\begin{proof}
First, by Propositions \ref{prop:PsiNUniformConvergence} and \ref{prop:PsiNUniformBound}, there is a constant C such that for all $z=a + bi$ with $a\in[0, 1]$ and $b\in\R$,
\begin{align}
\abs{\psi^N(t, z)^2 - \psi(t, z)^2} &= \abs{\psi^N(t, z) + \psi(t, z)}\abs{\psi^N(t, z) - \psi(t, z)}\nonumber\\
&\le C(1 + b^2)\abs{\psi^N(t, z) - \psi(t, z)}\nonumber\\
&\le C(1 + b^6) e_N.\label{eqn:WeNeedThisInequality}
\end{align}

Since $F$ is a quadratic polynomial in the second variable, Proposition \ref{prop:PsiNUniformConvergence} and inequality \eqref{eqn:WeNeedThisInequality} imply that also $$\sup_{t\in[0,T]} \abs{F(z, \psi^N(t, z)) - F(z, \psi(t, z))} \le C(1 + b^6) e_N.$$

Denote by $f(z)$ the integral inside the exponential function. Since $g$ can be bounded on finite intervals, 
\begin{align}
\abs{f(z)} &\le C\int_0^T\abs{F(z, \psi(T-t, z)) - F(z, \psi^N(T-t, z))} \sdd t\nonumber\\
&\le C(1 + b^6)\int_0^T e_N \sdd t \le C(1 + b^6)e_N.\label{eqn:ThatConstantC}
\end{align}

We take $C_0 \coloneqq 2 C$, where $C$ is the final constant in \eqref{eqn:ThatConstantC}. Since $\abs{f(z)} \le 1/2$, 
\begin{align*}
\abs{1 - \exp\left(f(z)\right)} &\le \sum_{n=1}^\infty \frac{\abs{f(z)}^n}{n!} \le \frac{3}{2}\abs{f(z)} \le \frac{3}{2}C_0(1+b^6)e_N.\qedhere
\end{align*}
\end{proof}

\begin{lemma}\label{lem:SecondStupidLemma}
For all $T\ge 0$ there exists a constant $C_0$ such that for all $N\ge 0$ and $z=a+bi$ with $a\in[0,1]$ and $b\in\R$ satisfying $$C_0(1 + b^4)e_N \le 1,$$ we have
\begin{align*}
\abs{1 - \exp\left(\int_0^T F(z, \psi(T-t, z))\left(g^N(t) - g(t)\right)\sdd t\right)} \le \frac{3}{2}C_0(1 + b^4)e_N.
\end{align*}
\end{lemma}

\begin{proof}
    The proof of this lemma is similar to the proof of Lemma \ref{lem:FirstPreliminaryLemma}.
\end{proof}

We can now proceed with the proof of Theorem \ref{thm:PhiNErrorBound}.

\begin{proof}[Proof of Theorem \ref{thm:PhiNErrorBound}]
We have $$\abs{\varphi(z) - \varphi^N(z)} = \abs{\varphi(z)}\abs{1 - \frac{\varphi^N(z)}{\varphi(z)}} \le \varphi(a)\abs{1 - \frac{\varphi^N(z)}{\varphi(z)}},$$ where the last inequality holds since $\varphi$ is a characteristic function. 

Now, for the specific choice of $$a = \exp\left(\int_0^T \left(F(z, \psi^N(T-t,z)) - F(z, \psi(T-t,z))\right)g(t)\sdd t\right),$$ we use the inequality
\begin{align*}
\abs{1 - \frac{\varphi^N(z)}{\varphi(z)}} &\le \abs{1 - a} + \abs{a - \frac{\varphi^N(z)}{\varphi(z)}} = \abs{1 - a} + \abs{a}\abs{1 - \frac{1}{a}\frac{\varphi^N(z)}{\varphi(z)}}\\
&\le \abs{1 - a} + \left(1 + \abs{1 - a}\right)\abs{1 - \frac{1}{a}\frac{\varphi^N(z)}{\varphi(z)}}.
\end{align*}

We then use Lemma \ref{lem:FirstPreliminaryLemma} to bound $\abs{1 - a}$, and, using \eqref{eqn:CharacteristicFunctionSemiClosedForm}, we see that $\abs{1 - \frac{1}{a}\frac{\varphi^N(z)}{\varphi(z)}}$ can be bounded using Lemma \ref{lem:SecondStupidLemma}. Therefore,
\begin{align*}
    \abs{1 - \frac{\varphi^N(z)}{\varphi(z)}} &\le \frac{3}{2}C_0(1 + b^6)e_N + \left(1 + \frac{3}{2}C_0(1 + b^6)e_N\right)\frac{3}{2}C_0(1 + b^4)e_N\\
    &\le 6C_0(1 + b^6)e_N.\qedhere
\end{align*}

\end{proof}

\subsection{Weak error bound for the log-stock price}\label{sec:WeakErrorBoundForTheLogStockPrice}

We have proved that the error in the characteristic function $\varphi$ can be bounded using $e_N$. The next step will be to extend this result to the weak error in $X_T \coloneqq \log(S_T)$. To this end, let us introduce some notation related to the Fourier transform.

Given a function $f:\R\to\R$, we denote its generalized Fourier transform by $$\widehat{f}(z) \coloneqq \int_{\R} e^{izx} f(x)\sdd x$$ for all $z\in\C$ for which the above integral is well-defined. Furthermore, for $R\in\R$, we define the damped function $$f_R(x) \coloneqq e^{-Rx} f(x).$$ Define the sets
\begin{equation}\label{eqn:DefinitionSetsIJ}
\mathcal{I} \coloneqq \left\{R\in\R\colon f_R\in L^1_{\textup{bc}},\ \widehat{f_R}\in L^1\right\},\qquad \mathcal{J} \coloneqq \left\{R\in\R\colon \varphi(R) < \infty\right\},
\end{equation}
where $L^1_{\textup{bc}}$ denotes the set of $L^1$-functions that are bounded and continuous. We remark that we have $[0, 1]\subseteq \mathcal{J}.$ Then, for each $R\in \mathcal{I}\cap \mathcal{J}$, we have the Fourier pricing formula
\begin{equation}\label{eqn:FourierInversionPricing}
\E f(X_T) = \frac{1}{2\pi}\int_{\R} \varphi(R - i u) \widehat{f}(u + i R)\sdd u,
\end{equation}
see e.g. \cite[Theorem 2.2]{eberlein2010analysis}. Using this formula, we can prove the following theorem.

\begin{theorem}\label{thm:TheWeakErrorTheoremGeneralVersion}
Let $f:\R\to\R$ be a payoff function of the log-stock price, and let $a\in \mathcal{I}\cap [0, 1]\neq\emptyset$. Assume that there are $C_1,\delta > 0$ such that
\begin{equation}\label{eqn:FourierTransformDecayCondition}
\abs{\widehat{f}(b + ia)} \le C_1 (1 + \abs{b})^{-(1+\delta)}.
\end{equation}
Then there exist constants $\widetilde{C}, \eps > 0$ independent of $f,C_1,\delta$, such that if $e_N < \eps$ then, for $\delta \neq 6$, $$\abs{\E f(X_T) - \E f(X^N_T)} \le \widetilde{C}\left(\frac{1}{\abs{6-\delta}} + \frac{1}{\delta}\right)C_1e_N^{\frac{\delta}{6}\land 1}.$$
For $\delta = 6$, we get
$$\abs{\E f(X_T) - \E f(X^N_T)} \le \widetilde{C}C_1\log(e_N^{-1})e_N.$$
\end{theorem}

\begin{remark}
    We remark that the proof of Theorem \ref{thm:TheWeakErrorTheoremGeneralVersion} only needs Theorem \ref{thm:PhiNErrorBound}, and no further properties of rough Heston. In particular, let $X, X^N$ for $N\in\N$ be random variables, with characteristic functions $\varphi,\varphi^N$, such that $\varphi(a)<\infty$ for all $a\in[0,1]$. Then, for any sequence $e_N\to 0$ such that Theorem \ref{thm:PhiNErrorBound} is satisfied, Theorem \ref{thm:TheWeakErrorTheoremGeneralVersion} holds. Moreover, the interval $[0,1]$ here and in Theorems \ref{thm:PhiNErrorBound} and \ref{thm:TheWeakErrorTheoremGeneralVersion} may be replaced by any subinterval of $\mathcal{J}$ (where $\mathcal J$ is defined as in \eqref{eqn:DefinitionSetsIJ}).
\end{remark}

\begin{proof}[Proof of Theorem \ref{thm:TheWeakErrorTheoremGeneralVersion}]
By \eqref{eqn:FourierInversionPricing}, we have for all $L\ge 0$,
\begin{align*}
\abs{\E f(X_T) - \E f(X^N_T)} &= \frac{1}{2\pi}\abs{\int_{\R}\widehat{f}(b+ia)\left(\varphi(a-ib) - \varphi^N(a-ib)\right)\sdd b}\\
&\le \frac{1}{2\pi}\abs{\int_{-L}^L \widehat{f}(b+ia)\left(\varphi(a-ib) - \varphi^N(a-ib)\right)\sdd b}\\
&\qquad + \frac{1}{2\pi}\abs{\int_{\R\backslash[-L,L]}\widehat{f}(b+ia)\left(\varphi(a-ib) - \varphi^N(a-ib)\right)\sdd b}.
\end{align*}

We now bound these two summands separately. Considering the second summand first, we have
\begin{align*}
\abs{\int_{\R\backslash[-L,L]}\varphi(a - ib) \widehat{f}(b + ia)\sdd b} &\le \int_{\R\backslash[-L,L]}\abs{\varphi(a - ib)}\abs{\widehat{f}(b + ia)}\sdd b\\
&\le 2C_1\varphi(a)\int_L^\infty b^{-(1+\delta)} \sdd b = \frac{2C_1}{\delta}\varphi(a)L^{-\delta},
\end{align*}
where a similar statement is true for $\varphi^N$. Moreover, note that Theorem \ref{thm:PhiNErrorBound} implies that
\begin{align*}
\varphi^N(a) &\le \varphi(a) + \abs{\varphi(a) - \varphi^N(a)} \le \varphi(a) + 6C_0\varphi(a)e_N \le (6C_0+1)\varphi(a).
\end{align*}
Hence, the second summand can be bounded by
\begin{align*}
\Bigg|\int_{\R\backslash[-L,L]}\widehat{f}(b+ia) \left(\varphi(a-ib) - \varphi^N(a-ib)\right)\sdd b\Bigg| &\le \frac{2C_1}{\delta}(\varphi(a)+\varphi^N(a))L^{-\delta}\\
&= \frac{4C_1}{\delta}(3C_0+1)\varphi(a)L^{-\delta}.
\end{align*}

Next, we want to bound the first summand using Theorem \ref{thm:PhiNErrorBound}.  Let $C_0$ be as in Theorem \ref{thm:PhiNErrorBound}. and define $$L\coloneqq e_N^{-\frac{1}{6+\eta}}$$ for some $\eta > 0$. Note that we have
\begin{align*}
\lim_{N\to\infty} &\sup_{b\in[-L, L]}C_0(1 + b^6) e_N = \lim_{N\to\infty}C_0\left(1 + e_N^{-\frac{6}{6+\eta}}\right)e_N = 0.
\end{align*}
In particular, for large enough $N$, i.e. for $e_N$ sufficiently small, say
\begin{equation}\label{eqn:ENSmall}
e_N < \gamma^{\frac{6 + \eta}{\eta}}
\end{equation}
for some $\gamma \in (0, 1)$ independent of $\eta$, condition \eqref{eqn:StupidCondition} of Theorem \ref{thm:PhiNErrorBound} is satisfied for all $b\in[-L, L]$. 

Hence, we can apply Theorem \ref{thm:PhiNErrorBound} to bound the first summand, and we have 
\begin{align*}
    \Bigg|\int_{-L}^L \widehat{f}(b+ia) &\left(\varphi(a-ib) - \varphi^N(a-ib)\right)\sdd b\Bigg|\\
    &\le \int_{-L}^L C_1 (1 + \abs{b})^{-(1+\delta)}6 C_0(1 + b^6)\varphi(a)e_N \sdd b.
\end{align*}

Now, for $\delta \neq 6$,
\begin{align*}
    \int_{-L}^L (1 + \abs{b})^{-(1+\delta)}(1 + b^6)\sdd b &\le 2\int_0^L (1 + b)^{5-\delta} \sdd b = \frac{2}{\abs{6 - \delta}}\abs{1 - (1 + L)^{6-\delta}}\\
    &\le \frac{2(1 + L)^{(6-\delta)\lor 0}}{\abs{6 - \delta}} \le \frac{2^7L^{(6-\delta)\lor 0}}{\abs{6 - \delta}} = \frac{2^7}{\abs{6 - \delta}}e_N^{\frac{\delta-6}{6 + \eta}\land 0}.
\end{align*}

Similarly, for $\delta = 6$,
\begin{align*}
    \int_{-L}^L (1 + \abs{b})^{-(1+\delta)}(1 + b^6)\sdd b &\le 2\int_0^L (1 + b)^{-1} \sdd b = 2\log(1 + L)\\
    &= 2\log\left(1 + e^{1-\frac{1}{6 + \eta}}\right) \le \frac{2}{\eta}e_N^{-\frac{\eta}{6 + \eta}}.
\end{align*}

Altogether, we get for $e_N < \gamma^{\frac{6 + \eta}{\eta}}$ (and $\delta \neq 6$)
\begin{align*}
\abs{\E f(X_T) - \E f(X^N_T)} &\le \frac{1}{2\pi}\frac{768}{\abs{6-\delta}}C_0C_1\varphi(a)e_N^{\frac{\delta + \eta}{6+\eta}\land 1} + \frac{1}{2\pi}\frac{4C_1}{\delta}(3C_0 + 1)\varphi(a)e_N^{\frac{\delta}{6+\eta}}\\
&\le \widetilde{C}\left(\frac{1}{\abs{6-\delta}} + \frac{1}{\delta}\right)C_1e_N^{\frac{\delta}{6+\eta}\land 1},
\end{align*}
with $\widetilde{C} = \varphi(a)(123C_0\lor 1),$ and similarly, for $\delta = 6$,
$$\abs{\E f(X_T) - \E f(X^N_T)} \le \widetilde{C}\left(1 + \frac{1}{\eta}\right)C_1e_N^{\frac{6}{6+\eta}}.$$

This is almost the statement of the theorem, except that we need to get rid of the parameter $\eta$. We do this by choosing an appropriate value of $\eta$, where we restrict ourselves to $\eta\in(0, 1)$. We may then sharpen requirement \eqref{eqn:ENSmall} to $$e_N \le \gamma^{7/\eta}, \text{ and thus choose } \eta = \frac{7\log\gamma}{\log e_N}.$$

We remark that $\eta\in(0, 1)$ for $e_N$ sufficiently small (i.e. $e_N < \eps$ for some $\eps > 0$). Next, note that for $\delta < 6 + \eta$, we have 
\begin{align*}
    \frac{\delta}{6 + \eta} &= \frac{\delta}{6} - \frac{\eta\delta}{6(6+\eta)} \ge \frac{\delta}{6} - \frac{\eta}{6} = \frac{\delta}{6} - \frac{7\log \gamma}{6\log e_N},
\end{align*}
and in particular (if additionally $\delta \neq 6$),
\begin{align*}
    \abs{\E f(X_T) - \E f(X^N_T)} &\le \widetilde{C}\left(\frac{1}{\abs{6-\delta}} + \frac{1}{\delta}\right)C_1\exp\left(\left(\frac{\delta}{6} - \frac{7\log \gamma}{6\log e_N}\right)\log e_N\right)\\
    &= \widetilde{C}\left(\frac{1}{\abs{6-\delta}} + \frac{1}{\delta}\right)C_1e^{-\frac{7\log \gamma}{6}}e_N^{\delta/6}.
\end{align*}
Now, the requirements $\delta < 6 + \eta$ and $\delta \neq 6$ are certainly satisfied for $\delta < 6$, proving this error bound in this regime.

Next, consider $\delta = 6$. In this case we still have $\delta < 6 + \eta$, and we get
\begin{align*}
    \abs{\E f(X_T) - \E f(X^N_T)} &\le \widetilde{C}\left(1 + \frac{\log e_N}{7\log \gamma}\right)C_1\exp\left(-\frac{7\log\gamma}{6\log e_N}\log e_N\right)e_N\\
    &\le \widehat{C} C_1 \log (e_N^{-1}) e_N.
\end{align*}

For $\delta \ge 7$, we may choose $\eta = 1$, finishing the proof in this case.

Finally, consider $\delta\in(6, 7)$. Clearly, if the decay assumption on $\widehat{f}$ holds true for some $\delta > 6$, it is also true for $\delta = 6$, implying that for $\delta\in(6, 7)$ we also have the same bound as for $\delta = 6$. On the other hand, we may choose $\eta = \delta - 6$, which is admissible for $N$ such that $e_N^{\delta - 6} \le \gamma^7$. Thus, we get a bound of the form
\begin{align*}
    \abs{\E f(X_T) - \E f(X^N_T)} &\le 
    \begin{cases}
        \widehat{C}C_1 \frac{1}{\delta - 6} e_N,\qquad &e_N^{\delta - 6} \le \gamma^7\\
        \widehat{C}C_1\log(e_N^{-1})e_N,\qquad &e_N^{\delta - 6} > \gamma^7.
    \end{cases}\\
    &\le 
    \begin{cases}
        \widehat{C}C_1 \frac{1}{\delta - 6} e_N,\qquad &e_N^{\delta - 6} \le \gamma^7\\
        \widehat{C}C_1\log\left(\gamma^{-\frac{7}{\delta - 6}}\right)e_N,\qquad &e_N^{\delta - 6} > \gamma^7.
    \end{cases}\\
    &\le \frac{\widehat{C}C_1}{\delta - 6}e_N.
\end{align*}
This proves the corollary.
\end{proof}

Essentially, Theorem \ref{thm:TheWeakErrorTheoremGeneralVersion} is already what we wanted to show, in that it demonstrates that the weak error in the log-stock price $\log(S_T)$ can be bounded by $e_N$. From here on, we merely want to polish this result. Mainly, we want to formulate the decay condition \eqref{eqn:FourierTransformDecayCondition} in terms of the payoff function $f$ itself (Corollary \ref{cor:TheGeneralTheoremInTermsOfDifferentiableFunctions}), weaken or generalize the decay condition (Corollary \ref{cor:FinalCorollaryInTermsOfLogStock}), and to formulate the result in terms of payoffs of the stock price $S_T$, rather than the log-stock price $\log(S_T)$ (Corollary \ref{cor:FinalCorollaryInTermsOfStock}). Finally, in Corollaries \ref{cor:FinalCorollaryBoundedPayoff}, \ref{cor:FinalCorollaryProse} and \ref{cor:FinalCorollaryLipschitz}, we give more transparent special cases.

First, we formulate the decay condition \eqref{eqn:FourierTransformDecayCondition} on the Fourier transform $\widehat{f}$ in terms of the payoff function $f$ itself.

\begin{corollary}\label{cor:TheGeneralTheoremInTermsOfDifferentiableFunctions}
    Let $f:\R\to\R$ be a payoff function of the log-stock price, and let $n\ge 2$ be an integer. Assume that $f$ is $n$ times weakly differentiable, and that there is $a\in[0, 1]$ and $r\in L^1$ such that $$\abs{f(x)} \le e^{ax}r(x),\qquad \abs{f^{(n)}(x)} \le e^{ax}r(x).$$ Then, there exist constants $\widetilde{C},\eps>0$ independent of $f$ and $r$, such that if $e_N < \eps,$ then, for $n\neq 7$, $$\abs{\E f(X_T) - \E f(X^N_T)} \le \widetilde{C} \|r\|_1 e_N^{\frac{n-1}{6}\land 1}.$$ For $n = 7$, we get $$\abs{\E f(X_T) - \E f(X^N_T)} \le \widetilde{C} \|r\|_1 \log(e_N^{-1})e_N.$$
\end{corollary}

\begin{proof}
    Observe that we have 
    \begin{align*}
        \abs{\widehat{f^{(n)}}(b + ai)} &= \abs{\int_{\R} e^{i(b+ai)x} f^{(n)}(x) \sdd x} \le \int_{\R} e^{-ax} \abs{f^{(n)}(x)} \sdd x \le \int_{\R}r(x) \sdd x = \|r\|_1.
    \end{align*}
    In particular, since $\abs{b^n} \le \abs{(b+ai)^n},$ we have $$\abs{b^n \widehat{f}(b + ai)} \le \abs{(b + ai)^n \widehat{f}(b + ai)} = \abs{\widehat{f^{(n)}}(b + ai)} \le \|r\|_1.$$ Moreover, in the same spirit as above, we have $$\abs{\widehat{f}(b + ai)} = \abs{\int_{\R}e^{i(b + ai)x} f(x) \sdd x} \le \|r\|_1.$$ This shows that $$\abs{\widehat{f}(b + ai)} \le \|r\|_1\left(1\land \abs{b}^{-n}\right) \le 2^n\|r\|_1 (1 + \abs{b})^{-n}.$$ 
    
    The corollary follows by an easy application of Theorem \ref{thm:TheWeakErrorTheoremGeneralVersion} once we show that $a\in\mathcal{I}$, i.e. $f_a\in L^1_{\textup{bc}}$ and $\widehat{f_a}\in L^1$. First, $f_a\in L^1_{\textup{bc}}$ follows immediately from the assumptions. For the second condition, we merely note that $\widehat{f_a}(b) = \widehat{f}(b + ai).$
\end{proof}

In Corollary \ref{cor:TheGeneralTheoremInTermsOfDifferentiableFunctions}, the payoff $f$ needs to satisfy a very specific growth condition of the form $\abs{f(x)} \le e^{ax}r(x)$ for some $r\in L^1$. Since $f$ is a function of the log-stock price $X_T \coloneqq \log(S_T)$, this translates roughly to a polynomial growth condition of the form $\abs{h(s)} \le s^a r(\log s)$, where $h(s) = f(\log s)$ is the corresponding payoff function of the stock price $S_T$. In general, this is a reasonable requirement, since $M_q\coloneqq \E S_T^q$ may not exist for $q\in\R\backslash[0, 1]$. This non-existence of the moments of the stock price is also referred to as moment-explosion, and has been studied e.g. in \cite{gerhold2019moment} for the rough Heston model. However, depending on the parameters of the rough Heston model, the moments $M_q$ may exist for all $T>0$. In Lemma \ref{lem:MomentsExist} in the appendix, we give some sufficient conditions for the existence of moments $M_q$ for some $q < 0$ or $q > 1$. 

In the following corollary, we want to extend Corollary \ref{cor:TheGeneralTheoremInTermsOfDifferentiableFunctions} to allow also other growth conditions on the payoff function, depending on the existence of the moments $M_q$. The statement below with all its parameters may look quite formidable, but all it says is that if the payoff function $f$ grows slightly slower than would be permitted for the existence of $\E f(X_T)$, and if the derivatives of $f$ up to some order $n$ grow at most exponentially fast (for $x\to\pm\infty$), then we can give a rate of convergence. More precisely, the parameters $q_1$ and $-q_2$ should be interpreted as the largest and smallest moments of $S_T$ that exist, $\gamma_1$ and $\gamma_2$ denote how much slower $f$ grows than permitted by $q_1$ and $q_2$, and $p_1$ and $p_2$ determine how fast the derivatives of $f$ grow. Furthermore, we apply Corollary \ref{cor:TheGeneralTheoremInTermsOfDifferentiableFunctions} with the parameter $a \coloneqq (\widetilde{a} \land 1) \lor 0$. Since this Corollary is rather difficult to comprehend, we give some much simpler special cases in Corollaries \ref{cor:FinalCorollaryBoundedPayoff}, \ref{cor:FinalCorollaryProse} and \ref{cor:FinalCorollaryLipschitz} below.

\begin{corollary}\label{cor:FinalCorollaryInTermsOfLogStock}
Let $q_1 \ge 1$, let $q_2,p_1,p_2 \ge 0$, and let $\gamma_1, \gamma_2 > 0$. Let $f:\R\to\R$ be a payoff function of the log-stock price, and let $n\ge 2$ be an integer. Assume that $f$ is $n$ times weakly differentiable, and that there is $C_1 > 0$, $r\in L^1$ taking values in $(0, 1]$ such that $r^{-1}$ is locally bounded, and $\beta\in[0, 1]$ such that for all $k=1,\dots,n$,
\begin{align}
    \abs{f(x)} &\le 
    \begin{cases}
        C_1 e^{-(q_2 - \gamma_2)x}r(x)^\beta,\qquad &x\le 0,\\
        C_1 e^{(q_1 - \gamma_1)x}r(x)^\beta,\qquad &x\ge 0,
    \end{cases}\label{eqn:BoundOnf}\\
    \abs{f^{(k)}(x)} &\le 
    \begin{cases}
        C_1 e^{-(p_2 + q_2 - \gamma_2)x}r(x)^\beta,\qquad &x\le 0,\\
        C_1 e^{(p_1 + q_1 - \gamma_1)x}r(x)^\beta,\qquad &x\ge 0.
    \end{cases}\label{eqn:BoundOnDerivativesOff}
\end{align}
Then, there exist constants $\widetilde{C},\eps>0$ independent of $f, C_1,r,q_1,q_2,p_1,p_2,\gamma_1,\gamma_2,\beta$, such that if $e_N < \eps,$ then, for $n\neq 7$,
\begin{equation}\label{eqn:TheConvergenceRate}
\abs{\E f(X_T) - \E f(X^N_T)} \le \widetilde{C} C_1\left(e^{p_1+p_2+q_1+q_2}\widetilde{r}(e_N)^{\beta-1} + M_{-q_2} + M_{q_1}\right)e_N^\alpha,
\end{equation}
where
\begin{align*}
\alpha &\coloneqq 
\begin{cases}
\left(\frac{\gamma_1 + \gamma_2}{p_1+p_2+q_1+q_2}\land 1\right)\left(\frac{n-1}{6}\land 1\right),\quad &\widetilde{a}\in[0, 1],\\
\left(\frac{\gamma_2}{p_2+q_2}\land 1\right)\left(\frac{n-1}{6}\land 1\right),\quad &\widetilde{a} < 0,\\
\left(\frac{\gamma_1}{p_1 + q_1 - 1}\land 1\right)\left(\frac{n-1}{6}\land 1\right)\quad &\widetilde{a} > 1,
\end{cases},\quad \widetilde{a} \coloneqq \frac{(p_1+q_1)\gamma_2 - (p_2+q_2)\gamma_1}{\gamma_1+\gamma_2},
\end{align*}
and $$\widetilde{r}(e_N) \coloneqq \inf\left\{r(x)\colon x\in\left[-\frac{\alpha}{\gamma_2} \log e_N^{-1}-1,\frac{\alpha}{\gamma_1} \log e_N^{-1}+1\right]\right\}.$$ The same result holds true for $n=7$ if we multiply the upper bound \eqref{eqn:TheConvergenceRate} by $\log(e_N^{-1}).$
\end{corollary}

\begin{proof}
    The proof is rather technical, so we only give the general idea. Fix some parameters $b_1, b_2 > 0$. Let $g:\R\to\R$ be a function that is $n$ times weakly differentiable such that for $k=1,\dots,n$
    \begin{align}
    g(x) &= 
    \begin{cases}
        f(x),\qquad &x\in [-b_2, b_1],\\
        0,\qquad &x\in(-\infty, -b_2-1]\cup[b_1+1,\infty),
    \end{cases}\nonumber\\
    \abs{g(x) - f(x)} &\le 
    \begin{cases}
        C_1 e^{-(q_2-\gamma_2)x}r(x)^\beta,\qquad &x\in(-\infty, -b_2],\\
        C_1 e^{(q_1-\gamma_1)x}r(x)^\beta,\qquad &x\in[b_1, \infty),
    \end{cases}\label{eqn:ErrorOfThisApproximation}\\
    \abs{g^{(k)}(x)} &\le 
    \begin{cases}
        C_1 c_k e^{-(p_2 + q_2-\gamma_2)x}r(x)^\beta,\qquad &x\in[-b_2-1, -b_2],\\
        C_1 c_k e^{(p_1 + q_1-\gamma_1)x}r(x)^\beta,\qquad &x\in[b_1, b_1+1].
    \end{cases}\nonumber
    \end{align}

    The existence of such a function $g$ is ensured by Corollary \ref{cor:MakeFunctionCompact}, and it satisfies the conditions of Corollary \ref{cor:TheGeneralTheoremInTermsOfDifferentiableFunctions}. Using 
    \begin{align}
        \abs{\E f(X_T) - \E f(X^N_T)} &\le \abs{\E g(X_T) - \E g(X^N_T)}\label{eqn:FirstTerm}\\
        &\qquad + \abs{\E f(X_T) - \E g(X_T)} + \abs{\E f(X^N_T) - \E g(X^N_T)},\label{eqn:SecondTerms}
    \end{align}
    we estimate \eqref{eqn:FirstTerm} by Corollary \ref{cor:TheGeneralTheoremInTermsOfDifferentiableFunctions}, and \eqref{eqn:SecondTerms} by \eqref{eqn:ErrorOfThisApproximation} together with $M_{q_1}$ and $M_{-q_2}$. Finally, the choice $$b_1 = -\frac{\alpha}{\gamma_1}\log e_N,\qquad b_2 = -\frac{\alpha}{\gamma_2}\log e_N$$ yields the desired result.
\end{proof}

Finally, we formulate Corollary \ref{cor:FinalCorollaryInTermsOfLogStock} in terms of payoff functions of the stock price $S_T$.

\begin{corollary}\label{cor:FinalCorollaryInTermsOfStock}
Let $q_1 \ge 1$, let $q_2,p_1,p_2 \ge 0$, and let $\gamma_1, \gamma_2 > 0$. Let $h:\R_+\to\R$ be a payoff function of the stock price, and let $n\ge 2$ be an integer. Assume that $h$ is $n$ times weakly differentiable, and that there is $C_1 > 0$, $r:\R_+\to\R$ taking values in $(0, 1]$ such that $r^{-1}$ is locally bounded and $\int_{-\infty}^\infty r(e^x) dx < \infty$, and $\beta\in[0, 1]$ such that for all $k=1,\dots,n$,
\begin{align*}
    \abs{h(x)} &\le 
    \begin{cases}
        C_1 x^{-(q_2 - \gamma_2)}r(x)^\beta,\qquad &x\le 1,\\
        C_1 x^{q_1 - \gamma_1}r(x)^\beta,\qquad &x\ge 1,
    \end{cases}\\
    \abs{h^{(k)}(x)} &\le 
    \begin{cases}
        C_1 x^{-(p_2 + q_2 - \gamma_2 + k)}r(x)^\beta,\qquad &x\le 1,\\
        C_1 x^{p_1 + q_1 - \gamma_1 - k}r(x)^\beta,\qquad &x\ge 1.
    \end{cases}
\end{align*}
Then, there exist constants $\widetilde{C},\eps>0$ independent of $h, C_1,r,q_1,q_2,p_1,p_2,\gamma_1,\gamma_2,\beta$ ($\eps$ also independent of $n$), such that if $e_N < \eps,$ then, for $n\neq 7$, $$\abs{\E h(S_T) - \E h(S^N_T)} \le \widetilde{C} C_1\left(e^{p_1+p_2+q_1+q_2}\widetilde{r}(e_N)^{\beta-1} + M_{-q_2} + M_{q_1}\right)e_N^\alpha,$$ where $\alpha$ is as in Corollary \ref{cor:FinalCorollaryInTermsOfLogStock}, and $$\widetilde{r}(e_N) \coloneqq \inf\left\{r(x)\colon x\in\left[e^{-1}e_N^{\frac{\alpha}{\gamma_2}},ee_N^{-\frac{\alpha}{\gamma_1}}\right]\right\}.$$ The same results hold true for $n=7$ if we multiply the upper bounds by $\log(e_N^{-1}).$
\end{corollary}

\begin{proof}
    Define the function $f:\R\to\R$ by $f(x) \coloneqq h(e^x)$. Then, we only have to verify equations \eqref{eqn:BoundOnf} and \eqref{eqn:BoundOnDerivativesOff}. However, the bound \eqref{eqn:BoundOnf} for $f$ itself is trivial, and \eqref{eqn:BoundOnDerivativesOff} follows immediately from the Faà di Bruno formula.
\end{proof}

We do not claim that the rate of convergence $e_N^\alpha$ in Corollary \ref{cor:FinalCorollaryInTermsOfStock} is necessarily best possible in general. However, the assumptions imposed on the payoff function are quite natural. They essentially amount to requiring that $h$ grows slightly slower than the largest moment of $S_T$ that exists. If this is the case we can give a specific convergence rate in $e_N$.

Since Corollary \ref{cor:FinalCorollaryInTermsOfStock} can be quite intimidating, we now give a few special cases.

\begin{corollary}\label{cor:FinalCorollaryBoundedPayoff}
Let $h:\R_+\to\R$ be 8 times weakly differentiable and compactly supported. Then there is $C>0$ such that $$\abs{\E h(S_T) - \E h(S^N_T)} \le Ce_N.$$
\end{corollary}

\begin{proof}
    Choose $\beta = 1$, $r(x) = (1 + \abs{\log x})^{-2}$, $p_1=p_2=0$, $q_2 > 0$ such that $M_{-q_2} < \infty$ (by Lemma \ref{lem:MomentsExist}), $q_1 = 1$, $\gamma_1 = q_1$, and $\gamma_2 = q_2$, so that $\widetilde{a} = 0$. We then get constants $C, \eps > 0$ such that if $e_N < \eps$, then $$\abs{\E h(S_T) - \E h(S^N_T)} \le Ce_N.$$ We can drop the restriction $e_N < \eps$ by noting that $h$ is uniformly bounded, and hence for $e_N \ge \eps$,
    \begin{align*}
    \abs{\E h(S_T) - \E h(S^N_T)} &\le 2\|h\|_\infty \le \frac{2\|h\|_\infty}{\eps}e_N.\qedhere
    \end{align*}
\end{proof}

\begin{corollary}\label{cor:FinalCorollaryProse}
    Let $h:\R_+\to\R$ be a payoff function of the stock price that is twice weakly differentiable. Let $q_1\ge 1$ and $q_2\ge 0$ be such that $\E S_T^{q_1} + \E S_T^{-q_2} < \infty$, and assume that $h(x) = O(x^{q_1-\eps})$ for $x\to\infty$ and $h(x) = O(x^{-q_2+\eps})$ for $x\to 0$. Assume furthermore that the first two derivatives of $h$ grow only polynomially in $x$ for $x\to 0$ and $x\to\infty$. Then there exist some $C,\alpha > 0$ such that $$\abs{\E h(S_T) - \E h(S^N_T)} \le C e_N^\alpha.$$
\end{corollary}

\begin{proof}
    This is really just a non-quantitative reformulation of Corollary \ref{cor:FinalCorollaryInTermsOfStock}.
\end{proof}

The following corollary has convenient assumptions that are usually satisfied in practice. It merely requires (global) Lipschitz continuity of the payoff, and the existence of a moment $M_q \coloneqq \E S_T^q$ with $q>1$. We remark that the convergence rate $\frac{q-1}{12q}$ is of course very bad, and certainly not optimal. However, the purpose of this corollary is to demonstrate that the weak error can be bounded using the $L^1$-error in the kernel, not just the $L^2$-error.

\begin{corollary}\label{cor:FinalCorollaryLipschitz}
    Let $h:\R_+\to\R$ be a Lipschitz continuous payoff function of the stock price, and let $q\in(1, 2]$. Then, there exist constants $C, \eps>0$ independent of $h$ and $q$, such that if $e_N < \eps$, then 
    \begin{align*}
        \abs{\E h(S_T) - \E h(S^N_T)} &\le C \|h\|_{\lip}\left((q-1)^{-1} + M_q^{1/2}\right)\log(e_N^{-1})e_N^{\frac{q-1}{12q}}.
    \end{align*}
\end{corollary}

\begin{remark}
    The condition $e_N < \eps$ can be removed by a localization argument similar to Corollary \ref{cor:FinalCorollaryBoundedPayoff} using Lipschitz continuous approximations with compact support, at the cost of a worse convergence rate.
\end{remark}

\begin{proof}
    Since $h$ is Lipschitz continuous, the limit $h(0)\coloneqq \lim_{x\downarrow 0} h(x)$ certainly exists. Furthermore, by replacing $h$ with $\widetilde{h} = h - h(0)$, we may assume without loss of generality that $h(0) = 0$. Denote by $L$ the Lipschitz constant of $h$.

    Now, assume for a moment that $h$ is additionally twice weakly differentiable with bounded derivatives. Let $\|h\|_{C_2}$ be the supremum norm on $h'$ and $h''$. In Corollary \ref{cor:FinalCorollaryInTermsOfStock}, set $\beta = 0$, $r=(1 + \abs{\log(x)})^{-2}$, $q_2=p_2=0$, $\gamma_2=1$, $q_1 \in(1, 2]$, $\gamma_1 = q_1 - 1$, $p_1=1$, $C_1=\|h\|_{C_2}$. Then, there exist constants $\widetilde{C},\eps>0$ independent of $h,q_1$, such that if $e_N < \eps,$ then, $$\abs{\E h(S_T) - \E h(S^N_T)} \le \widetilde{C} \|h\|_{C_2}\left((q_1-1)^{-2} + M_{q_1}\right)\log(e_N^{-1})^2e_N^{\frac{q_1-1}{6q_1}}.$$

    Next, we define a sequence of twice weakly differentiable functions $h_n$ approximating $h$. The functions $h_n$ are defined via their second weak derivative $h_n''$, together with the initial conditions $h_n(0) = h_n'(0) = 0$. The function $h_n''$ is given by $$h_n''(x) \coloneqq \begin{cases}
        4\left(h\left(\frac{k+1}{n}\right) - h\left(\frac{k}{n}\right)\right)n^2,\qquad &x\in\left[\frac{k}{n},\frac{k + 1/2}{n}\right],\\
        -4\left(h\left(\frac{k+1}{n}\right) - h\left(\frac{k}{n}\right)\right)n^2,\qquad &x\in\left[\frac{k+1/2}{n},\frac{k + 1}{n}\right].
    \end{cases}
    $$
    It is easily verified that $$h_n'(k/n) = 0,\quad h_n(k/n) = h(k/n),\quad \|h_n\|_{C_2} \le 4\|h\|_{\lip}n,\quad\textup{and}\quad \|h-h_n\|_\infty \le \frac{\|h\|_{\lip}}{n},$$ where $\|h\|_{\lip}$ is the Lipschitz-norm of $h$.

    Therefore, we get
    \begin{align*}
        \abs{\E h(S_T) - \E h(S^N_T)} &\le \abs{\E h(S_T) - \E h_n(S_T)} + \abs{\E h_n(S_T) - \E h_n(S^N_T)}\\
        &\qquad + \abs{\E h_n(S^N_T) - \E h(S^N_T)}\\
        &\le 2\frac{\|h\|_{\lip}}{n} + \widetilde{C} \|h\|_{\lip}n\left((q_1-1)^{-2} + M_{q_1}\right)\log(e_N^{-1})^2e_N^{\frac{q_1-1}{6q_1}}.
    \end{align*}

    Setting $$n = \left(\left((q_1-1)^{-2} + M_{q_1}\right)\log(e_N^{-1})^2e_N^{\frac{q_1-1}{6q_1}}\right)^{-1/2},$$ we get
    \begin{align*}
        \abs{\E h(S_T) - \E h(S^N_T)} &\le \widetilde{C} \|h\|_{\lip}\left((q_1-1)^{-1} + M_{q_1}^{1/2}\right)\log(e_N^{-1})e_N^{\frac{q_1-1}{12q_1}}.\qedhere
    \end{align*}
\end{proof}

\section{$L^1$-approximation of the fractional kernel}\label{sec:ErrorRates}

The aim of this section is to give kernels $K^N$ of the form \eqref{eqn:KNForm}, such that $$e_N\coloneqq \int_0^T \abs{K(t) - K^N(t)} \sdd t$$ converges quickly for $N\to\infty.$ We give two different approximations in Sections \ref{sec:GeometricGaussianApproximations} and \ref{sec:NonGeometricGaussianApproximations}, respectively. Throughout, we will make heavy use of the representation of $K$ in terms of its inverse Laplace transform, i.e., 
\begin{equation}\label{eqn:LaplaceRepresentationOfK}
K(t) = \frac{t^{H-1/2}}{\Gamma(H+1/2)} = c_H \int_0^\infty e^{-xt} x^{-H-1/2} \sdd x,\quad c_H = \frac{1}{\Gamma(H+1/2)\Gamma(1/2-H)}.
\end{equation}

\subsection{Geometric Gaussian approximations}\label{sec:GeometricGaussianApproximations}

Let us now define the approximations $K^N$ we use, which deviate slightly from the approximations in \cite{bayer2023markovian}. We denote by $\mathrm{rd}:\R\to\mathbb{N}_+$ the function rounding to the nearest positive integer.

\begin{definition}[Geometric Gaussian Rules]\label{def:GeometricGaussianRules}
Let $N\in\N$ be the total number of nodes and $\alpha,\beta,a,b\in(0,\infty)$ be parameters of the scheme. Define
\begin{equation}\label{eqn:MNDefinitionGGR}
m\coloneqq \mathrm{rd}\left(\beta\sqrt{(1/2+H)N}\right),\qquad n\coloneqq \mathrm{rd}\left(\frac{1}{\beta}\sqrt{\frac{N}{1/2+H}}\right) (\approx N/m),
\end{equation}
$$\xi_0\coloneqq 0,\quad \xi_n\coloneqq b\exp\left(\frac{\alpha}{\sqrt{1/2+H}}\sqrt{N}\right),\quad \xi_i=a\left(\frac{\xi_n}{a}\right)^{\frac{i}{n}},\qquad i=1,\dots,n.$$ Let $(x_j)_{j=1}^m$ be the nodes and $(\widetilde{w}_j)_{j=1}^m$ be the weights of a Gaussian quadrature formula of level $m$ on the interval $[0, \xi_1]$ with the weight function $w(x) = c_H x^{-H-1/2}$. Furthermore, let $(x_j)_{j=im+1}^{(i+1)m}$ be the nodes and $(\widetilde{w}_j)_{j=im+1}^{(i+1)m}$ be the weights of a Gaussian quadrature formula of level $m$ on to the intervals $[\xi_i, \xi_{i+1}]$ for $i=1,\dots,n-1$ with the weight function $w(x) \equiv 1$. Then we define the geometric Gaussian rule of type $(H,N,\alpha,\beta,a,b)$ to be the set of nodes $(x_j)_{j=1}^{mn}$ with weights $(w_j)_{j=1}^{mn}$ defined by $$w_j \coloneqq \begin{cases}
\widetilde{w}_j,\qquad &\textup{if } j=1,\dots,m,\\
c_H\widetilde{w}_j x_j^{-H-1/2},\qquad &\textup{else}.
\end{cases}$$ The approximation of $K$ is then given by $$K^N(t)\coloneqq \sum_{j=1}^{mn}w_je^{-x_jt}.$$ In what follows we will often drop the function $\mathrm{rd}$ in \eqref{eqn:MNDefinitionGGR} and assume that $m$ and $n$ can be real numbers, and that $N = nm$ exactly, purely for convenience. This does not affect the convergence rates.
\end{definition}

\begin{remark}\label{rem:IntuitionsOfGeometricGaussianRuleParameters}
In geometric Gaussian rules, we have four parameters that we can freely choose: $\alpha,\beta,a,b$. These parameters can be interpreted as follows:
\begin{itemize}
    \item The parameter $\alpha$ determines the cutoff point $\xi_n$ of the integral in \eqref{eqn:LaplaceRepresentationOfK}, i.e. we approximate the integral only on the interval $[0, \xi_n]$.
    \item The parameter $\beta$ determines the level of the Gaussian quadrature rule.
    \item The parameter $a$ determines the size of the first interval $[0, \xi_1]$. This interval is special due to the singularity in the weight function $w$.
    \item The parameter $b$ is used for fine-tuning the results.
\end{itemize}
In particular, parameter $\alpha$ is mainly responsible for controlling the error in our approximation on the interval $[\xi_n, \infty)$, parameter $\beta$ is mainly responsible for controlling the error on $[\xi_1, \xi_n]$, and parameter $a$ is mainly responsible for controlling the error on $[0,\xi_1]$.
\end{remark}

Throughout this section, we will use the following proposition. It states that Gaussian quadrature is lower-biased for completely monotone functions.

\begin{proposition}\label{prop:GaussianQuadratureUnderestimatesCompletelyMonotoneFunctions}
Let $f:[a,b]\to\R$ be completely monotone, and let $(x_i)_{i=1}^n$ be the nodes and $(w_i)_{i=1}^m$ be the weights of Gaussian quadrature with respect to the weight function $w$. Then,
$$\int_a^b f(x) w(x) \sdd x - \sum_{i=1}^n w_i f(x_i) \ge 0.$$
\end{proposition}

\begin{proof}
This follows immediately from Theorem \ref{thm:GaussGeneralErrorRepresentationFormula}, using that derivatives of even order of completely monotone functions are non-negative.
\end{proof}

We now start with the proof of the convergence rate, by stating a lemma on the error of Gaussian quadrature. This is a sharpened version of \cite[Theorem 19.3]{trefethen2019approximation} for the specific function $f(x) = (x+c)^{-\gamma}$. The proof is very technical, and delegated to Appendix \ref{sec:ProofOfImprovedGaussianErrorForSpecificPowerFunction}.

\begin{lemma}\label{lem:ImprovedGaussianErrorForSpecificPowerFunction}
Let $f:[-1, 1]\to\R,$ $f(x) \coloneqq (x+c)^{-\gamma},$ where $c>1$, and $\gamma > 1$. Let $(x_i)_{i=1}^m$ be the nodes and $(w_i)_{i=1}^m$ be the weights of Gaussian quadrature of level $m$ on $[-1, 1]$ with weight $w(x)\equiv 1$. Define $\mu\coloneqq \frac{2m}{\gamma-1}$, assume that $\mu \ge \frac{1}{c-1}\lor \frac{3c}{2\sqrt{c^2-1}},$ set $\eps \coloneqq \frac{\sqrt{\mu^2(c^2-1)+1} - c}{\mu^2-1},$ and $r \coloneqq c-\eps + \sqrt{(c-\eps)^2-1}.$ Then,
\begin{align*}
\abs{\int_{-1}^1 f(x) \sdd x - \sum_{i=1}^m w_i f(x_i)} &\le \nu \eps^{1-\gamma}\left(c + \sqrt{c^2-1}\right)^{-2m},
\end{align*}
where $$\nu \coloneqq \nu_{m, r, \gamma, c} \coloneqq \frac{8}{\pi}\frac{4m^2}{4m^2-1}\frac{r}{(r - r^{-1})^2}\Bigg(1 + \frac{(\pi/2)^\gamma}{\gamma - 1}\Bigg)e^{\gamma-1}.$$
\end{lemma}

Recall the representation \eqref{eqn:LaplaceRepresentationOfK} of $K$. In the following Lemmas, we split up the error of Gaussian quadrature on the interval $[0, \infty)$ into several smaller intervals that we treat separately. More precisely, in Lemma \ref{lem:SingleIntervalL1Bound}, we consider intervals $[\xi_i,\xi_{i+1}]$ with $i=1,\dots,n-1$, in Lemma \ref{lem:IntermediateMeanReversion} the interval $[\xi_1,\xi_n]$, in Lemma \ref{lem:LowMeanReversion} the interval $[0, \xi_1]$, and in Lemma \ref{lem:HighMeanReversion} the interval $[\xi_n, \infty)$. Finally, we will combine all these error bounds in Theorem \ref{thm:TheL1TheoremGeometric}.

\begin{lemma}\label{lem:SingleIntervalL1Bound}
Let $b>a>0$, and $H > -1/2$. Let $(x_i)_{i=1}^m$ be the nodes and $(w_i)_{i=1}^m$ be the weights of the Gaussian quadrature of level $m$ on the interval $[a, b]$ with weight function $w(x) \equiv 1$. Then,
\begin{align*}
\int_0^\infty \Bigg|\int_a^b e^{-xt} &x^{-H-1/2}\sdd x - \sum_{i=1}^m w_i x_i^{-H-1/2} e^{-x_it}\Bigg| \sdd t\\
&\le \nu_{m,r,H+3/2,c} \left(\frac{b-a}{2}\right)^{-H-1/2}\eps^{-H-1/2}\left(c + \sqrt{c^2-1}\right)^{-2m},
\end{align*}
where $c\coloneqq \frac{b+a}{b-a}$, $\gamma\coloneqq H + 3/2$, and  $\mu$, $\eps$ and $r$ are as in Lemma \ref{lem:ImprovedGaussianErrorForSpecificPowerFunction}.
\end{lemma}

\begin{proof}
Since the function $x\mapsto e^{-xt}x^{-H-1/2}$ is completely monotone, Proposition \ref{prop:GaussianQuadratureUnderestimatesCompletelyMonotoneFunctions} implies that
\begin{align*}
\int_0^\infty \Bigg|\int_a^b e^{-xt} &x^{-H-1/2}\sdd x - \sum_{i=1}^m w_i x_i^{-H-1/2} e^{-x_it}\Bigg| \sdd t\\
&= \int_0^\infty \Bigg(\int_a^b e^{-xt} x^{-H-1/2}\sdd x - \sum_{i=1}^m w_i x_i^{-H-1/2} e^{-x_it}\Bigg) \sdd t\\
&= \int_a^b x^{-H-3/2} \sdd x - \sum_{i=1}^m w_i x_i^{-H-3/2}.
\end{align*}

This is obviously the error of Gaussian quadrature for the function $x\mapsto x^{-H-3/2}$ on the interval $[a, b]$. By a simple linear transformation and Lemma \ref{lem:ImprovedGaussianErrorForSpecificPowerFunction}, we get the result.
\end{proof}

To simplify, we only prove asymptotic bounds from now on. More precisely, if we write $f(n)\approx g(n)$, then we mean $f(n) = g(n)(1 + o(n))$, so that leading constants remain valid.

\begin{lemma}\label{lem:IntermediateMeanReversion}
Let $(x_i)_{i=1}^N$ be the nodes and $(w_i)_{i=1}^N$ be the weights of a geometric Gaussian rule with parameters $\alpha > 0$, $\beta = \frac{\log(3 + 2\sqrt{2})}{\alpha}$, and $a\ge b > 0$, cf. Definition \ref{def:GeometricGaussianRules}. Assuming $c_n\coloneqq \left(\frac{b}{a}\right)^{1/n}\approx 1$, we have 
\begin{align*}
\int_0^\infty &\Bigg|c_H\int_{\xi_1}^{\xi_n} e^{-xt} x^{-H-1/2} \sdd x - \sum_{i=m+1}^N w_i e^{-x_it}\Bigg| \sdd t\\
&\lesssim C_1 a^{-H-1/2} N^{H/2+1/4}\left(1 + \frac{1}{2}(1-c_n)\right)^{-2m}\left(\sqrt{2}+1\right)^{-2m},
\end{align*}
where $$C_1 \coloneqq \frac{2}{\pi}c_H \left(\frac{2e\beta}{(3+2\sqrt{2})\sqrt{H+1/2}}\right)^{H+1/2}\frac{\left(\sqrt{2}+1\right)^{1/2-H}}{1-(3+2\sqrt{2})^{-H-1/2}}\Bigg(1 + \frac{(\pi/2)^{H+3/2}}{H+1/2}\Bigg).$$
\end{lemma}

\begin{proof}
Denote $c_n\coloneqq \left(\frac{b}{a}\right)^{1/n}$, and $L \coloneqq e^{\alpha\beta}$. Applying the triangle inequality for the intervals $[\xi_i, \xi_{i+1}]$ with $i=1,\dots,n-1$, and using Lemma \ref{lem:SingleIntervalL1Bound}, we get
\begin{align}
\int_0^\infty \Bigg| c_H\int_{\xi_1}^{\xi_n} e^{-xt} x^{-H-1/2} \sdd x &- \sum_{i=m+1}^N w_i e^{-x_it}\Bigg| \sdd t\nonumber\\
&\le C_2\sum_{i=1}^{n-1} \left(\frac{\xi_{i+1}-\xi_i}{2}\right)^{-H-1/2} \left(M + \sqrt{M^2-1}\right)^{-2m},\label{eqn:UpperBound801}
\end{align}
where $$C_2 \coloneqq \frac{8}{\pi}c_H\frac{4m^2}{4m^2-1}\frac{r}{(r - r^{-1})^2}\Bigg(1 + \frac{(\pi/2)^{H+3/2}}{H+1/2}\Bigg)e^{H+1/2}\eps^{-H-1/2},$$ and $M=\frac{c_nL+1}{c_nL-1}$, $\mu\coloneqq \frac{2m}{H + 1/2}$, $\eps = \frac{\sqrt{\mu^2(M^2-1)+1} - M}{\mu^2-1},$ and $r = M-\eps + \sqrt{(M-\eps)^2-1},$ and where we assume that $\mu \ge \frac{1}{M-1}\lor \frac{3M}{2\sqrt{M^2-1}}.$

Consider first the sum in \eqref{eqn:UpperBound801}. Since $\xi_i = ac^i_n L^i$, we have
\begin{align}
\sum_{i=1}^{n-1} \left(\frac{\xi_{i+1}-\xi_i}{2}\right)^{-H-1/2} &\le (ac_nL)^{-H-1/2}\left(\frac{c_nL-1}{2}\right)^{-H-1/2} \frac{1}{1 - (c_nL)^{-H-1/2}}.\label{eqn:GeometricSumBound}
\end{align}

Next, we want to determine the rate at which \eqref{eqn:UpperBound801} decays. Heuristically, it seems that the rate is mainly determined by the term $\left(M + \sqrt{M^2-1}\right)^{-2m}$. Note that $$M = \frac{L+1}{L-1} + 2L\frac{1-c_n}{(c_nL-1)(L-1)} =: M_0 + \delta\quad\textup{with}\quad M_0\coloneqq \frac{L+1}{L-1},$$ where $\delta \ll 1$, for $c_n\approx 1$. Let us hence consider the expression $\left(M_0 + \sqrt{M_0^2-1}\right)^{-2m}.$ Here, $M_0$ depends on $\alpha$ and $\beta$, while $m$ depends on $\beta$, so with this simplification we have removed the dependence of the rate on the parameters $a$ and $b$. Recalling Remark \ref{rem:IntuitionsOfGeometricGaussianRuleParameters}, we want to choose a good value of $\beta$ (i.e. a good degree of the Gaussian quadrature rule) to make the error as small as possible. Therefore, we consider the optimization problem 

$$\inf_{\beta > 0} \left(M_0 + \sqrt{M_0^2 - 1}\right)^{-2m}.$$ After some manipulations, we see that minimizing this is equivalent to minimizing $$\left(\frac{e^{\alpha\beta}-1}{e^{\alpha\beta} + 2e^{\alpha\beta/2} + 1}\right)^{\alpha\beta}.$$ Perhaps surprisingly, this can be optimized in closed form, and the solution is $$L = e^{\alpha\beta} = 3 + 2\sqrt{2},\quad\text{i.e.}\quad \beta = \frac{\log(3 + 2\sqrt{2})}{\alpha}.$$ This implies $$M_0 = \frac{e^{\alpha\beta}+1}{e^{\alpha\beta}-1} = \sqrt{2},\quad\textup{and}\quad M_0 + \sqrt{M_0^2-1} = \sqrt{2} + 1.$$

Moreover, for all $c_n \le 1$,
\begin{align*}
M + \sqrt{M^2-1} &= \sqrt{2}+\delta + \sqrt{(\sqrt{2}+\delta)^2 - 1}\\
&= \sqrt{2}+(1 + \sqrt{2})\frac{1-c_n}{c_n(3 + 2\sqrt{2})-1}\\
&\qquad + \sqrt{\left(\sqrt{2}+(1 + \sqrt{2})\frac{1-c_n}{c_n(3 + 2\sqrt{2})-1}\right)^2 - 1}\\
&\ge \left(\sqrt{2}+1\right)\left(1 + \frac{1}{2}(1-c_n)\right).
\end{align*}

Using this bound, and \eqref{eqn:GeometricSumBound}, we get
\begin{align*}
\int_0^\infty &\Bigg|c_H\int_{\xi_1}^{\xi_n} e^{-xt} x^{-H-1/2} \sdd x - \sum_{i=m+1}^N w_i e^{-x_it}\Bigg| \sdd t\\
&\le C_2 \left(\frac{2}{ac_nL}\right)^{H+1/2}\frac{\left(c_nL-1\right)^{-H-1/2}}{1-(c_nL)^{-H-1/2}} \left(1 + \frac{1}{2}(1-c_n)\right)^{-2m}\left(\sqrt{2}+1\right)^{-2m}.
\end{align*}

It is now that we start using asymptotic bounds. As $N\to\infty$, we have $m, n\to\infty$. Since $c_n\approx 1$, we have $M\approx\sqrt{2}$, $\mu = \frac{2m}{H+1/2}$, $\eps\approx \mu^{-1}$, and $r\approx \sqrt{2}+1$. Plugging in these values, we get the bound in the statement of the theorem. After noting that $\mu = \frac{2m}{H+1/2} \ge \sqrt{2} + 1 = \frac{1}{M-1}\lor\frac{3M}{2\sqrt{M^2-1}}$ is satisfied for all $m\ge 2$, this proves the lemma.
\end{proof}

We have now estimated the approximation error on the interval $[\xi_1, \xi_n]$. Next, we consider the interval $[\xi_0, \xi_1]$.

\begin{lemma}\label{lem:LowMeanReversion}
Let $(x_i)_{i=1}^m$ be the nodes and $(w_i)_{i=1}^m$ be the weights of Gaussian quadrature of level $m$ on the interval $[0, a]$ with respect to the weight function $w(x) = c_H x^{-H-1/2}$ with $H > -1/2$. Assume $a \le 2(m+1)T^{-1}$. Then,
\begin{align*}
\int_0^T\abs{c_H\int_0^a e^{-tx}x^{-H-1/2} \sdd x - \sum_{i=1}^m w_i e^{-x_it}}\sdd t &\le c_HT^{1/2+H}\frac{(Ta)^{2m+1/2-H}}{(2m+1)!(2m+1/2-H)}.
\end{align*}
\end{lemma}

\begin{proof}
Since $x\mapsto e^{-tx}$ is a completely monotone function, Proposition \ref{prop:GaussianQuadratureUnderestimatesCompletelyMonotoneFunctions} implies that
\begin{align*}
\int_0^T\Bigg|c_H\int_0^a e^{-tx}x^{-H-1/2} \sdd x &- \sum_{i=1}^m w_i e^{-x_it}\Bigg|\sdd t\\
&= \int_0^T\left(c_H\int_0^a e^{-tx}x^{-H-1/2} \sdd x - \sum_{i=1}^m w_i e^{-x_it}\right)\sdd t\\
&= c_H\int_0^a \left(1 - e^{-Tx}\right)x^{-H-3/2} \sdd x - \sum_{i=1}^m \frac{w_i}{x_i}\left(1 - e^{-Tx_i}\right).
\end{align*}

This is obviously the Gaussian quadrature error of the function
\begin{align*}
f(x) &= (1-e^{-Tx})x^{-1} = T\sum_{n=0}^\infty \frac{(-xT)^n}{(n+1)!}
\end{align*}
on the interval $[0, a]$ with respect to the weight function $w(x) = c_H x^{-H-1/2}$. In particular, because Gaussian quadrature integrates exactly polynomials up to degree $2m-1$, instead of $f$ we may consider $$g(x)\coloneqq f(x) - T\sum_{n=0}^{2m-1} \frac{(-xT)^n}{(n+1)!} = T^{2m+1} x^{2m} \sum_{n=0}^\infty \frac{(-xT)^n}{(n+2m+1)!}.$$ 

If $$x \le (2m+2)T^{-1},\quad \textup{then}\quad 0 \le g(x) \le \frac{T^{2m+1}x^{2m}}{(2m+1)!}.$$ Hence, for $a \le (2m+2)T^{-1}$, we have (due to the positivity of the weights)
\begin{align*}
\int_0^T\abs{c_H\int_0^a e^{-tx}x^{-H-1/2} \sdd x - \sum_{i=1}^m w_i e^{-x_it}}\sdd t &= c_H\int_0^a g(x)x^{-H-1/2} \sdd x - \sum_{i=1}^m w_ig(x_i)\\
&\le c_H \int_0^a \frac{T^{2m+1} x^{2m-H-1/2}}{(2m+1)!} \sdd x\\
&= c_H \frac{T^{2m+1} a^{2m+1/2-H}}{(2m+1)!(2m+1/2-H)}.\qedhere
\end{align*}
\end{proof}

Finally, the following formula for the approximation error on $[\xi_n, \infty)$ is trivial.

\begin{lemma}\label{lem:HighMeanReversion}
We have
\begin{align*}
\int_0^\infty c_H \int_a^\infty e^{-tx} x^{-H-1/2} \sdd x \sdd t &= \frac{c_H}{1/2 + H}a^{-H-1/2}.
\end{align*}
\end{lemma}

\begin{proof}
This follows immediately using Fubini's theorem with
\begin{align*}
\int_0^\infty c_H \int_a^\infty e^{-tx} x^{-H-1/2} \sdd x \sdd t &= c_H \int_a^\infty x^{-H-3/2} \sdd x = \frac{c_H}{1/2 + H}a^{-H-1/2}.\qedhere
\end{align*}
\end{proof}

We are finally able to state the following result on the convergence rate of the Gaussian approximations.

\begin{theorem}\label{thm:TheL1TheoremGeometric}
Let $(x_i)_{i=1}^N$ be the nodes and $(w_i)_{i=1}^N$ be the weights of the geometric Gaussian rule with $\alpha = \log(3 + 2\sqrt{2})$, $\beta = 1$, $a = \frac{10\sqrt{2}-14}{e}\sqrt{(H+1/2)N}T^{-1}$, and $b=\frac{10\sqrt{2}-14}{e}T^{-1}$. Then,
\begin{align}
\int_0^T \abs{K(t) - K^N(t)} \sdd t &\lesssim \frac{c_H}{H+1/2}T^{H+1/2}\left(\frac{e}{10\sqrt{2}-14}\right)^{H+1/2}\left(\sqrt{2}+1\right)^{-2\sqrt{(H+1/2)N}}.\label{eqn:FirstGeometricBound}
\end{align}

If instead we choose $b = \frac{10\sqrt{2}-14}{e}((H+1/2)N)^{1/4}T^{-1},$ then,
\begin{align}
\int_0^T \abs{K(t) - K^N(t)} \sdd t &\lesssim \frac{200}{3}c_HT^{H+1/2}N^{-H/4-1/8}\left(\sqrt{2}+1\right)^{-2\sqrt{(H+1/2)N}}.\label{eqn:SecondGeometricBound}
\end{align}
\end{theorem}

\begin{remark}
We decided to give two different bounds depending on the choice of the parameter $b$ in the theorem above. While the latter bound obviously yields the ever so slightly better convergence rate, the former bound may be more convenient in theoretical applications. This is because in the latter bound, the largest node $x$ has an additional polynomial factor of growth that is not present in the former bound. Hence, using the latter bound in theoretical applications may lead to additional annoying logarithmic error terms that can be avoided by using the former bound.
\end{remark}

\begin{proof}
We only prove \eqref{eqn:FirstGeometricBound}, the proof of \eqref{eqn:SecondGeometricBound} being similar. By the triangle inequality, and Lemmas \ref{lem:LowMeanReversion}, \ref{lem:IntermediateMeanReversion}, and \ref{lem:HighMeanReversion},

\begin{align*}
\int_0^T \abs{K(t) - K^N(t)} \sdd t
&\lesssim c_HT^{1/2+H}\frac{(Tac_n(3+2\sqrt{2}))^{2m+1/2-H}}{(2m+1)!(2m+1/2-H)}\\
&\qquad + C_1 a^{-H-1/2} N^{H/2+1/4}\left(1 + \frac{1}{2}(1-c_n)\right)^{-2m} \left(\sqrt{2}+1\right)^{-2m}\\
&\qquad + \frac{c_H}{1/2 + H}b^{-H-1/2}\exp\left(-\alpha\sqrt{(H+1/2)N}\right),
\end{align*}
where $C_1$ is the constant from Lemma \ref{lem:IntermediateMeanReversion}, and where we assume that $c_n\coloneqq \left(\frac{b}{a}\right)^{1/n}\approx 1$, $a\le 2(m+1)T^{-1}$, $\beta = \frac{\log(3+2\sqrt{2})}{\alpha}$, and $b\le a$. 

Let us now choose $\alpha$ such that $$\exp\left(-\alpha\sqrt{(H+1/2)N}\right) = \left(\sqrt{2} + 1\right)^{-2\beta\sqrt{(H+1/2)N}},$$ in order to ensure that the latter two terms converge at the same speed. It is clear that this is achieved for $$\alpha = \log\left(3 + 2\sqrt{2}\right),\quad\text{i.e.}\quad\beta=1.$$

Next, we want to ensure that the first summand converges at the same speed as the other two. Using Stirling's formula, we have
\begin{align*}
\frac{(Tac_n(3+2\sqrt{2}))^{2m+1/2-H}}{(2m+1)!(2m+1/2-H)} &\lesssim \frac{(Tac_n(3+2\sqrt{2}))^{2m+1/2-H}}{4m^2\sqrt{2\pi\cdot 2m} \left(\frac{2m}{e}\right)^{2m}}\\
&\approx \frac{(Ta(3+2\sqrt{2}))^{1/2-H}}{8m^{5/2}\sqrt{\pi}}c_n^{2m}\left(\frac{eTa(3+2\sqrt{2})}{2m}\right)^{2m}.
\end{align*}
To ensure a similar speed of convergence, we need $$\frac{eTa(3+2\sqrt{2})}{2m} = \sqrt{2} - 1,\quad\text{i.e.}\quad a = \frac{\left(10\sqrt{2}-14\right)m}{eT}.$$ Obviously, for this choice of $a$ we have $a \le 2(m+1)T^{-1}$. Furthermore, we now have
\begin{align*}
\int_0^T \abs{K(t) - K^N(t)} \sdd t &\lesssim c_HT^{H+1/2}\Bigg(\frac{(\frac{1}{e}(2\sqrt{2}-2))^{1/2-H}}{8\sqrt{\pi}}m^{-2-H}c_n^{2m}\\
&\qquad + \frac{2}{\pi}\left(\frac{e^2}{H+1/2}\right)^{H+1/2}\frac{\sqrt{2}+1}{1-(3+2\sqrt{2})^{-H-1/2}}\\
&\qquad \times \Bigg(1 + \frac{(\pi/2)^{H+3/2}}{H+1/2}\Bigg) \left(1 + \frac{1}{2}(1-c_n)\right)^{-2m}\\
&\qquad + \frac{1}{1/2 + H}(Tb)^{-H-1/2}\Bigg)\left(\sqrt{2}+1\right)^{-2\sqrt{(H+1/2)N}}.
\end{align*}

Assuming that we choose $b$ constant (i.e. independent of $N$), such that $b \le a$, we see that only the last summand in the above bound is relevant, due to the polynomial terms in $N$. Indeed, since $a = C\sqrt{N}$ for some $C,$ it is not hard to see that we have both $$m^{-2-H} c_n^m = o(1)\quad\textup{and}\quad \left(1 + \frac{1}{2}(1-c_n)\right)^{-2m} = o(1).$$

Thus,
\begin{align*}
\int_0^T \abs{K(t) - K^N(t)} \sdd t &\lesssim \frac{c_H}{H+1/2}T^{H+1/2}(Tb)^{-H-1/2}\left(\sqrt{2}+1\right)^{-2\sqrt{(H+1/2)N}}.
\end{align*}
To ensure that $b\le a$ for all $m\ge 1$, we choose $b = \frac{10\sqrt{2}-14}{eT},$ yielding \eqref{eqn:FirstGeometricBound}.
\end{proof}

\subsection{Non-geometric Gaussian approximations}\label{sec:NonGeometricGaussianApproximations}

We will try to further improve the results of Section \ref{sec:GeometricGaussianApproximations} by choosing non-geometrically spaced intervals $[\xi_i,\xi_{i+1}]$.

\begin{definition}[Non-geometric Gaussian Rules]
Let $N\in\N$ be the total number of nodes and $\beta,a,c\in(0,\infty)$ be parameters of the scheme, where $c > 1$. We define $$m\coloneqq \mathrm{rd}\left(\beta\sqrt{(1/2+H)N}\right),\qquad n\coloneqq \mathrm{rd}\left(\frac{1}{\beta}\sqrt{\frac{N}{1/2+H}}\right) (\approx N/m),$$
\begin{equation}\label{eqn:StrangeXiDefinition}
\xi_0\coloneqq 0,\quad \xi_1 \coloneqq a,\quad\xi_{i+1}=\left(\frac{c + \xi_i^{\frac{1/2+H}{2m}}}{c - \xi_i^{\frac{1/2+H}{2m}}}\right)^2 \xi_i,\qquad i=1,\dots,n-1.
\end{equation}
Here, we assume that $\xi_i^{\frac{1/2+H}{2m}} < c$ for all $i=1,\dots,n-1$, which is true if we choose $c$ large enough. Define $(x_i)_{i=1}^m$ to be the nodes and $(\widetilde{w}_i)_{i=1}^m$ to be the weights of Gaussian quadrature of level $m$ on the interval $[0, \xi_1]$ with the weight function $w(x) = c_H x^{-H-1/2}$. Furthermore, let $(x_i)_{i=m+1}^{mn}$ be the nodes and $(w_i)_{i=m+1}^{mn}$ be the weights of Gaussian quadrature of level $m$ on to the intervals $[\xi_i, \xi_{i+1}]$ for $i=1,\dots,n-1$ with the weight function $w(x) \equiv 1$. Then we define the non-geometric Gaussian rule of type $(H,N, c,\beta,a)$ to be the set of nodes $(x_i)_{i=1}^{mn}$ with weights $(w_i)_{i=1}^{mn}$ defined by $$w_i \coloneqq \begin{cases}
\widetilde{w}_i,\qquad &\textup{if } i=1,\dots,m,\\
c_H\widetilde{w}_i x_i^{-H-1/2},\qquad &\textup{else}.
\end{cases}$$
\end{definition}

\begin{remark}
    Non-geometric Gaussian rules have three free parameters: $\beta, a, c$. They can be interpreted as follows.
    \begin{itemize}
        \item The parameter $c$ determines the cutoff point $\xi_n$ of the interval in \eqref{eqn:LaplaceRepresentationOfK}, similar to the parameter $\alpha$ in geometric Gaussian rules. It is hence mainly used to control the error on $[\xi_n,\infty)$.
        \item The parameter $\beta$ determines the degree of the Gaussian quadrature rule, similar to the parameter $\beta$ in geometric Gaussian rules. It is hence mainly used to control the error on $[\xi_1,\xi_n]$.
        \item The parameter $a$ determines the size of the interval $[0, a]$, which is again special due to the singularity in the weight function $w$. Similarly to the parameter $a$ in geometric Gaussian rules, it is mainly used to control the error on $[0,\xi_1]$.
    \end{itemize}
    Finally, we remark that the reason for the specific definition of $\xi_{i+1}$ in \eqref{eqn:StrangeXiDefinition} will become apparent in the proof of Lemma \ref{lem:IntermediateNonGeometricMeanReversions}. Furthermore, the parameters $c,\beta,a$ will not depend on $N$ in our proofs.
\end{remark}

The following lemma is the equivalent of Lemma \ref{lem:IntermediateMeanReversion} for non-geometric Gaussian rules. It bounds the error of the quadrature rule on $[\xi_1,\xi_n]$.

\begin{lemma}\label{lem:IntermediateNonGeometricMeanReversions}
Let $(x_i)_{i=1}^N$ be the nodes and $(w_i)_{i=1}^N$ be the weights of a non-geometric Gaussian rule. Then,
\begin{align*}
\int_0^\infty &\abs{c_H\int_{\xi_1}^{\xi_n} e^{-xt}x^{-H-1/2} \sdd x - \sum_{i=m+1}^N w_i e^{-x_it}}\sdd t\\
&\lesssim c_H\frac{8}{\pi}\Bigg(1 + \frac{(\pi/2)^{H+3/2}}{H+1/2}\Bigg)e^{H+1/2}\sum_{i=1}^{n-1} \frac{r_i}{(r_i - r_i^{-1})^2}\left(\frac{L_i-1}{2}\right)^{-H-1/2}\eps_i^{-H-1/2}c^{-2m},
\end{align*}
where $L_i=\frac{\xi_{i+1}}{\xi_i}$, $M_i=\frac{L_i+1}{L_i-1}$, $\mu\coloneqq \frac{2m}{H + 1/2}$, $\eps_i = \frac{\sqrt{\mu^2(M_i^2-1)+1} - M_i}{\mu^2-1},$ and $r_i = M_i-\eps_i + \sqrt{(M_i-\eps_i)^2-1},$ and where we assume that $\mu \ge \frac{1}{M_i-1}\lor \frac{3M_i}{2\sqrt{M_i^2-1}}$ for $i=1,\dots,n-1$.
\end{lemma}

\begin{proof}
As in the proof of Lemma \ref{lem:IntermediateMeanReversion}, we apply the triangle inequality for the intervals $[\xi_i, \xi_{i+1}]$ with $i=1,\dots,n-1$, and use Lemma \ref{lem:SingleIntervalL1Bound}. Then, the statement of the lemma follows directly after noting that by the definition of  $\xi_{i+1}$ in \eqref{eqn:StrangeXiDefinition}, we have
\begin{align*}
\xi_i^{-H-1/2} \left(M_i + \sqrt{M_i^2-1}\right)^{-2m} &= c^{-2m}.\qedhere
\end{align*}
\end{proof}

While Lemma \ref{lem:IntermediateNonGeometricMeanReversions} is the equivalent of Lemma \ref{lem:IntermediateMeanReversion} for non-geometric Gaussian rules to bound the error on $[\xi_1,\xi_n]$, Lemma \ref{lem:LowMeanReversion} for the error on $[0,\xi_1]$ and Lemma \ref{lem:HighMeanReversion} for the error on $[\xi_n, \infty)$ can be reused exactly. The only difference to the geometric Gaussian rules is that we do not have an explicit formula for $\xi_n$, but merely the recursion \eqref{eqn:StrangeXiDefinition}. However, the size of $\xi_n$ is needed to determine the error contribution on $[\xi_n,\infty)$ in Lemma \ref{lem:HighMeanReversion}. Hence, in the following lemma we determine the approximate size of $\xi_n$ for non-geometric Gaussian rules.

\begin{lemma}\label{lem:L1LargestNodeEstimate}
    Let $\eta = \eta(c, \beta)$ be the solution of the ODE 
    \begin{equation}\label{eqn:EtaODE}
    \frac{\dd\eta_t}{\dd t} = 2\log\left(1 + \frac{2e^{\frac{\eta_t}{2\beta^2}}}{c-e^{\frac{\eta_t}{2\beta^2}}}\right),\qquad \eta_0 = 0.
    \end{equation}
    Note that $\eta$ has a finite explosion time $T_0$. The following two statements are equivalent.
    \begin{enumerate}
        \item The explosion time $T_0>1$, i.e. $\eta_1 < \infty$.
        \item There exists $\delta > 0$ such that for $N$ sufficiently large we have
        \begin{equation}\label{eqn:XIUniformBoundCondition}
        \max_{i=1,\dots,n} \xi_i^{\frac{H + 1/2}{2m}} < (1 - \delta)c.
        \end{equation}
    \end{enumerate}
    If one of these conditions is satisfied, then $$\xi_n = \Omega(1)\exp\left(\frac{1}{\beta}\sqrt{\frac{N}{H + 1/2}}\eta_1\right).$$
\end{lemma}

\begin{remark}
    We need to have $\max_{i=1,\dots,n} \xi_i^{\frac{1/2+H}{2m}} < c$ to ensure that the non-geometric Gaussian rule is well-defined, cf. \eqref{eqn:StrangeXiDefinition}. This lemma hence gives a semi-explicit asymptotic formula for $\xi_n$ if we satisfy this condition even with an additional $\delta > 0$ of leeway.
\end{remark}

\begin{proof}
Define $\kappa\coloneqq \frac{(H+1/2)n}{2m}\approx \frac{1}{2\beta^2},$ and assume that $N$ is large. Recall that $$\xi_{i+1} = \xi_i \left(\frac{c + \xi_i^{\frac{H+1/2}{2m}}}{c - \xi_i^{\frac{H+1/2}{2m}}}\right)^2 = \xi_i \left(1 + \frac{2\xi_i^{\kappa/n}}{c-\xi_i^{\kappa/n}}\right)^2.$$ Set $\eta^{(n)}_{i/n}\coloneqq \log \xi_i^{1/n}.$ Then, $$\eta_{\frac{i+1}{n}}^{(n)} = \eta_{\frac{i}{n}}^{(n)} + \frac{2}{n}\log\left(1 + \frac{2\exp\left(\kappa\eta_{\frac{i}{n}}^{(n)}\right)}{c-\exp\left(\kappa\eta_{\frac{i}{n}}^{(n)}\right)}\right),\qquad \eta^{(n)}_{1/n} = \frac{1}{n}\log a.$$ This is almost the Euler discretization of the ODE \eqref{eqn:EtaODE}, whose solution we denote by $\eta$.

If we now have the uniform bound $\xi_i^{\frac{H+1/2}{2m}} < (1 - \delta)c$, clearly also $\eta^{(n)}$ remains uniformly bounded in $n$. In particular, we may assume that the vector field in \eqref{eqn:EtaODE} is bounded (for $t\le 1$), and we have the convergence $$\lim_{n\to\infty}\eta^{(n)}_1 = \eta_1.$$ The error $\abs{\eta^{(n)}_1 - \eta_1}$ will of course depend on $\delta$. Furthermore, this also implies that $\eta_1 < \infty$, and hence $T_0 > 1$.

Conversely, assume that $T_0 > 1$, and for now, assume also $a \le 1$. Since the vector field governing $\eta$ is increasing and non-negative, the Euler discretization is lower-biased. This implies in particular that $\eta_1^{(n)} \le \eta_1$. If $a > 1$, then there exists $t_n \in (0,\infty)$ with $\eta_{t_n} = \frac{1}{n}\log a = \eta_{1/n}^{(n)}.$ Clearly, $t_n \to 0$ as $n\to\infty$. Again since the Euler scheme is lower-biased, we have $$\eta_1^{(n)} \le \eta_{1+t_n-1/n} \le \eta_{(1 + T_0) / 2} < \infty,$$ where the second inequality holds for $n$ sufficiently large. Therefore,
\begin{align*}
    \max_{i=1,\dots,n} \xi_i^{\frac{H + 1/2}{2m}} = \xi_n^{\frac{H + 1/2}{2m}} = \exp\left(n\frac{H + 1/2}{2m}\eta_1^{(n)}\right) \le \exp\left(\frac{1}{2\beta^2}\eta_{(1 + T_0)/2}\right) < c,
\end{align*}
where the last (strict!) inequality holds, since the explosion time $T_0$ is exactly at $\exp\left(\frac{1}{2\beta^2}\eta_{(1 + T_0)/2}\right) = c.$ Thus, since the inequality is strict, and since $\exp\left(\frac{1}{2\beta^2}\eta_{(1 + T_0)/2}\right)$ is independent of $n$, we can find a uniform $\delta$ satisfying \eqref{eqn:XIUniformBoundCondition}.

Finally, assume that one of these two conditions (and hence both) are satisfied. Then, we prove the asymptotic formula for $\xi_n$. Due to the boundedness of the vector field (since $T_0 > 1$), we have that the solution of the ODE is Lipschitz in the initial condition and in the driving vector field. Now, $\kappa = \frac{1}{2\beta^2} + O(N^{-1/2})$, $\eta^{(n)}_{1/n} = O(N^{-1/2})$, and $n=\sqrt{N}$, implying that we have $$\eta^{(n)}_1 = \left(1 + f(N^{-1})\right) \eta_1,$$ where $f$ satisfies $\abs{f(N)} = O(N^{-1/2}).$ Therefore,
\begin{align*}
\xi_n &= \exp\left(n \left(1 + f(N)\right)\eta_1\right) = \exp\left(\frac{1}{\beta}\sqrt{\frac{N}{H+1/2}} \eta_1\right)\Omega(1).\qedhere
\end{align*}
\end{proof}

We can now proceed with the proof of the error bound for non-geometric Gaussian rules.

\begin{theorem}\label{thm:TheL1TheoremNonGeometric}
Define the constants $\beta_0 = 0.92993273,$ and $c_0 = 3.60585021.$ Let $K^N$ be a Gaussian approximation coming from a non-geometric Gaussian rule with parameters $c \ge c_0$, $\beta\ge\beta_0$, and $a>0$, where $c>c_0$ or $\beta > \beta_0$. Then, 
\begin{align*}
\int_0^T \abs{K(t) - K^N(t)} \dd t &\lesssim \frac{c_H}{H+1/2}\xi_n^{-1/2-H}\\
&= \Omega(1)\exp\left(-\frac{\eta_1(c, \beta)}{\beta}\sqrt{(H+1/2)N}\right).
\end{align*}
In the first inequality, the higher order terms explode as $(c,\beta)\to (c_0, \beta_0)$, and in the second inequality, the leading constant explodes for the same limit. Furthermore, we have $$\frac{\eta_1(c_0, \beta_0)}{\beta_0} = 2.3853845446404978.$$
\end{theorem}

In Figure \ref{fig:ExponentsOfVariousQuadratureRules} we compare the convergence rates of our results with the previously achieved strong convergence rate in \cite[Theorem 2.1]{bayer2023markovian}.

\begin{figure}
    \centering
    \includegraphics[scale=0.6]{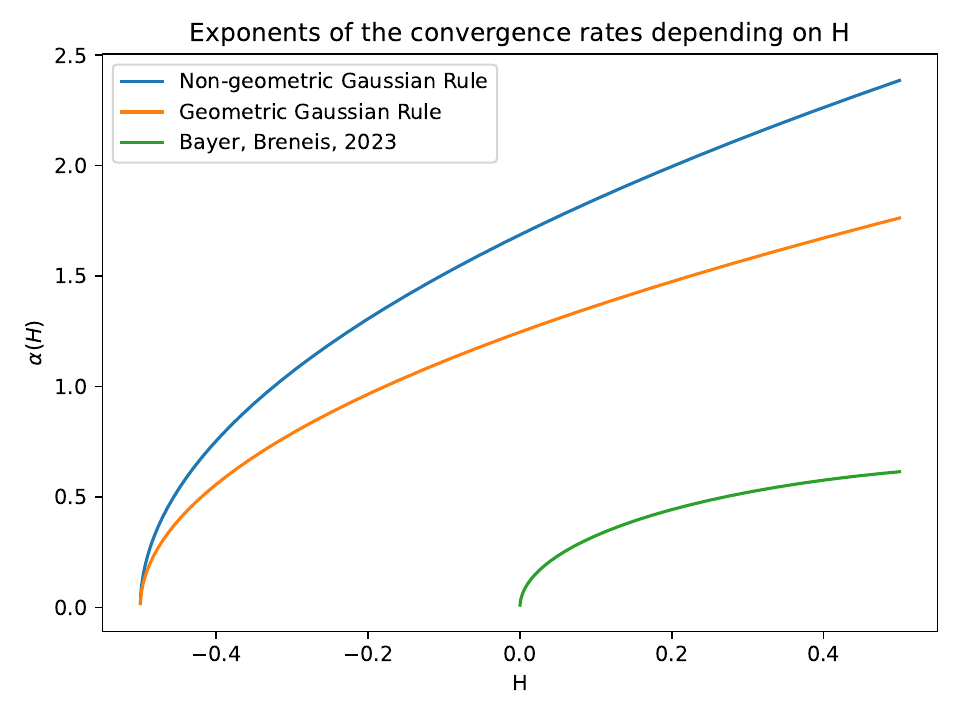}
    \caption{The error bounds for the approximations $K^N$ suggested in \cite[Theorem 2.1]{bayer2023markovian}, and in our Theorems \ref{thm:TheL1TheoremGeometric} and \ref{thm:TheL1TheoremNonGeometric} are all of the form $C N^\gamma e^{-\alpha(H)\sqrt{N}}$. This figure illustrates how the exponent of convergence $\alpha(H)$ varies in $H$ for these three approximations.}
    \label{fig:ExponentsOfVariousQuadratureRules}
\end{figure}

\begin{proof}[Proof of Theorem \ref{thm:TheL1TheoremNonGeometric}]
Assume for now that the conditions of Lemma \ref{lem:L1LargestNodeEstimate} are satisfied, i.e. the explosion time $T_0$ of $\eta$ in \eqref{eqn:EtaODE} satisfies $T_0 > 1$. Then we can apply the triangle inequality on the intervals $[0, \xi_1]$, $[\xi_1,\xi_n]$, and $[\xi_n,\infty)$, and apply Lemmas \ref{lem:LowMeanReversion}, \ref{lem:IntermediateNonGeometricMeanReversions}, \ref{lem:HighMeanReversion} and \ref{lem:L1LargestNodeEstimate}, to get
\begin{align}
\int_0^T \abs{K(t) - K^N(t)} \sdd t &\le O(1)\Bigg(\frac{(Ta)^{2m+1/2-H}}{m^2(2m)!}\nonumber\\
&\qquad + \sum_{i=1}^{n-1} \frac{r_i}{(r_i - r_i^{-1})^2}\left(L_i-1\right)^{-H-1/2}\eps_i^{-H-1/2}c^{-2\beta\sqrt{(H+1/2)N}}\nonumber\\
&\qquad + \exp\left(-\frac{1}{\beta}\sqrt{(H+1/2)N}\eta_1(c, \beta)\right)\Bigg),\label{eqn:ThisBloodyBound}
\end{align}
where $a\le 2(m+1)T^{-1}$, $L_i=\frac{\xi_{i+1}}{\xi_i}$, $M_i=\frac{L_i+1}{L_i-1}$, $\mu\coloneqq \frac{2m}{H + 1/2}$, $\eps_i = \frac{\sqrt{\mu^2(M_i^2-1)+1} - M_i}{\mu^2-1},$ and $r_i = M_i-\eps_i + \sqrt{(M_i-\eps_i)^2-1},$ and where we assume that $\mu \ge \frac{1}{M_i-1}\lor \frac{3M_i}{2\sqrt{M_i^2-1}}.$

The first summand in this bound converges factorially fast (for $a$ independent of $N$), and will hence be of no further concern. On the other hand, we want that the second and the third summand converge at the same speed. Hence, we wish to solve the optimization problem $$\arginf_{c>1, \beta > 0} c^{-2\beta\sqrt{(H+1/2)N}},\ \textup{subject to}\  c^{-2\beta\sqrt{(H+1/2)N}} = \exp\left(-\frac{1}{\beta}\sqrt{(H+1/2)N}\eta_1(c, \beta)\right).$$ This is equivalent to
\begin{equation}\label{eqn:DimensionFreeOptimizationProblem}
\arginf_{c>1, \beta > 0} -\beta\log c,\quad\textup{subject to}\quad 2\beta^2\log c = \eta_1(c, \beta).
\end{equation}
The solution can be numerically computed to be $$\beta_0 = 0.92993273,\qquad c_0 = 3.60585021.$$ For these values, we have $$c_0^{-2\beta_0} = e^{-2.3853845446404978}.$$

However, we still need to ensure that $T_0 = T_0(c,\beta) > 1$. In fact, for the choice $(c, \beta) = (c_0,\beta_0)$, it is not hard to see that this condition will likely be violated. Indeed, recall that the explosion happens exactly when $c = \exp\left(\frac{\eta_{T_0}}{2\beta^2}\right),$ or equivalently, when $2\beta^2\log c = \eta_{T_0}$. But this is precisely the constraint in the optimization problem \eqref{eqn:DimensionFreeOptimizationProblem} with $T_0=1$. Therefore, Lemma \ref{lem:L1LargestNodeEstimate} is in fact not applicable for $(c, \beta) = (c_0,\beta_0)$. This can be remedied by choosing $c$ or $\beta$ slightly larger. Indeed, note that the vector field in \eqref{eqn:EtaODE} is strictly decreasing in both $c$ and $\beta$, and hence, $T_0(c, \beta)$ is strictly increasing. Therefore, if $c\ge c_0$, $\beta\ge \beta_0$, and $c > c_0$ or $\beta > \beta_0$, then $T_0(c, \beta) > 1$, and our application of Lemma \ref{lem:L1LargestNodeEstimate} was valid.

What does this choice of $c$ and $\beta$ mean for the convergence rate of the kernel error? Recall that $(c_0, \beta_0)$ was the solution to the optimization problem where we wanted to optimize over the convergence rate subject to the constraint that the second and third summand in \eqref{eqn:ThisBloodyBound} converge at the same speed. If we now increase $c$ or $\beta$, then the second term will converge faster, while the third term will converge more slowly. We can further quantify this by stating that $$c^{-2\beta\sqrt{(H+1/2)N}} \le O(1)e^{-\gamma\sqrt{N}}\exp\left(-\frac{1}{\beta}\sqrt{(H+1/2)N}\eta_1(c, \beta)\right)$$ for some small $\gamma > 0$ depending on $c, \beta$. 

We have now almost shown that the second term in \eqref{eqn:ThisBloodyBound} converges faster than the third term, and that we can hence (asymptotically) ignore it. However, we still need to treat the constants in front, to ensure that they do not cause additional troubles as $N\to\infty$. Indeed, one can show that $$\sum_{i=1}^{n-1} \frac{r_i}{(r_i + r_i^{-1})^2} (L_i - 1)^{-H-1/2} \eps_i^{-H-1/2} = O\left(\sum_{i=1}^{n-1} m^{H+1/2}\right) = O\left(N^{3/4+H/2}\right).$$

Altogether, this shows the bound
\begin{align*}
\int_0^T \abs{K(t) - K^N(t)} \sdd t &\le O(1)\exp\left(-\frac{1}{\beta}\sqrt{(H+1/2)N}\eta_1(c, \beta)\right).
\end{align*}

Finally, we note that only the error on the interval $[\xi_n, \infty)$ remains in the above bound, the errors on the other intervals having converged with higher order. Therefore, we also have
\begin{align*}
\int_0^T \abs{K(t) - K^N(t)} \sdd t &\lesssim \int_0^T \int_{\xi_n}^\infty c_H e^{-tx} x^{-H-1/2} \sdd x \sdd t \le \frac{c_H}{1/2 + H}\xi_n^{-H-1/2}.\qedhere
\end{align*}
\end{proof}

\section{Numerics}\label{sec:Numerics}

Throughout this section, we compare the approximations $K^N$ of $K$ given in Section \ref{sec:ErrorRates} with previously proposed methods. We demonstrate the efficiency of the geometric and non-geometric Gaussian rules for the pricing of various options under the rough Heston model, but we also indicate possible ways how these theoretically derived rules can be further improved numerically. With that in mind, below is a list of the quadrature rules we compare for approximating $K$ by $K^N$.

\begin{enumerate}
    \item ``GG'' (Geometric Gaussian quadrature): This is the quadrature rule from Theorem \ref{thm:TheL1TheoremGeometric} with $a = 4/T$ and $b = 1/(2T)$.
    \item `'NGG'' (Non-Geometric Gaussian quadrature): This is the quadrature rule from Theorem \ref{thm:TheL1TheoremNonGeometric} with $c=c_0$, $\beta=\beta_0$, and $a = 3/T$.
    \item ``OL1'' (Optimal $L^1$-error): This is the quadrature rule we get by optimizing the $L^1$-error between $K^N$ and $K$. We use the intersection algorithm outlined in Appendix \ref{sec:L1NormComputation} for computing the $L^1$-errors.
    \item ``OL2'' (Optimal $L^2$-error): This is the quadrature rule we get by optimizing the $L^2$-error between $K^N$ and $K$. We use \cite[Proposition 2.11]{bayer2023markovian} for computing $L^2$-errors.
    \item ``BL2'' (Bounded $L^2$-error): This quadrature rule essentially minimizes the $L^2$-error between $K^N$ and $K$, but penalizes large nodes $x$. A more precise description of the algorithm underlying this quadrature rule can be found in Appendix \ref{sec:BL2}.
    \item ``AE'' (Abi Jaber, El Euch): This is the quadrature rule suggested in \cite[Section 4.2]{abi2019multifactor}.
    \item ``AK'' (Alfonsi, Kebaier): This is the quadrature rule suggested in \cite[Table 6, left column]{alfonsi2021approximation}, which seems to be the best (completely monotone) quadrature rule in \cite{alfonsi2021approximation}.
\end{enumerate}

We note that the algorithms OL2, AK, and BL2 focus on minimizing the $L^2$-error, while the algorithms GG, NGG, OL1, AE focus on the $L^1$-error. Since the $L^2$-norm of $K$ is infinite for $H\le 0$, the $L^2$-algorithms can only be applied for positive $H$. In contrast, the $L^1$-algorithms work for all $H > -1/2$. Perhaps surprisingly, BL2 still works for negative $H$ despite minimizing the $L^2$-norm, since we penalize large mean-reversions, which corresponds to a smoothing of the singularity of the kernel, see also Appendix \ref{sec:BL2}.

Throughout this section we compute option prices under rough Heston using Fourier inversion, since the characteristic function $\varphi$ of the log-stock price $\log(S_T)$ in \eqref{eqn:RHestonStock} is known in semi-closed form see e.g. \cite{abi2019affine}. The same is true for the characteristic function $\varphi^N$ of the Markovian approximation $\log(S^N_T)$ obtained by replacing $K$ by $K^N$, see e.g. \cite{abi2019markovian}. For the computation of $\varphi$, we need to solve a fractional Riccati equation. This is done using the Adam's scheme, see e.g. \cite{diethelm2004detailed}. To obtain $\varphi^N$, we need to solve an ordinary ($N$-dimensional) Riccati equation. This is done using a predictor-corrector scheme, as explained in \cite{abi2019multifactor, bayer2023markovian}. We remark that generally speaking, Fourier inversion of the Markovian approximation tends to be faster, as we only have to solve an ordinary (higher-dimensional) Riccati equation, which has a cost of $O(n)$, where $n$ is the number of time steps. In contrast, solving the fractional Riccati equation, which is a Volterra integral equation, has a cost of $O(n^2)$. This difference is especially pronounced if we need to compute prices with high accuracy, small maturity, or tiny (especially negative) Hurst parameter $H$.

This numerical section is structured as follows. First, we compare the computational cost of computing all these quadrature rules, as well as their largest nodes in Section \ref{sec:ListOfQuadratureRules}. Next, in Section \ref{sec:ComparingOL1WithOL2} we verify that the weak error can indeed be bounded by the $L^1$-error in the kernel as shown in Section \ref{sec:MainTheorySection}, rather than the $L^2$-error, by conducting some specific comparisons of the algorithms OL1 and OL2. In Section \ref{sec:L1ErrorInTheKernelNumerics} we numerically verify the convergence rates of the $L^1$-errors between $K^N$ and $K$ that we proved in Section \ref{sec:ErrorRates}, and compare these rates with the other algorithms. Finally, in Section \ref{sec:WeakConvergenceNumerics}, we compute various option prices and verify stylized facts and properties of the rough Heston model and its Markovian approximations. This includes implied volatility smiles in Section \ref{sec:ImpliedVolatilitySmiles}, implied volatility surfaces in Section \ref{sec:ImpliedVolatilitySurfaces}, implied volatility skews in Section \ref{sec:ImpliedVolatilitySkews}, and finally, prices of digital European call options in Section \ref{sec:DigitalOptions}, where we want to illustrate that the Markovian approximations still converge quickly despite the lack of regularity in the payoff function that is required in Section \ref{sec:MainTheorySection}.

\subsection{Computational Times and Largest Nodes}\label{sec:ListOfQuadratureRules}

Before discussing convergence rates below, let us compare two important quantities associated with these quadrature rules: computational time and the largest nodes. 

First, the time it took to compute the quadrature rules for the various algorithms was largely independent of the Hurst parameter $H$. For the quadrature rules GG, NGG, OL2, AE and AK, the computation is basically instant, taking at most a few ms. BL2 is still quite fast, with computational times ranging from a few ms (for $N=1$) to a few seconds (for $N=10$). Finally, OL1 is the most expensive, where computing $K^N$ for $N=10$ already takes about 15 min.

The largest node gives us some insight into how well the singularity of the kernel is captured. If a quadrature rule with small nodes performs very well, this may indicate that the singularity (and hence the roughness of the volatility) is not very important for that particular problem, and vice versa. Also, quadrature rules with small nodes are typically preferable over ones with large nodes (if they perform equally well), as larger nodes might lead to more problems of numerical stability and higher computational times. It will turn out that this is not a particularly severe problem for European options as we consider them, but it is much more relevant if one decides to simulate paths using the Markovian approximation, see also \cite{bayer2023++++simulation}.


The largest nodes are given in Table \ref{tab:LargestNodes}. We see that quadrature rules that aim to minimize the $L^2$-error (i.e. OL2, AK) reach very large nodes for small $H$. Compared to that, the largest nodes seem to remain bounded for $H\to 0$ for quadrature rules that aim to minimize the $L^1$-error (i.e. GG, NGG, OL1, AE). Interestingly, BL2 also has bounded nodes for small $H$, with largest nodes of similar sizes as for OL1.

\begin{table}
    \centering
    \resizebox{\textwidth}{!}{\begin{tabular}{c|c|c|c|c|c|c|c|c|c|c|c|c|c|c|c|c|c|c|c}
         & \multicolumn{5}{c|}{$H=-0.1$} & \multicolumn{7}{c|}{$H=0.001$} & \multicolumn{7}{c}{$H=0.1$}\\ \hline
        N &  GG  &  NGG  & OL1         & BL2  &  AE
          &  GG  &  NGG  & OL1  & OL2  & BL2  &  AE   &  AK 
          &  GG  &  NGG  & OL1  & OL2  & BL2  &  AE   &  AK\\ \hline
        1 & 0.18 &  0.05 & 0.24        & 0.24 & -0.46
          & 0.12 &  0.00 & 0.14 & 0.87 & 0.87 & -0.53 & - 
          & 0.06 & -0.07 & 0.02 & 0.34 & 0.34 & -0.62 & -\\
        2 & 1.17 &  0.95 & 1.40        & 1.35 &  0.08
          & 1.02 &  0.95 & 1.28 & 6.83 & 1.43 &  0.05 & 3.32 
          & 0.92 &  0.95 & 1.18 & 2.56 & 0.94 &  0.03 & 1.48\\
        3 & 1.59 &  1.70 & 2.24        & 2.08 &  0.27
          & 1.39 &  1.70 & 2.07 & 12.2 & 1.83 &  0.25 & 3.32
          & 1.25 &  1.70 & 1.93 & 4.28 & 1.67 &  0.22 & 1.48\\
        4 & 1.94 &  2.49 & 2.98        & 2.80 &  0.39
          & 1.70 &  2.49 & 2.75 & 17.3 & 2.43 &  0.37 & 7.10 
          & 1.58 &  2.49 & 2.57 & 5.77 & 2.24 &  0.34 & 3.13\\
        5 & 2.24 &  3.32 & 3.66        & 3.45 &  0.48
          & 2.02 &  3.32 & 3.35 & 36.7 & 3.02 &  0.46 & 7.10
          & 1.81 &  1.09 & 3.14 & 7.11 & 2.81 &  0.43 & 3.13\\
        6 & 2.57 &  4.16 & 4.24        & 4.10 &  0.55
          & 2.26 &  1.86 & 3.91 & 44.0 & 3.52 &  0.53 & 11.5 
          & 2.04 &  1.86 & 3.65 & 8.34 & 3.32 &  0.50 & 4.40\\
        7 & 2.82 &  2.66 & 4.79        & 4.67 &  0.61
          & 2.48 &  2.66 & 4.42 & 31.1 & 4.15 &  0.59 & 11.5 
          & 2.24 &  2.66 & 4.28 & 9.48 & 3.77 &  0.56 & 4.40\\
        8 & 3.04 &  2.66 & 5.32        & 5.17 &  0.66
          & 2.68 &  2.66 & 4.90 & 35.5 & 4.54 &  0.64 & 15.5 
          & 2.42 &  2.66 & 4.58 & 10.5 & 4.52 &  0.61 & 5.45\\
        9 & 3.24 &  2.66 & 6.03        & 5.62 &  0.70
          & 2.86 &  2.66 & 5.35 & 39.7 & 4.90 &  0.69 & 15.5
          & 2.58 &  2.66 & 5.20 & 11.6 & 4.90 &  0.66 & 5.45\\
        10 & 3.44 &  3.49 & 6.65        & 6.03 &  0.74
          & 3.04 &  3.49 & 6.02 & 43.9 & 5.25 &  0.72 & 20.0
          & 2.75 &  3.49 & 5.74 & 12.5 & 5.27 &  0.70 & 6.35
    \end{tabular}}
    \caption{Largest nodes of the quadrature rules with maturity $T=1$. Since the largest nodes increase considerably with $N$, we give the logarithm with base 10 of the largest node (i.e. a value of 3 in the table corresponds to a largest node of $10^3 = 1000$).}
    \label{tab:LargestNodes}
\end{table}

\subsection{Comparing OL1 with OL2}\label{sec:ComparingOL1WithOL2}

We have shown in Section \ref{sec:MainTheorySection} that the weak error of the Markovian approximations of the rough Heston model can be bounded by the $L^1$-error in the kernel. This suggests, in contrast to e.g. the work in \cite{bayer2023markovian} that if we want to price options using the Markovian approximations, we may want to focus on minimizing the $L^1$-error in the kernel, instead of the $L^2$-error, i.e. should use the algorithm OL1 rather than OL2. The difference becomes especially pronounced for very small $H$, as in that case, $K(t) = t^{H-1/2}$ is barely square-integrable. Furthermore, weak formulations of rough Heston are possible in the regime of $H \in (-1/2, 0]$, see e.g. \cite{abi2021weak, abi2023reconciling, jusselin2020no}.

In the following figures we aim to show that the error of approximating European call option implied volatility smiles indeed decreases rapidly as the number $N$ of nodes in the kernel $K^N$ increases, if those nodes are chosen such that the $L^1$-error in the kernel is minimized. We compare these results with using the kernel $K^N$ resulting from minimization of the $L^2$-error.

Throughout, we consider European call options with the same parameters as in \cite[Section 4.2]{abi2019multifactor}, i.e. $$S_0 = 1,\ V_0=0.02,\ \theta=0.02,\ \lambda=0.3,\ \nu=0.3,\ \rho=-0.7,\ T=1.$$ We choose 201 values for the log-moneyness linearly spaced in the interval $[-1.5, 0.75]$. The implied volatility smiles of the true (i.e. non-approximated) rough Heston model are computed using Fourier inversion with the Adams scheme, see e.g. \cite{diethelm2004detailed}. The implied volatility smiles for the Markovian approximations are computed using a similar predictor-corrector scheme, see e.g. \cite{abi2019multifactor, bayer2023markovian}. All implied volatility smiles were computed with a maximal relative error of $1.5\cdot 10^{-5}$.

The left plot of Figure \ref{fig:ImpliedVolatilitySmiles} compares the implied volatility smiles for the European call option with a tiny $H=0.001$, and using $N=4$ nodes. The Markovian approximation using the $L^1$-optimized kernel (blue) is directly over the true rough Heston smile (black). The approximation using the $L^2$-optimized kernel (red) is still a bit off. Moreover, it will converge only very slowly to the black line, as is illustrated in the right plot of Figure \ref{fig:ImpliedVolatilitySmiles}. Here, we see how the errors in the implied volatility smiles decrease very rapidly for the $L^1$-approximation as $N$ increases, even for tiny (or negative) $H$. In contrast, while the $L^2$-approximation still achieves good results for $H$ large enough, we see that the convergence rate is very slow for tiny $H$. The precise values of the errors are given in Table \ref{tab:ImpliedVolatilityErrors}.

\begin{figure}
\centering
\begin{minipage}{.49\textwidth}
  \centering
  \includegraphics[width=\linewidth]{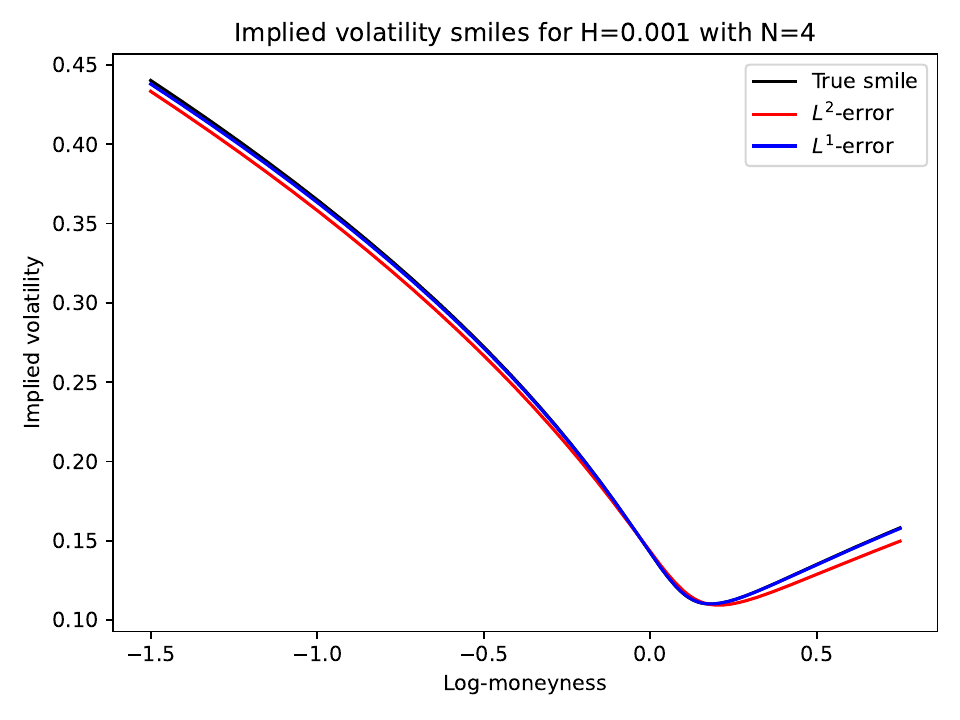}
\end{minipage}%
\begin{minipage}{.49\textwidth}
  \centering
  \includegraphics[width=\linewidth]{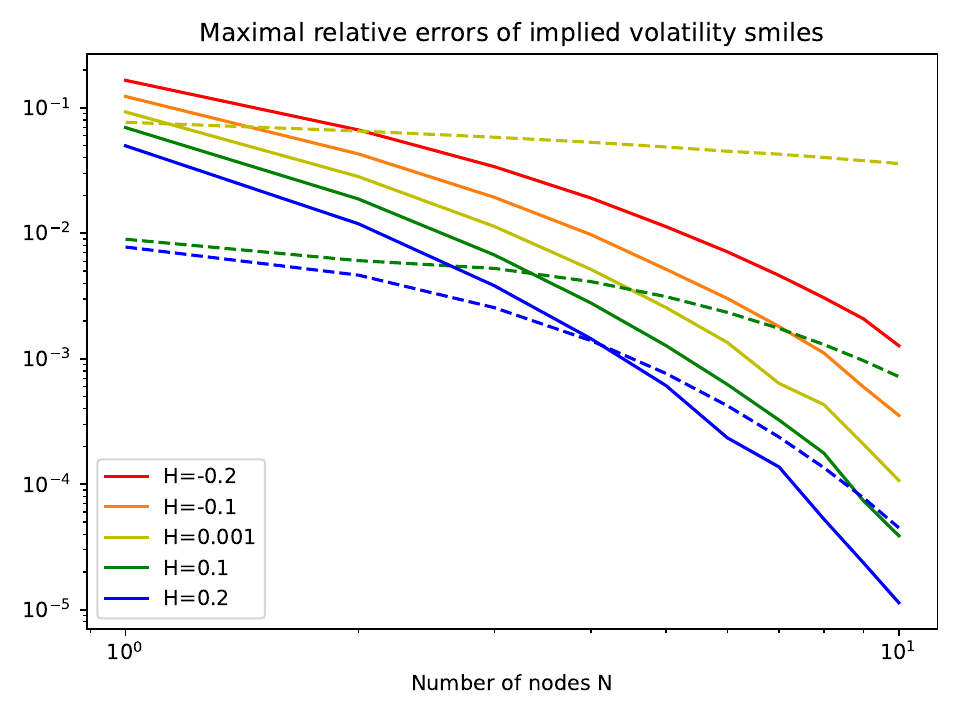}
\end{minipage}
\caption{Left: Implied volatility smiles for European call options with tiny $H=0.001$. The truth (black) is barely visible under the $L^1$-approximation (blue). The $L^2$-approximation (red) is still a bit off. Right: Maximal relative errors of the implied volatility smiles for different $H$ and $N$, using kernel approximations $K^N$ minimizing the $L^1$-error (solid) or the $L^2$-error (dashed).}
\label{fig:ImpliedVolatilitySmiles}
\end{figure}

\begin{table}[!htbp]
\centering
\resizebox{\textwidth}{!}{\begin{tabular}{c|c|c|c|c|c|c|c|c}
 & \multicolumn{5}{c|}{$K^N$ minimizing $\|K - K^N\|_{L^1}$} & \multicolumn{3}{c}{$K^N$ minimizing $\|K - K^N\|_{L^2}$}\\ \hline
$N$ & $H=-0.2$ & $H=-0.1$ & $H=0.001$ & $H=0.1$ & $H=0.2$ & $H=0.001$ & $H=0.1$ & $H=0.2$\\ \hline
1 & 16.58 & 12.33 & 9.281 & 6.961 & 4.986 & 7.678 & 0.897 & 0.778\\
2 & 6.675 & 4.297 & 2.833 & 1.879 & 1.192 & 6.541 & 0.606 & 0.464\\
3 & 3.396 & 1.935 & 1.133 & 0.672 & 0.382 & 5.835 & 0.525 & 0.255\\
4 & 1.908 & 0.975 & 0.515 & 0.278 & 0.144 & 5.317 & 0.412 & 0.139\\
5 & 1.130 & 0.517 & 0.256 & 0.127 & 0.061 & 4.874 & 0.313 & 0.076\\
6 & 0.713 & 0.303 & 0.135 & 0.062 & 0.023 & 4.513 & 0.234 & 0.042\\
7 & 0.462 & 0.181 & 0.064 & 0.032 & 0.014 & 4.276 & 0.175 & 0.024\\
8 & 0.306 & 0.111 & 0.043 & 0.018 & 0.005 & 4.022 & 0.130 & 0.014\\
9 & 0.208 & 0.060 & 0.021 & 0.007 & 0.002 & 3.798 & 0.097 & 0.008\\
10 & 0.127 & 0.035 & 0.011 & 0.004 & 0.001 & 3.596 & 0.072 & 0.004
\end{tabular}}
\caption{Maximal relative errors in $\%$ of the European call smiles for different choices of $N$ and $H$. The values in this table have an error of at most $0.003$.}
\label{tab:ImpliedVolatilityErrors}
\end{table}

Finally, we want to demonstrate that the error of the implied volatility smiles indeed converges at the same rate as the $L^1$-error of the kernels. To this end, we fix the Hurst parameter $H=0.1$, and compute the $L^1$-errors and the $L^2$-errors of the kernels, as well as the maximal errors of the volatility smiles for kernels $K^N$ minimizing the $L^1$-error $\|K - K^N\|_{L^1}$ or the $L^2$-error $\|K - K^N\|_{L^2}$. The results are illustrated in Figure \ref{fig:ConvergenceRates}. Note that the errors of the implied volatility smiles decrease roughly at the same speed as the $L^1$-errors of the kernels, consistent with our theoretical results.

\begin{figure}
\centering
\begin{minipage}{.49\textwidth}
  \centering
  \includegraphics[width=\linewidth]{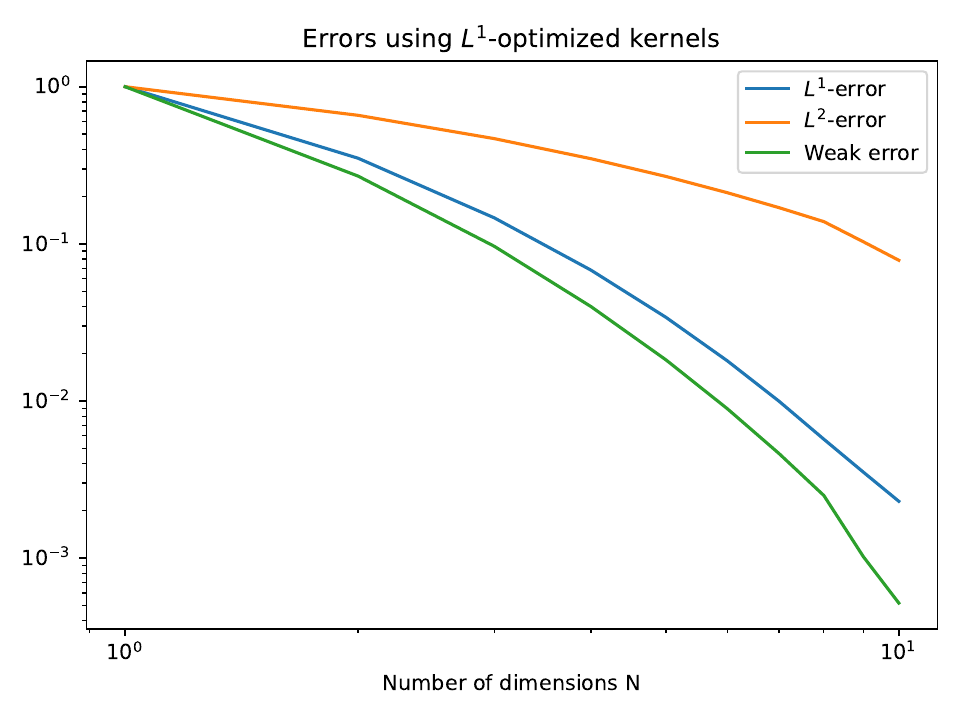}
\end{minipage}%
\begin{minipage}{.49\textwidth}
  \centering
  \includegraphics[width=\linewidth]{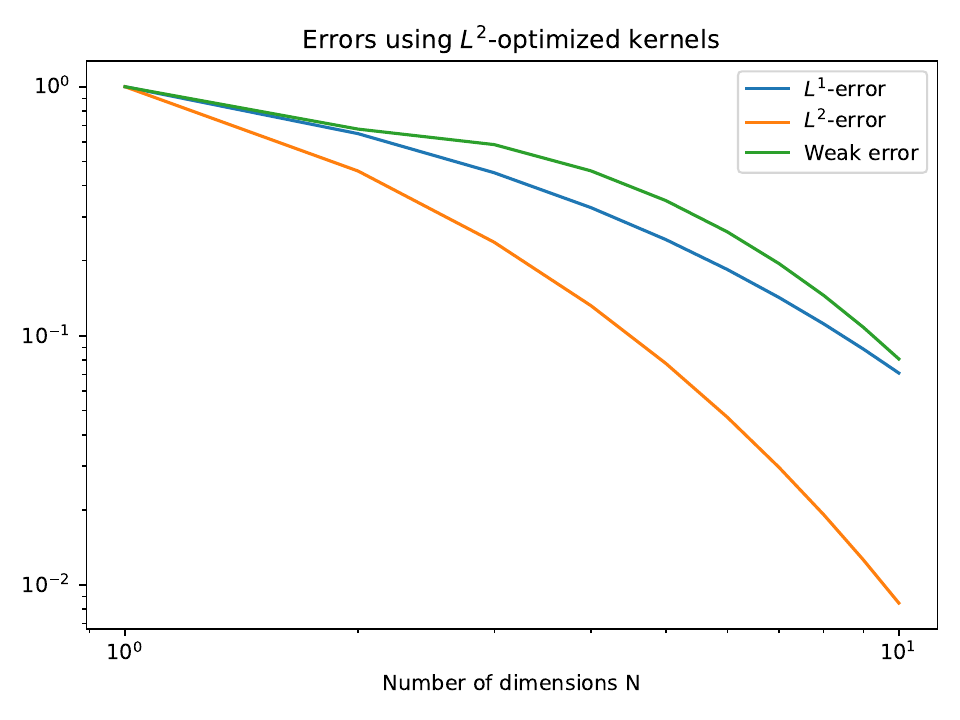}
\end{minipage}
\caption{$L^1$- and $L^2$-errors of the kernels, and errors of implied volatility smiles. Due to the different magnitudes of the errors, they were rescaled to be equal to $1$ for $N=1$. In the left figure we use kernels $K^N$ minimizing the $L^1$-error $\|K - K^N\|_{L^1}$, in the right figure kernels $K^N$ minimizing the $L^2$-error $\|K - K^N\|_{L^2}$.}
\label{fig:ConvergenceRates}
\end{figure}

\subsection{$L^1$-error in the kernel}\label{sec:L1ErrorInTheKernelNumerics}

In this section, we want to numerically verify our that the $L^1$-errors of GG and NGG converge super-polynomially, and that their convergence rate does not become arbitrarily bad for small $H$. 
 
In the left plot of Figure \ref{fig:L1ErrorsGGNGG} we see the relative $L^1$-errors for GG and NGG for $H=0.1$ and $T=1$. We clearly see that both exhibit super-polynomial convergence. Furthermore, this figure includes lines representing the convergence rates in Theorems \ref{thm:TheL1TheoremGeometric} and \ref{thm:TheL1TheoremNonGeometric} (with suitably fitted leading constants), showing that the errors converge at the rate that we proved there. The right plot of Figure \ref{fig:L1ErrorsGGNGG} shows that the convergence rates do not become arbitrarily bad for tiny $H>0$, and only significantly worsens once we approach $H\approx -1/2$. We remark that the error lines are wiggly due to the discrete choice of the level $m$ of the Gaussian quadrature rule.

Finally, for small $N$, we compare GG with OL1. In Figure \ref{fig:L1ErrorsVaryingHWithOL1} we show that there still is significant room for improvement. The figure only goes up to $N=10$ because of the high computational cost of computing $L^1$-errors, cf. Section \ref{sec:ListOfQuadratureRules}.

\begin{figure}
    \centering
    \begin{minipage}{.5\textwidth}
        \centering
        \includegraphics[width=\linewidth]{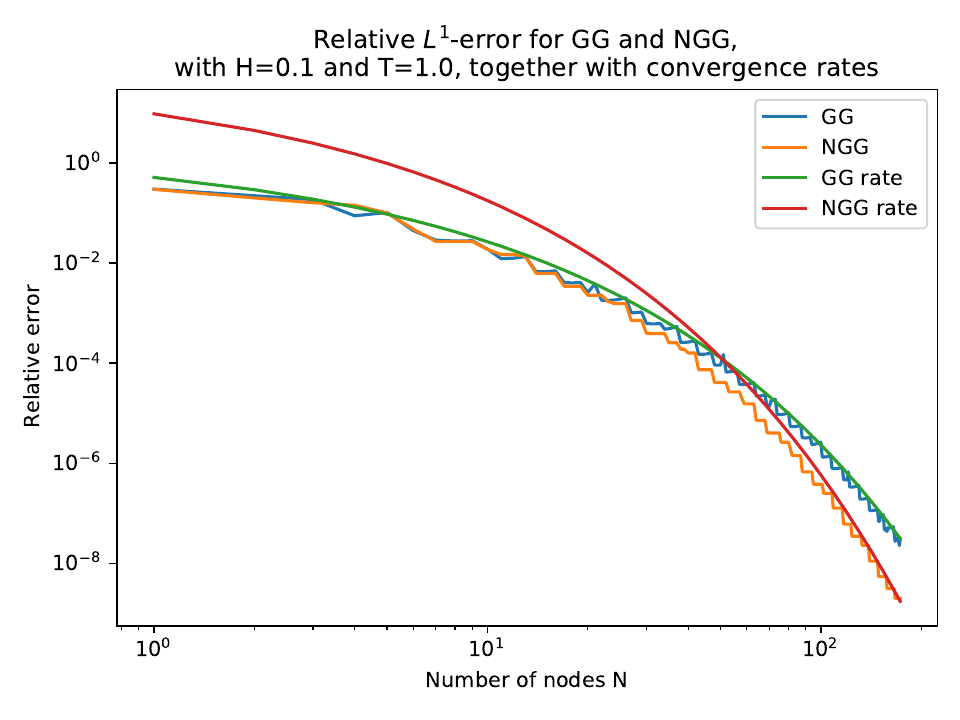}
    \end{minipage}%
    \begin{minipage}{.5\textwidth}
        \centering
        \includegraphics[width=\linewidth]{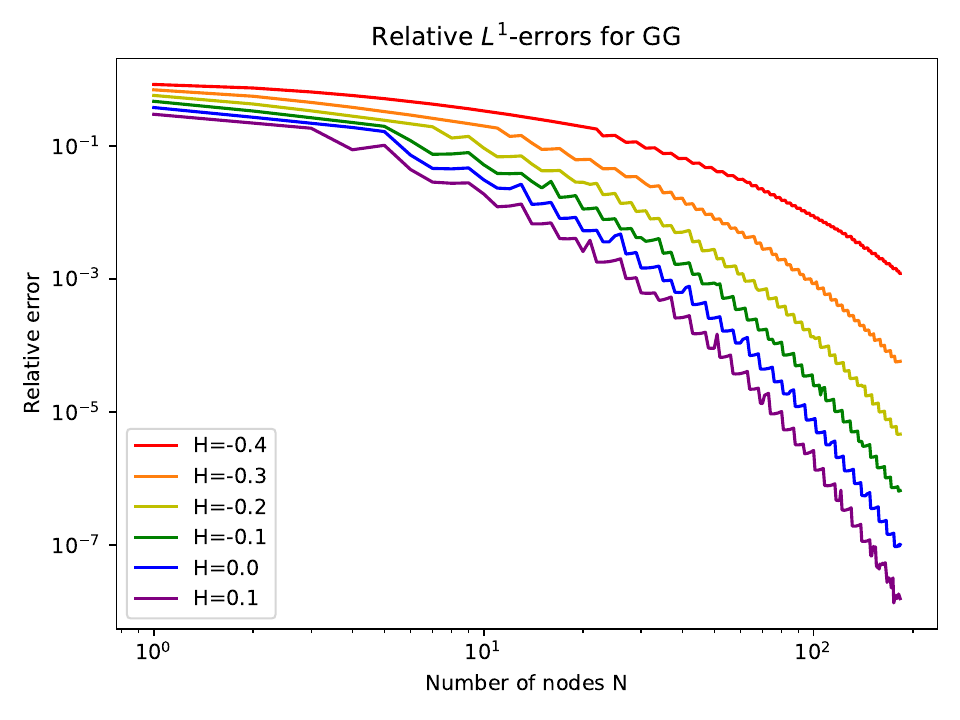}
    \end{minipage}
    \caption{Left: Relative $L^1$-errors of GG and NGG for $H=0.1$ and $T=1$. The green line is the function $2\left(\sqrt{2} + 1\right)^{-2\sqrt{(H+1/2)N}}$ and the red line is the function $60 e^{-2.38\sqrt{(H+1/2)N}}$, corresponding to the error bounds (with fitted leading constants) in Theorems \ref{thm:TheL1TheoremGeometric} and \ref{thm:TheL1TheoremNonGeometric}, respectively. Right: Relative $L^1$-errors for GG for varying $H$ and $T=1$.}
    \label{fig:L1ErrorsGGNGG}
\end{figure}

\begin{figure}
    \centering
    \includegraphics[scale=0.6]{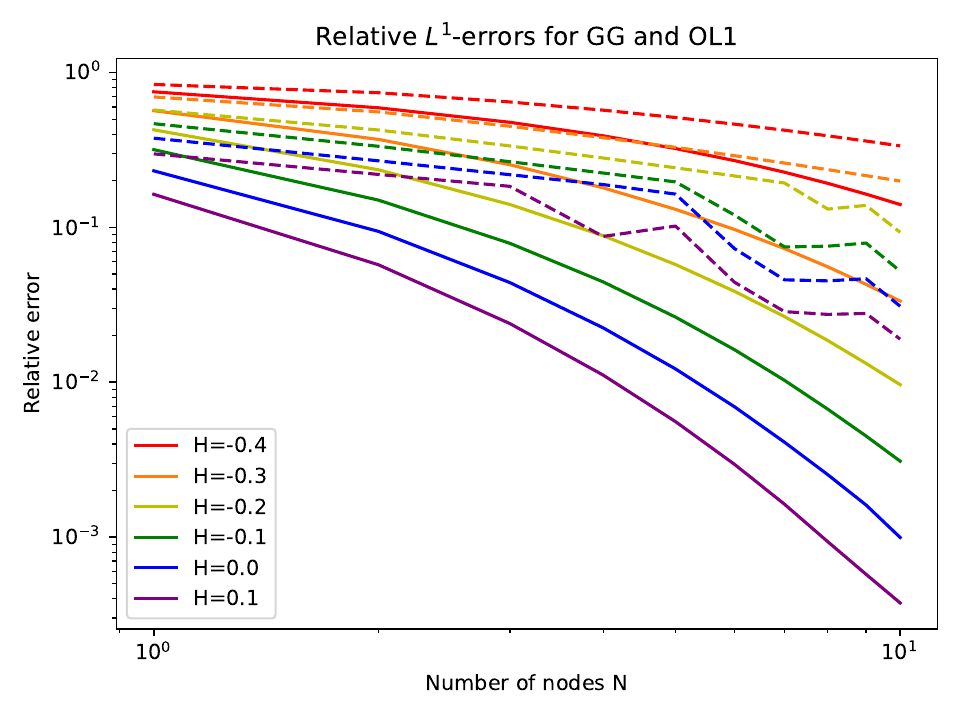}
    \caption{Relative $L^1$-errors for GG (dashed lines) and OL1 (solid lines) for $T=1$ and varying $H$.}
    \label{fig:L1ErrorsVaryingHWithOL1}
\end{figure}

\subsection{Weak convergence of the stock price}\label{sec:WeakConvergenceNumerics}

In this section we numerically investigate the weak convergence of the Markovian approximations of the rough Heston model. Throughout, we use the same parameters as in \cite[Section 4.2]{abi2019multifactor}, i.e. $$\lambda = 0.3,\ \theta = 0.3\cdot 0.02,\ \rho=-0.7,\ \nu=0.3,\ V_0=0.02,\ S_0=1.$$ 

As mentioned before, we carry out all these computations using Fourier inversion. This introduces discretization errors for both the approximation of the Fourier inversion integral, and the numerical solution of the (fractional) Riccati equations needed for the characteristic function. Since we are not interested in the behaviour of these discretization errors, we will always try to keep them very small. Our implementation of the Fourier inversion (both the Adams scheme and the Markovian approximations) takes an error tolerance $\tol$ and returns option prices or implied volatilities that were computed with a combined relative error of the Fourier inversion and the Riccati equations that is less than $\tol$. The specific choice of $\tol$ will vary depending on the option we consider, and will always be stated explicitly. 

\subsubsection{Implied volatility smiles for European call options}\label{sec:ImpliedVolatilitySmiles}

First, let us consider call options with 301 different values of log-moneyness linearly spaced in the interval $[-1, 0.5]\cdot\sqrt{T}$, where we set the maturity $T=0.01$, and $H = -0.1, 0.001, 0.1$. We compute these smiles using Fourier inversion with a relative accuracy of at least $10^{-5} = 0.001\%$ for $H=0.001, 0.1$, and $10^{-4} = 0.010\%$ for $H=-0.1$ (due to the higher computational cost associated with smaller $H$). 

Indeed, especially for high-accuracy computations with tiny $T$ or $H$, the computational times can become very large. However, the Adams scheme is more affected by this than the Markovian approximations. For example, the computation of the smile for $H=-0.1$ with relative accuracy $10^{-4}$ took 6183 seconds using the Adams scheme, while the Markovian approximations took $40$ -- $270$ seconds, depending on $N$ and the quadrature rule used. This is because in both cases, we need to solve a Riccati equation to compute the characteristic function, and this is a fractional Riccati equation for the Adams scheme, while it is an ordinary ($N$-dimensional) Riccati equation for the Markovian approximations. The numerical solution of a fractional Riccati equation has a computational cost of $O(n^2)$ for $n$ time steps, compared to $O(n)$ for ordinary Riccati equations. This difference becomes especially pronounced for large $n$, which is needed for high accuracy, and when $T$ and $H$ are small.

In Figure \ref{fig:IVSmileGGBL2HMinus01} we see the approximations using GG and BL2 quadrature rules for the implied volatility smile with $H=-0.1$. We see that already a very small number of nodes is sufficient to achieve a high accuracy. Indeed, for BL2 already $N=2$ is sufficiently accurate that we cannot see the any difference to the exact, non-approximated smile anymore without zooming in.

In Table \ref{tab:ErrorsIVSmiles} we see the maximal relative errors of the implied volatility smiles in $\%$. As expected, quadrature rules relying on the $L^2$-error (i.e. OL2, AK) perform very badly for tiny $H$ (and are not well-defined for negative $H$), while quadrature rules relying on the $L^1$-error (i.e. GG, NGG, OL1, AE) still work well even for $H=-0.1$. The exception here is BL2, which, despite optimizing the (penalized) $L^2$-error outperforms every other method by multiple magnitudes, even for tiny or negative $H$. We study this phenomenon a bit further in Appendix \ref{sec:BL2}. Aside from that, both GG and NGG seem to clearly outperform AE, with GG yielding better results. We guess that for sufficiently large $N$, NGG will outperform GG, but such large $N$ are likely not necessary for usual applications. Also, there is still a lot of room for improvement, as OL1 (and of course BL2) show.

In Figure \ref{fig:ErrorsIVSmiles} we compare all the methods. In the left figure, we see that algorithms using the $L^2$-error show almost no convergence, while BL2 significantly outperforms all other algorithms. In the right figure, we see that the error does not deteriorate too much for GG, OL1, and BL2 as $H$ becomes small or negative, while it does for OL2.

\begin{figure}
    \centering
    \begin{minipage}{.5\textwidth}
        \centering
        \includegraphics[width=\linewidth]{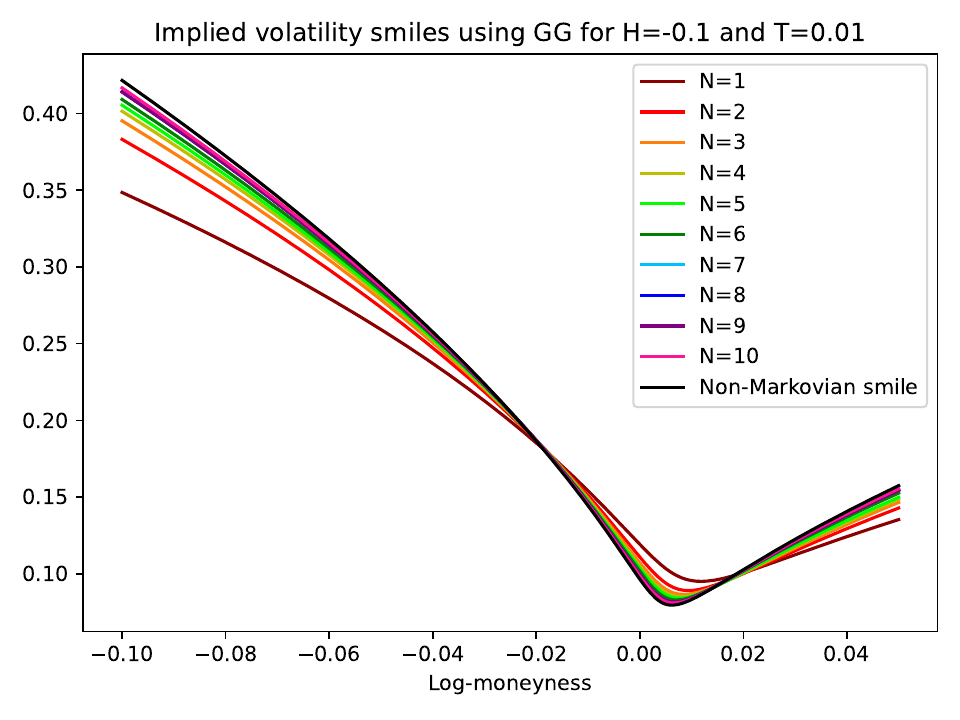}
    \end{minipage}%
    \begin{minipage}{.5\textwidth}
        \centering
        \includegraphics[width=\linewidth]{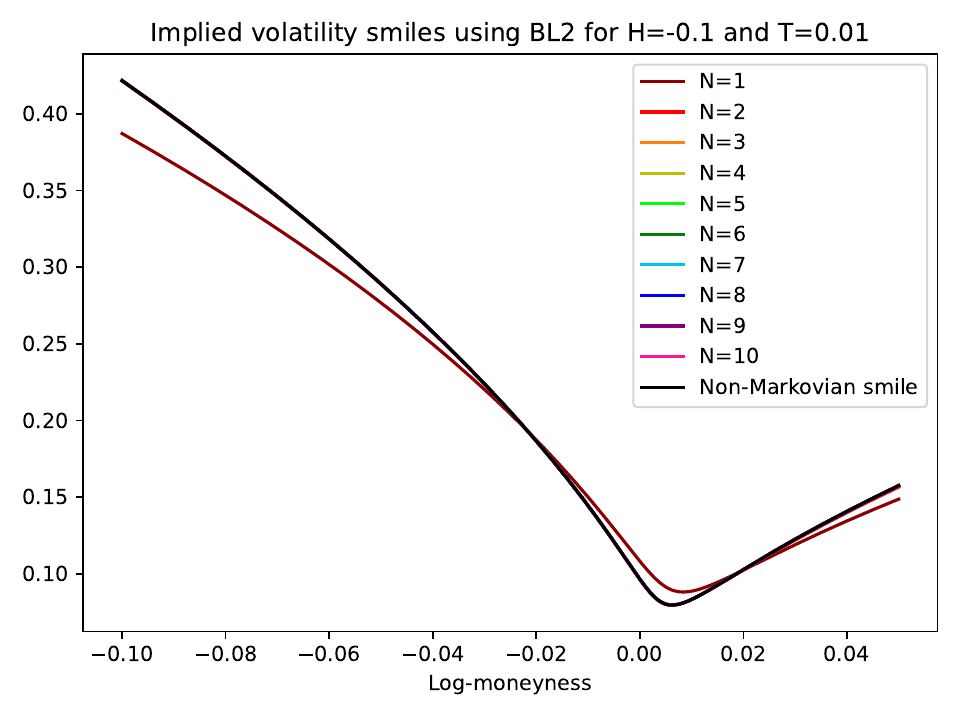}
    \end{minipage}
    \caption{Implied volatility smiles for $T=0.01$ and $H=-0.1$, together with the Markovian approximations using GG (left) and BL2 (right).}
    \label{fig:IVSmileGGBL2HMinus01}
\end{figure}

\begin{table}
\centering
\resizebox{\textwidth}{!}{\begin{tabular}{c|c|c|c|c|c|c|c|c|c|c|c|c|c|c|c|c|c|c|c}
     & \multicolumn{5}{c|}{$H=-0.1$} & \multicolumn{7}{c|}{$H=0.001$} & \multicolumn{7}{c}{$H=0.1$} \\ \hline
    N &  GG   &  NGG  &  OL1         &  BL2  &  AE
      &  GG   &  NGG  &  OL1 &  OL2  &  BL2  &  AE   & AK
      &  GG   &  NGG  &  OL1 &  OL2  &  BL2  &  AE   & AK\\ \hline
    1 & 29.93 & 31.86 & 14.93        & 14.93 & 49.44
      & 18.29 & 19.38 & 8.538 & 8.315 & 8.315 & 31.40 & -
      & 13.43 & 14.77 & 7.348 & 0.894 & 0.894 & 24.36 & -\\
    2 & 18.35 & 21.82 & 5.526         & 0.591 & 39.77
      & 11.55 & 14.16 & 2.818 & 7.133 & 0.223 & 25.07 & 8.138 
      & 8.288 & 10.67 & 2.050 & 0.650 & 0.442 & 20.07 & 1.112\\
    3 & 13.35 & 23.46 & 2.527         & 0.112 & 34.46
      & 8.704 & 15.75 & 1.147 & 6.394 & 0.101 & 21.56 & 8.138 
      & 6.017 & 12.31 & 0.738 & 0.553 & 0.066 & 17.60 & 1.112\\
    4 & 10.58 & 20.14 & 1.281         & 0.012 & 31.01
      & 7.066 & 13.16 & 0.525 & 5.848 & 0.007 & 19.28 & 24.38 
      & 4.405 & 9.812 & 0.306 & 0.427 & 0.005 & 15.92 & 2.004\\
    5 & 8.802 & 15.84 & 0.681         & 0.002 & 28.56
      & 7.599 & 10.74 & 0.262 & 5.395 & 0.001 & 17.65 & 24.38 
      & 5.058 & 6.501 & 0.140 & 0.319 & 0.001 & 14.67 & 2.004\\
    6 & 6.109 & 13.38 & 0.399         & 0.000 & 26.72
      & 3.161 & 10.83 & 0.138 & 5.051 & 0.000 & 16.41 & 29.50
      & 2.121 & 9.107 & 0.069 & 0.236 & 0.000 & 13.68 & 1.659\\
    7 & 3.707 & 11.78 & 0.238         & 0.000 & 25.27
      & 1.965 & 7.282 & 0.077 & 4.740 & 0.000 & 15.43 & 29.50 
      & 1.371 & 5.525 & 0.029 & 0.174 & 0.000 & 12.88 & 1.659\\
    8 & 3.697 & 11.78 & 0.147         & 0.001 & 24.07
      & 1.898 & 7.282 & 0.044 & 4.469 & 0.000 & 14.63 & 31.27 
      & 1.245 & 5.525 & 0.019 & 0.128 & 0.000 & 12.21 & 0.999\\
    9 & 3.844 & 11.78 & 0.079         & 0.002 & 23.07
      & 1.932 & 7.282 & 0.026 & 4.228 & 0.000 & 13.95 & 31.27 
      & 1.206 & 5.525 & 0.008 & 0.094 & 0.000 & 11.64 & 0.999\\
    10 & 2.482 & 7.052 & 0.047         & 0.003 & 22.22 
      & 1.263 & 4.476 & 0.012 & 4.010 & 0.000 & 13.38 & 30.54
      & 0.804 & 3.414 & 0.004 & 0.070 & 0.000 & 11.15 & 0.503
\end{tabular}}
\caption{Maximal relative errors in $\%$ for the Markovian approximations for implied volatility smiles of the European call option. The discretization error of these errors is at most $0.002\%$ (for $H=0.1$ and $H=0.001$) and $0.020\%$ (for $H=-0.1$).}
\label{tab:ErrorsIVSmiles}
\end{table}

\begin{figure}
    \centering
    \begin{minipage}{0.5\textwidth}
        \centering
        \includegraphics[width=\linewidth]{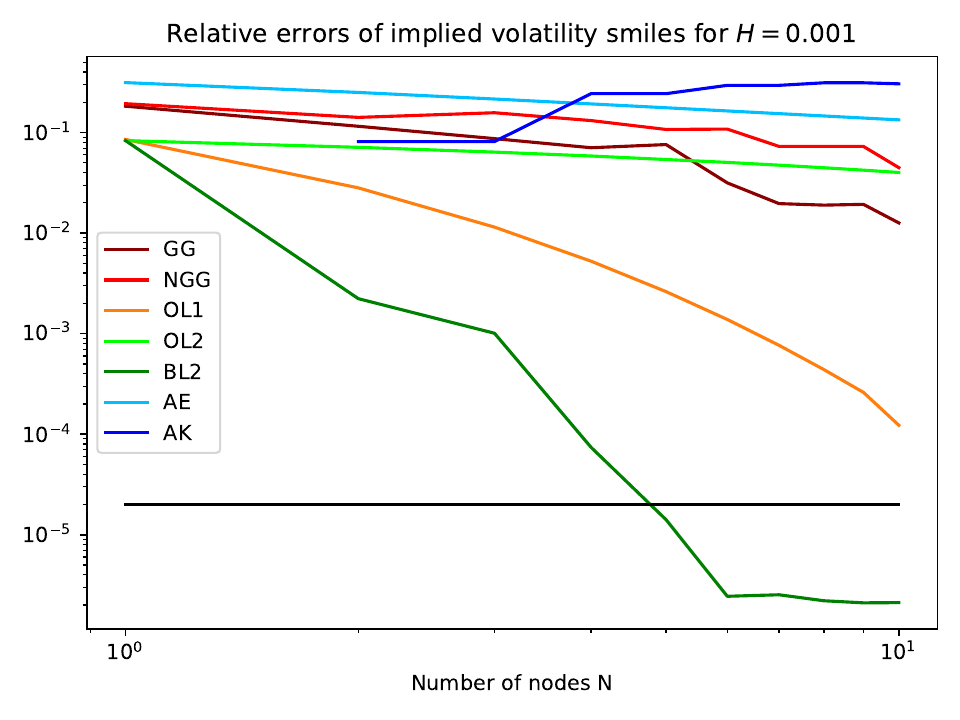}
    \end{minipage}%
    \begin{minipage}{0.5\textwidth}
        \centering
        \includegraphics[width=\linewidth]{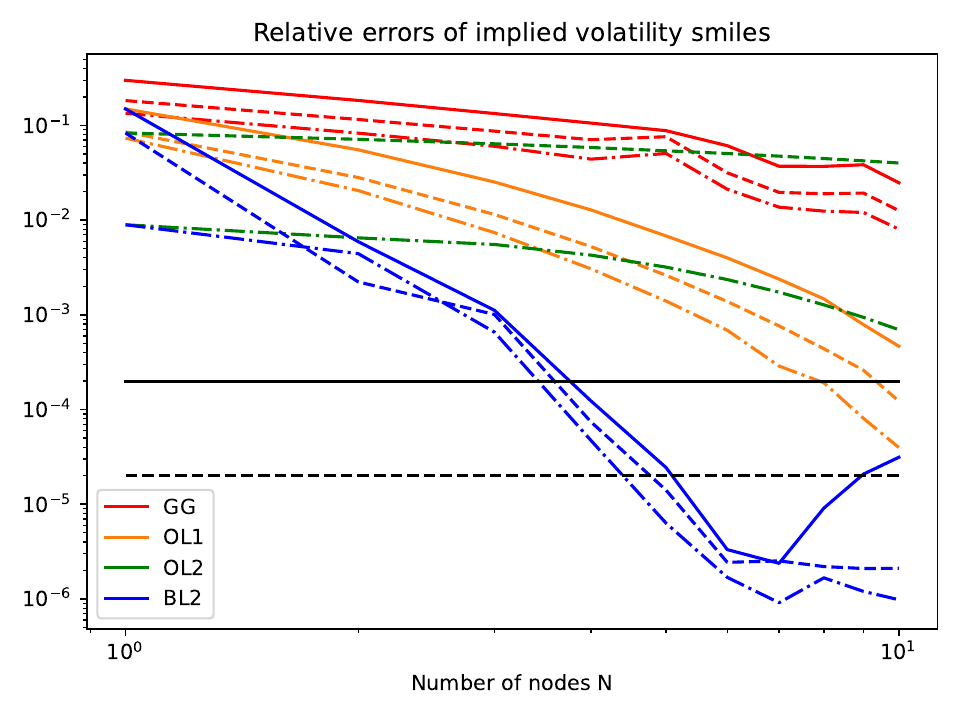}
    \end{minipage}
    \caption{Maximal relative errors of the implied volatility smiles for various quadrature rules. The figure on the left compares all 7 algorithms for $H=0.001$. The figure on the right compares GG, OL1, OL2, BL2 for $H=-0.1$ (solid lines), $H=0.001$ (dashed lines), and $H=0.1$ (dash-dotted lines). The black lines are the accuracy of the computation of the relative errors (due to errors in the Fourier inversion and the solution of the Riccati equations), $2\cdot 10^{-5}$ for $H\ge 0$ and $2\cdot 10^{-4}$ for $H < 0$.}
    \label{fig:ErrorsIVSmiles}
\end{figure}

\subsubsection{Implied volatility surfaces for European call options}\label{sec:ImpliedVolatilitySurfaces}

Next, we consider implied volatility surfaces, using 25 maturities linearly spaced in $[0.04, 1]$, where we define $T_{\textup{min}} \coloneqq 0.04$ and $T_{\textup{max}} \coloneqq 1$. For each maturity $T$ we take 301 linearly spaced values of log-moneyness in the interval $[-1, 0.5] \cdot \sqrt{T}$. 

Of course, despite computing implied volatilities for multiple maturities, we use the same quadrature rule for all of them, i.e. we do not adapt the quadrature rule to each individual maturity. But the quadrature rules above all require us to consider $K$ on some interval $[0, T_0]$. Thus, the question arises how we choose $T_0$. A natural choice would be $T_0 = T_{\max}$, the maximal maturity. This is also a very reasonable choice for large $N$, since if we chose $T_0<T_{\max}$, we could not expect a good approximation on the interval $[T_0, T_{\max}]$. However, when $N$ is very small, choosing $T_0 = T_{\max}$ yields very good results for maturities $T \approx T_{\max}$, while giving very bad approximations for maturities $T \ll T_{\max}$. After some numerical experiments, the choice
\begin{equation}\label{eqn:TVectorChoice}
T_0 = T_{\min}^{\alpha(N)} T_{\max}^{1 - \alpha(N)},
\end{equation}
where $$\alpha(1) = \frac{3}{5},\ \alpha(2) = \frac{1}{2},\ \alpha(3) = \frac{1}{3},\ \alpha(4) = \frac{1}{4},\ \alpha(5) = \frac{1}{6},\ \alpha(6) = \frac{1}{10},\ \alpha(N) = 0\ (N\ge 7),$$ seemed to yield good results. 

We computed all surfaces with a relative error tolerance of $10^{-5}$ for $H=0.001,0.1$ and $10^{-4}$ for $H=-0.1$. The computational times for the Adams scheme was 1951 seconds for $H=-0.1$, 1598 seconds for $H=0.001$, and $529.7$ seconds for $H=0.1$. In contrast, typical computational times for the Markovian approximations were between 30 and 100 seconds.

The maximal relative errors (maximum over both the strikes and maturities) are reported in Table \ref{tab:ErrorsIVSurfaces}. The results are largely similar to Table \ref{tab:ErrorsIVSmiles}, although the errors generally tend to be a bit larger since we are now jointly approximating a volatility surface rather than just a single smile. We note that overall, BL2 again performed best, achieving errors around $1\%$ for $N=3$ and errors below the discretization error of $0.002\%$ were reached for $N=9$ largely independent of $H$.

In Figure \ref{fig:ErrorsIVSurfaces} we see the maximal relative errors for $H=0.001$ for all quadrature rules on the left, and a comparison for varying $H$ on the right.

\begin{table}
\centering
\resizebox{\textwidth}{!}{\begin{tabular}{c|c|c|c|c|c|c|c|c|c|c|c|c|c|c|c|c|c|c|c}
     & \multicolumn{5}{c|}{$H=-0.1$} & \multicolumn{7}{c|}{$H=0.001$} & \multicolumn{7}{c}{$H=0.1$} \\ \hline
    N &  GG   &  NGG  &  OL1          &  BL2  &  AE
      &  GG   &  NGG  &  OL1  &  OL2  &  BL2  &  AE   &  AK 
      &  GG   &  NGG  &  OL1  &  OL2  &  BL2  &  AE   &  AK\\ \hline
    1 & 36.76 & 40.53 & 23.21         & 23.21 & 52.90
      & 26.71 & 29.43 & 19.54 & 23.78 & 23.78 & 40.62 & -
      & 21.42 & 23.68 & 16.53 & 17.75 & 17.75 & 33.85 & -\\
    2 & 21.84 & 23.04 & 9.025         & 7.255 & 51.21  
      & 17.34 & 16.35 & 8.027 & 19.04 & 4.223 & 38.77 & 17.02 
      & 15.58 & 13.63 & 7.787 & 13.30 & 6.111 & 32.23 & 3.945\\
    3 & 18.14 & 18.14 & 4.975         & 0.806 & 52.76  
      & 15.12 & 12.87 & 3.203 & 14.52 & 0.948 & 40.45 & 19.10 
      & 14.22 & 11.09 & 2.632 & 6.808 & 1.012 & 33.71 & 4.089\\
    4 & 14.68 & 13.81 & 2.881         & 0.644 & 53.02 
      & 12.26 & 9.995 & 1.583 & 11.88 & 1.044 & 40.74 & 24.45
      & 12.27 & 8.611 & 1.233 & 3.679 & 0.807 & 33.07 & 2.426\\
    5 & 12.38 & 11.15 & 1.620         & 0.176 & 53.61 
      & 12.43 & 8.213 & 0.900 & 14.55 & 0.274 & 41.42 & 26.41
      & 11.67 & 19.72 & 0.676 & 1.505 & 0.194 & 34.57 & 2.153\\
    6 & 11.12 & 9.455 & 1.043         & 0.012 & 54.04 
      & 8.618 & 11.90 & 0.531 & 12.52 & 0.056 & 41.94 & 26.59
      & 8.732 & 9.581 & 0.372 & 0.623 & 0.036 & 35.03 & 1.218\\
    7 & 8.756 & 10.84 & 0.720         & 0.012 & 55.16 
      & 7.406 & 6.733 & 0.338 & 16.00 & 0.012 & 43.36 & 26.98
      & 7.845 & 5.410 & 0.188 & 0.414 & 0.008 & 36.31 & 1.473\\
    8 & 7.604 & 10.84 & 0.434         & 0.002 & 54.58 
      & 5.997 & 6.733 & 0.196 & 16.02 & 0.002 & 42.60 & 26.56
      & 6.253 & 5.410 & 0.126 & 0.284 & 0.002 & 35.62 & 1.197\\
    9 & 6.933 & 10.84 & 0.229         & 0.001 & 54.02 
      & 5.077 & 6.733 & 0.117 & 16.00 & 0.001 & 41.92 & 26.56
      & 5.148 & 5.410 & 0.053 & 0.289 & 0.000 & 35.01 & 1.197\\
    10 & 5.121 & 5.967 & 0.136         & 0.001 & 53.50 
      & 3.829 & 3.557 & 0.056 & 15.96 & 0.001 & 41.29 & 26.56
      & 3.940 & 2.831 & 0.026 & 0.228 & 0.000 & 34.45 & 0.994
\end{tabular}}
\caption{Maximal relative errors in $\%$ for the Markovian approximations for implied volatility surfaces of the European call option. The discretization error of these errors is at most $0.002\%$ for $H=0.001, 0.1$ and $0.020\%$ for $H=-0.1$.}
\label{tab:ErrorsIVSurfaces}
\end{table}

\begin{figure}
    \centering
    \begin{minipage}{0.5\textwidth}
        \centering
        \includegraphics[width=\linewidth]{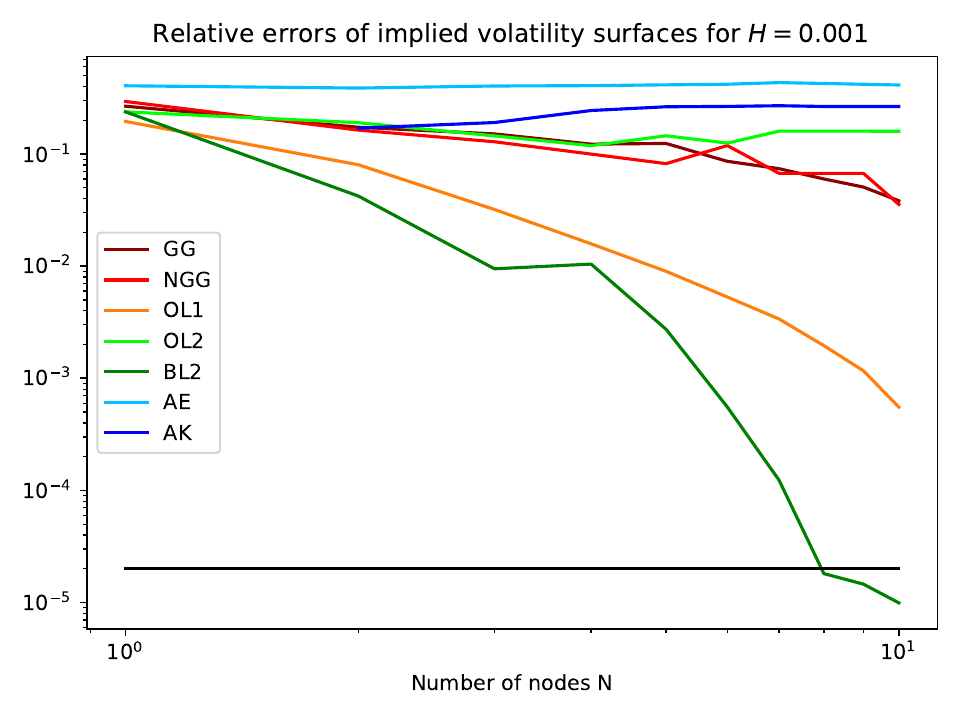}
    \end{minipage}%
    \begin{minipage}{0.5\textwidth}
        \centering
        \includegraphics[width=\linewidth]{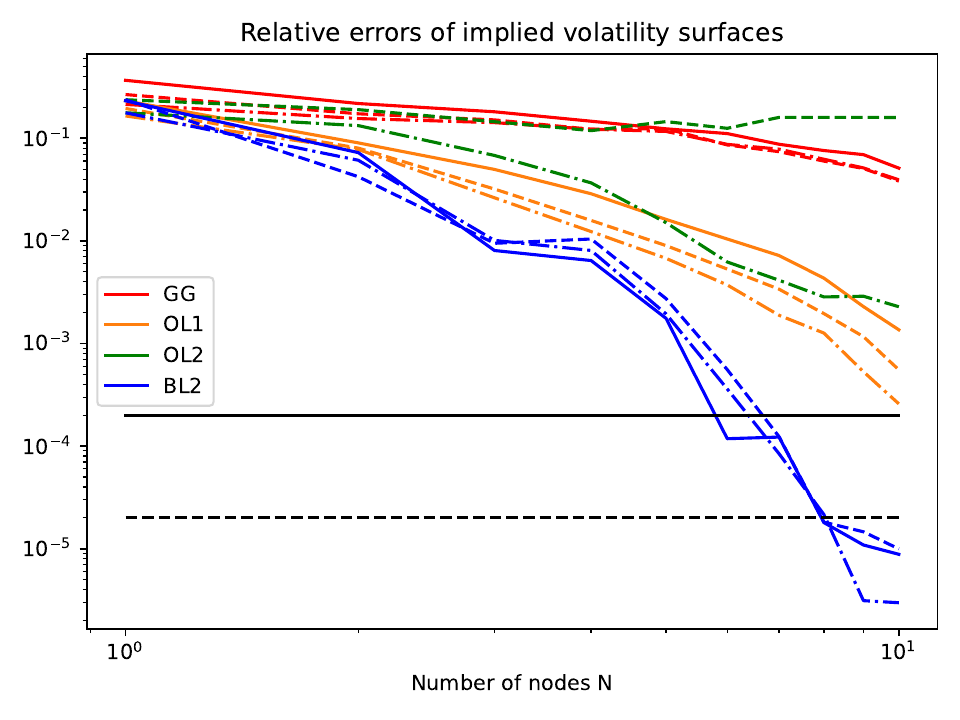}
    \end{minipage}
    \caption{Maximal relative errors of the implied volatility surfaces for various quadrature rules. The figure on the left compares all 7 algorithms for $H=0.001$. The figure on the right compares GG, OL1, OL2, BL2 for $H=-0.1$ (solid lines), $H=0.001$ (dashed lines), and $H=0.1$ (dash-dotted lines). The black lines are the accuracy of the computation of the relative errors, $2\cdot 10^{-5}$ for $H\ge 0$ and $2\cdot 10^{-4}$ for $H<0$.}
    \label{fig:ErrorsIVSurfaces}
\end{figure}

\subsubsection{Skew of European call options}\label{sec:ImpliedVolatilitySkews}

Next, we compute the skew of the implied volatility surface for European call options. We recall that in non-rough stochastic volatility models, the skew remains bounded for $T\to 0$, while in rough Heston with Hurst parameter $H$, the skew explodes like $T^{1/2-H}$. Our Markovian approximations are of course just ordinary stochastic differential equations. Hence, the skew will remain bounded as $T\to 0$ for fixed quadrature rules. However, we hope to be able to show that with the right choice of quadrature rule we can exhibit a similar explosion of the skew on reasonable time scales $T$. To this end, we use 25 geometrically spaced maturities on the interval $[0.004, 1]$, i.e. we have maturities between 1 day and 1 year that we want to jointly approximate.

Similarly to the computation of implied volatility surfaces in Section \ref{sec:ImpliedVolatilitySurfaces}, we are faced with a vector of maturities $T$, and need to choose a suitable representative $T_0$ for approximating $K$ by $K^N$. We use the same choice of $T_0$ as for the implied volatility surfaces given in \eqref{eqn:TVectorChoice}.

We computed the skews with relative errors of at most $0.01\%$ for $H=0.001$ and $H=0.1$, and with relative errors of at most $0.5\%$ for $H=-0.1$. The computational times for the skews for the Adams scheme were 6095 seconds for $H=-0.1$, 4677 seconds for $H=0.001$, and $303.5$ seconds for $H=0.1$. In contrast, typical computational times for the Markovian approximations were between 60 and 240 seconds.

In Figure \ref{fig:SkewsH0001} we see the skews for GG on the left, and for BL2 on the right. While we can still clearly see that the skews for GG with $N=10$ and the non-Markovian skew do not align perfectly, it is evident that the Markovian skews yield an explosion similar to $T^{1/2-H}$ on the time interval $[0.001, 1]$. On the other hand, we cannot make out any difference between the skew for BL2 with $N=3$ and the non-Markovian skew with the naked eye, illustrating that the Markovian approximations using BL2 capture the explosion of the skew well on reasonable time intervals. In particular, BL2 achieved errors below $1\%$ for $N=6$ and errors below the discretization error of $0.020\%$ for $N=10$ for positive $H$.

\begin{figure}
    \centering
    \begin{minipage}{0.5\textwidth}
        \centering
        \includegraphics[width=\linewidth]{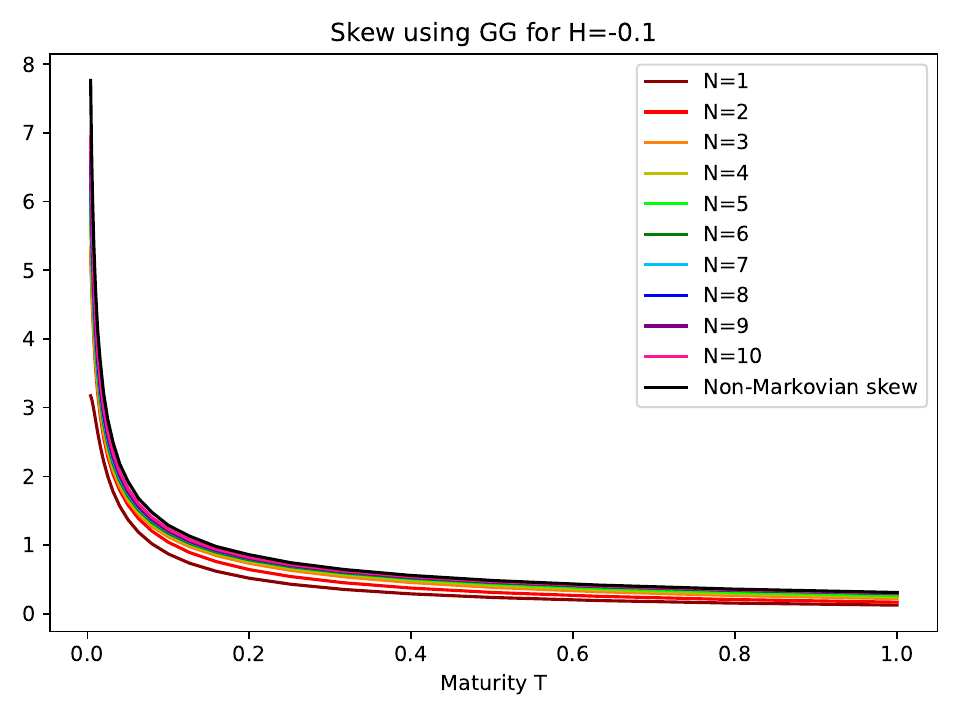}
    \end{minipage}%
    \begin{minipage}{0.5\textwidth}
        \centering
        \includegraphics[width=\linewidth]{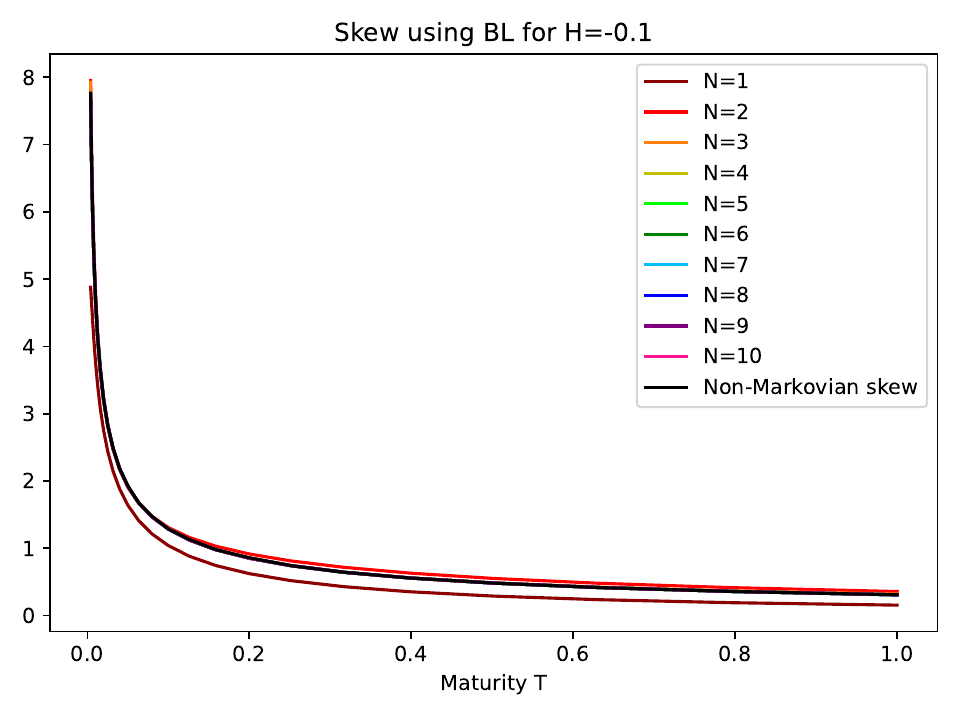}
    \end{minipage}
    \caption{Implied volatility skews for $H=0.001$ together with the Markovian approximations using GG (left) and BL2 (right).}
    \label{fig:SkewsH0001}
\end{figure}

\subsubsection{Pricing digital European call options}\label{sec:DigitalOptions}

Finally, we consider digital European call options. This illustrates that it is not strictly necessary for the payoff function to be smooth. We consider the option prices themselves and not implied volatilities, as these are not necessarily well-defined for digital options. We take $T=1$ and 301 linearly spaced values of log-moneyness in the interval $[-1, 0.5]$.

All option prices were computed with a relative accuracy of at least $10^{-5}$. The computational times for the option prices for the Adams scheme were 950.8 seconds for $H=-0.1$, 988.9 seconds for $H=0.001$, and $110.8$ seconds for $H=0.1$. In contrast, typical computational times for the Markovian approximations were between 1 and 5 seconds.

In Figure \ref{fig:ErrorsPricesDigital} we see the maximal relative errors for $H=0.001$ for all quadrature rules on the left, and a comparison for varying $H$ on the right. Again, BL2 performed best overall, achieving errors around $1\%$ for $N=3$ and errors below the discretization error of $0.002\%$ for $N=7$, largely independent of $H$.

\begin{figure}
    \centering
    \begin{minipage}{0.5\textwidth}
        \centering
        \includegraphics[width=\linewidth]{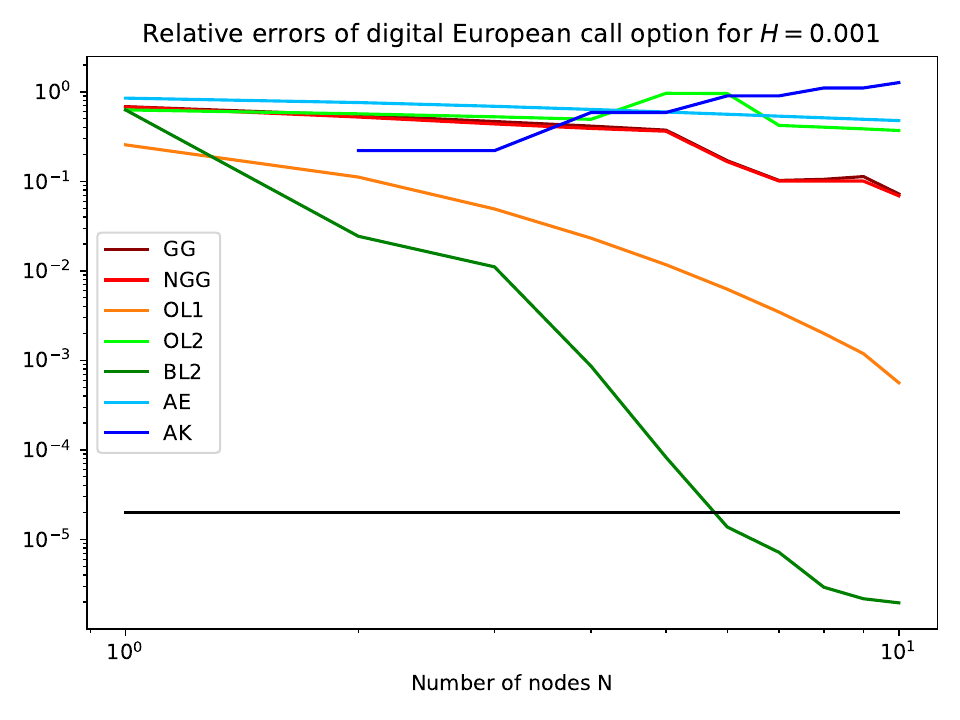}
    \end{minipage}%
    \begin{minipage}{0.5\textwidth}
        \centering
        \includegraphics[width=\linewidth]{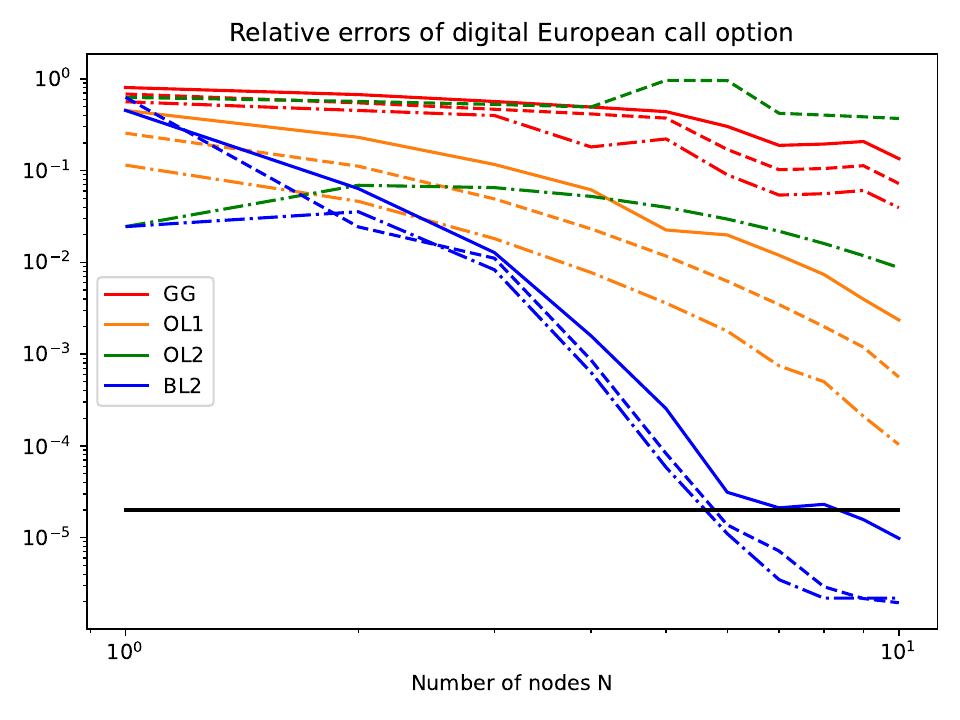}
    \end{minipage}
    \caption{Maximal relative errors of the digital European call prices for various quadrature rules. The figure on the left compares all 7 algorithms for $H=0.001$. The figure on the right compares GG, OL1, OL2, BL2 for $H=-0.1$ (solid lines), $H=0.001$ (dashed lines), and $H=0.1$ (dash-dotted lines). The black line is at $2\cdot 10^{-5}$, which is the accuracy of the computation of the relative errors.}
    \label{fig:ErrorsPricesDigital}
\end{figure}

\appendix

\section{Moment explosion of the rough Heston model}

In this section, we recall and expand on the results of \cite{gerhold2019moment} on the existence of moments of $S_T$. For all $q\in\R$, we denote $$M_q\coloneqq \E S_T^q \in (0, \infty],\qquad \widetilde M_q \coloneqq \E S_T^q + \sup_{N\in\N} \E (S_T^N)^q \in (0, \infty].$$

\begin{lemma}\label{lem:MomentsExist}
Let $q\in\R$. Then, $\widetilde M_q < \infty$ for all $T>0$ if and only if $M_q < \infty$ for all $T>0$. This is the case if and only if $$q\in[0, 1],\ \textup{or}\ \left(\rho\nu q - \lambda < 0\ \textup{and}\ (\rho\nu q - \lambda)^2 - \nu^2q(q-1)\ge 0\right).$$ In particular, the set of $q\in\R$ for which $M_q < \infty$ for all $T>0$ is an interval $\mathcal{I}$, and there is an $\eps > 0$ such that 
    \begin{align}
        \mathcal{J} \coloneqq 
        \begin{cases}
            [-\eps, 1 + \eps],\qquad &\rho\le 0\textup{ or } \lambda > \rho\nu,\\
            [-\eps, 1],\qquad &\textup{else}
        \end{cases}\label{eqn:IntervalJ}
    \end{align}
    is contained in $\mathcal{I}$. In particular, some negative moment of $S_T$ always exists for all $T$.
\end{lemma}

\begin{proof}
    By \cite[Theorem 2.4]{gerhold2019moment}, $M_q < \infty$ for all $T>0$ if and only if $q\in[0, 1]$ or
    \begin{equation}\label{eqn:TheseBloodyConditions}
       \rho\nu q - \lambda < 0\ \textup{and}\ (\rho\nu q - \lambda)^2 - \nu^2q(q-1)\ge 0. 
    \end{equation}
    The first case ($q\in[0, 1]$) is proved in \cite[Proposition 3.7]{gerhold2019moment}, while the second case is proved in \cite[Proposition 3.6]{gerhold2019moment}. In both cases, the function $f$ in these proofs can be bounded by a constant $a$ that depends on the function $G$ in these proofs. The function $G$ does not depend on the kernel, making the bound independent of $N$. Also, the kernels $K^N$ satisfy all the assumptions of \cite[Proposition 3.1]{gerhold2019moment}, so that the arguments in the proof remain unchanged for $K^N$. Hence, $\widetilde M_q < \infty$ in both cases. On the other hand, if $\widetilde M_q < \infty$, then trivially $M_q < \infty$, proving the first claim of the lemma.

    Next, we want to prove that the set $q\in\R$ with $M_q < \infty$ is an interval $\mathcal{I}$, and that the interval $\mathcal{J}$ in \eqref{eqn:IntervalJ} is contained in $\mathcal{I}$. Note that conditional on the first condition in \eqref{eqn:TheseBloodyConditions} being true, we have $(\rho\nu q - \lambda)^2 > 0$, and the second condition in \eqref{eqn:TheseBloodyConditions} is necessarily true on an interval $q\in [-\delta_1, 1 + \delta_2]$ with $\delta_1, \delta_2 \in (0, \infty]$. Since the second condition is a quadratic polynomial in $q$ (for $\nu > 0$, with negative leading sign), the second condition is in fact satisfied on an interval of this form.
    
    If $\rho < 0$, the condition $\rho\nu q - \lambda < 0$ amounts to the interval $(-\frac{\lambda}{\abs{\rho}\nu},\infty),$ proving the claim for $\rho < 0$.

    If $\rho > 0$, the condition $\rho\nu q - \lambda < 0$ amounts to the interval $(-\infty, \frac{\lambda}{\rho\nu})$, proving the statement for $\rho > 0$. 

    Finally, if $\rho = 0$, the condition $\rho\nu q - \lambda < 0$ is always satisfied, yielding the interval $\R$, proving the statement for $\rho = 0$.
\end{proof}

\section{Localizing functions}

In this section we prove a simple lemma on approximating smooth functions on $\R$ by compactly supported smooth functions. While the following results are undoubtedly well-known, we were not able to find precise references.

\begin{lemma}\label{lem:MakeFunctionCompact}
    Let $n\ge 1$, and let $f:\R\to\R$ be an $n$ times weakly differentiable function. Let $-\infty < a < b < \infty$. Then there exists a function $g:\R\to\R$ which is $n$ times weakly differentiable such that
    \begin{align*}
    g(x) &= \begin{cases}
        f(x),\qquad &x\le a,\\
        0,\qquad &x \ge b,
    \end{cases}\\
    \abs{g^{(k)}(x)} &\le \sum_{i=0}^k \binom{k}{i} c_{k-i}(b-a)^{-(k-i)}\abs{f^{(i)}(x)},\qquad k=0,\dots,n,\ x\in[a, b],
    \end{align*}
    where $c_k$, $k=0,\dots,n,$ are absolute constants independent of everything but $k$.
\end{lemma}

\begin{proof}
    Let $\sigma\in C^\infty(\R, \R)$ be non-negative and supported on $[0, 1]$ with $$\int_{\R} \sigma(x) \sdd x = 1,\qquad c_k \coloneqq \sup_{x\in\R}\abs{\sigma^{(k)}(x)},\qquad k\in\N_0.$$ Define the function $$\mu:\R\to\R,\qquad \mu(x) = \frac{1}{b-a}\int_x^\infty \sigma\left(\frac{b-y}{b-a}\right)\sdd y = \frac{1}{b-a}\int_0^\infty \sigma\left(\frac{b-y-x}{b-a}\right)\sdd y.$$ We immediately get that $\mu$ is infinitely differentiable and non-increasing, that $\mu(x) = 1$ for $x\le a$, that $\mu(x) = 0$ for $x\ge b$, and that
    \begin{align*}
    \abs{\mu^{(k)}(x)} &= \abs{\frac{\dd^k}{\dd x^k}\frac{1}{b-a}\int_0^\infty \sigma\left(\frac{b-y-x}{b-a}\right)\sdd y}\\
    &= (b-a)^{-k-1}\abs{\int_0^\infty \sigma^{(k)}\left(\frac{b-y-x}{b-a}\right)\sdd y} \le c_k(b-a)^{-k}.
    \end{align*}

    Now, define $g(x) \coloneqq f(x) \mu(x)$. Obviously, $g(x) = f(x)$ for $x\le a$ and $g(x) = 0$ for $x\ge b$. Moreover,
    \begin{align*}
    \abs{g^{(k)}(x)} &= \abs{\sum_{i=0}^k \binom{k}{i} f^{(i)}(x) \mu^{(k-i)}(x)} \le \sum_{i=0}^k \binom{k}{i} c_{k-i}(b-a)^{-(k-i)}\abs{f^{(i)}(x)}.\qedhere
    \end{align*}
\end{proof}

\begin{corollary}\label{cor:MakeFunctionCompact}
    Let $n\ge 1$, and let $f:\R\to\R$ be an $n$ times weakly differentiable function. Let $-\infty < a_1 < b_1 < a_2 < b_2 < \infty$. Then there exists a function $g:\R\to\R$ which is $n$ times weakly differentiable such that
    \begin{align*}
    g(x) &= \begin{cases}
        f(x),\qquad &x\in[b_1,a_2],\\
        0,\qquad &x \in (-\infty,a_1] \cup [b_2,\infty),
    \end{cases}\\
    \abs{g^{(k)}(x)} &\le \sum_{i=0}^k \binom{k}{i} c_{k-i}(b_j-a_j)^{-(k-i)}\abs{f^{(i)}(x)},\quad k=0,\dots,n,\ x\in[a_j, b_j],\ j=1,2,
    \end{align*}
    where $c_k$, $k=0,\dots,n,$ are absolute constants independent of everything but $k$.
\end{corollary}

\begin{proof}
    Just apply Lemma \ref{lem:MakeFunctionCompact} twice.
\end{proof}

\section{Proof of Lemma \ref{lem:ImprovedGaussianErrorForSpecificPowerFunction}}\label{sec:ProofOfImprovedGaussianErrorForSpecificPowerFunction}

The aim of this section is to prove Lemma \ref{lem:ImprovedGaussianErrorForSpecificPowerFunction}. First, some technical lemmas are needed. The interesting case in the following lemma is $\mu \to \infty$. The purpose of this lemma is that it allows us to remove a quantity $\delta$ in the rate of convergence at the cost of the factor $e$.

\begin{lemma}\label{lem:RemoveThatDeltaFromTheRate}
Let $M > 1$ and $\mu \ge \frac{1}{M-1}\lor\frac{3M}{2\sqrt{M^2-1}}.$ Then, 
$$\inf_{\delta\in(0, M-1)} \delta^{-1} \left( M - \delta + \sqrt{\left(M-\delta\right)^2 - 1}\right)^{-\mu} \le e(\delta^*)^{-1} (M + \sqrt{M^2-1})^{-\mu},$$ where $$\delta^* \coloneqq \frac{\sqrt{\mu^2(M^2-1) + 1} - M}{\mu^2 - 1}$$ is where the infimum is attained.
\end{lemma}

\begin{proof}
Taking the derivative with respect to $\delta$ and setting it $0$ conveniently leads to a quadratic equation with the (unique positive) solution $\delta = \delta^*,$ where the infimum is attained. For simplicity, we write $\delta$ instead of $\delta^*$.

Define $c\coloneqq \sqrt{M^2-1}\mu$ for brevity. Then,
\begin{align*}
M-\delta + \sqrt{(M-\delta)^2-1} &= \frac{1}{\mu^2-1}\Bigg(M\mu^2 - \sqrt{c^2 + 1}\\
&\quad + \sqrt{c^2\mu^2 - 2M\mu^2\sqrt{c^2+1} + (M^2 + 1)\mu^2}\Bigg).
\end{align*}

Our goal is now to show the lower bound \eqref{eqn:ThisLowerBoundIsTheGoal} below. Since $x\mapsto \sqrt{x}$ is concave, we have $\sqrt{c^2 + 1} \le c + \frac{1}{2c}.$ Hence,
\begin{align*}
M-\delta + \sqrt{(M-\delta)^2-1} &\ge \frac{1}{\mu^2-1}\Bigg(M\mu^2 - c - \frac{1}{2c} + \mu\sqrt{c^2 - 2Mc - \frac{M}{c} + M^2 + 1}\Bigg).
\end{align*}
Furthermore,
\begin{align*}
\sqrt{c^2 - 2Mc - \frac{M}{c} + M^2 + 1} &= \sqrt{\left(c - M\right)^2 + 1 - \frac{M}{c}} = (c - M)\sqrt{1 + \frac{1}{c(c - M)}}.
\end{align*}

Since $c>M$ by the lower bound on $\mu$ in the statement of the lemma, we can apply the bound $$\sqrt{1 + x} \ge 1 + \frac{x}{2} - \frac{x^2}{8},$$ valid for all $x\ge 0$, to get
\begin{align*}
    (c-M)\sqrt{1 + \frac{1}{c(c - M)}} &\ge c-M + \frac{1}{2c} - \frac{1}{8c^2(c - M)},
\end{align*}
yielding in total the bound 
\begin{align*}
    M-\delta + \sqrt{(M-\delta)^2-1} &\ge \frac{1}{\mu^2-1}\Bigg((M\mu + c)(\mu - 1) + \frac{\mu - 1}{2c} - \frac{\mu}{8c^2(c - M)}\Bigg)
\end{align*}

We now want to show that the last two summands are non-negative. Indeed, since $\mu \ge \frac{3}{2}\frac{M}{\sqrt{M^2 - 1}},$ we have $c = \sqrt{M^2 - 1} \mu \ge \frac{3}{2}M,$ and hence,
\begin{align*}
    \frac{\mu - 1}{2c} - \frac{\mu}{8c^2(c - M)} &= \frac{1}{2c}\left(\mu - 1 - \frac{\mu}{4c(c - M)}\right) \ge \frac{1}{2c}\left(\mu - 1 - \frac{\mu}{3M^2}\right)\\
    &> \frac{1}{2c}\left(\frac{2}{3}\mu - 1\right) \ge \frac{1}{2c}\left(\frac{M}{\sqrt{M^2 - 1}} - 1\right) > 0,
\end{align*}
where we additionally used $M > 1.$ Therefore, we have the bound
\begin{align}
    M-\delta + \sqrt{(M-\delta)^2-1} &\ge \frac{\mu - 1}{\mu^2-1}(M\mu + c) = \left(M + \sqrt{M^2 - 1}\right)\frac{\mu}{\mu + 1}.\label{eqn:ThisLowerBoundIsTheGoal}
\end{align}

We remark that all the computations were carried out in such a way that the error in this lower bound is of order $O(\mu^{-2})$ as $\mu\to\infty$, for fixed $M>1$. Finally,
\begin{align*}
\inf_{\delta\in(0, M-1)} \delta^{-1} \left( M - \delta + \sqrt{\left(M-\delta\right)^2 - 1}\right)^{-\mu} &\le \delta^{-1} \left(\left(M + \sqrt{M^2-1}\right)\left(\frac{\mu}{\mu+1}\right)\right)^{-\mu}\\
&\le e\delta^{-1} \left(M + \sqrt{M^2-1}\right)^{-\mu}.
\end{align*}

It merely remains to verify that $\delta < M-1$. We have
\begin{align*}
\delta &= \frac{\sqrt{\mu^2(M^2-1) + 1} - M}{\mu^2-1} \le \frac{M\mu - M}{\mu^2-1} = \frac{M}{\mu + 1}.
\end{align*}
Then, $\delta < M-1$ if $\mu \ge \frac{1}{M-1}.$
\end{proof}

In the proof of Lemma \ref{lem:ImprovedGaussianErrorForSpecificPowerFunction}, we will represent $f$ using its Chebyshev series. Consequently, we need bounds for the coefficients of that series. The following Lemma gives such a bound. We remark that a bound of the form $\eps^{-\gamma}r^{-k}$ below would be trivial by bounding $f$ by its maximum. The improvement to $\eps^{1-\gamma}r^{-k}$ is achieved by noting that the integral only spends $\eps$ amount of time close to the singularity.

\begin{lemma}\label{lem:BoundOfChebyshevCoefficient}
Let $f:\C\backslash\{-c\}\to\C$, $f(x) = (x + c)^{-\gamma},$ where $c > 1$ and $\gamma > 1$. Let $r>1$ and $\eps > 0$ be such that $c = \frac{r + r^{-1}}{2} + \eps$, and let $k\in\R$. Then,
\begin{align*}
\abs{\frac{1}{\pi i}\int_{\abs{z} = 1}z^{-1-k} f\left(\frac{z + z^{-1}}{2}\right) \sdd z} &\le \frac{4/\pi}{r - r^{-1}}\Bigg(1 + \frac{(\pi/2)^\gamma}{\gamma - 1}\Bigg)\eps^{1-\gamma}r^{-k}.
\end{align*}
\end{lemma}

\begin{proof}
Define $\delta \coloneqq \frac{2}{r + r^{-1}}\eps,$ so that $c = (1 + \delta)\frac{r + r^{-1}}{2}.$ We have
\begin{align*}
\bigg|\frac{1}{\pi i} \int_{\abs{z} = 1} &z^{-1-k} f\left(\frac{z + z^{-1}}{2}\right) \sdd z\bigg|\\
&= \abs{\frac{1}{\pi i} \int_{\abs{z} = r}z^{-1-k} f\left(\frac{z + z^{-1}}{2}\right) \sdd z}\\
&\le \frac{r^{-k}}{\pi} \int_0^{2\pi} \abs{\frac{r e^{i\theta} + r^{-1}e^{-i\theta}}{2} + c}^{-\gamma}\sdd\theta\\
&= \frac{r^{-k}}{\pi} \int_0^{2\pi} \abs{\frac{r + r^{-1}}{2}\left(1 + \delta + \cos\theta\right) + i\frac{r - r^{-1}}{2}\sin\theta}^{-\gamma}\sdd\theta\\
&= \frac{r^{-k}}{\pi}\left(\frac{r + r^{-1}}{2}\right)^{-\gamma} \int_0^{2\pi} \left(\left(1 + \delta + \cos\theta\right)^2 + \left(\frac{r - r^{-1}}{r + r^{-1}}\right)^2\left(\sin\theta\right)^2\right)^{-\gamma/2}\sdd\theta.
\end{align*}

The integrand only comes close to the singularity for $\theta = \pi$. Hence, we split the integral into a part with $\theta\approx\pi$, and a part where $\theta$ is far from $\pi$. First, let us consider the part where $\theta\approx \pi$. There,
\begin{align*}
\int_{\pi - \frac{r+r^{-1}}{r-r^{-1}}\delta}^{\pi + \frac{r+r^{-1}}{r-r^{-1}}\delta} \left(\left(1 + \delta + \cos\theta\right)^2 + \left(\frac{r - r^{-1}}{r + r^{-1}}\right)^2\left(\sin\theta\right)^2\right)^{-\gamma/2}\sdd\theta \le 2\frac{r+r^{-1}}{r-r^{-1}} \delta^{1-\gamma}.
\end{align*}

For the remainder of the integral, we note that
\begin{align*}
\Bigg(\int_0^{\pi - \frac{r+r^{-1}}{r-r^{-1}}\delta} &+ \int_{\pi + \frac{r+r^{-1}}{r-r^{-1}}\delta}^{2\pi}\Bigg) \left(\left(1 + \delta + \cos\theta\right)^2 + \left(\frac{r - r^{-1}}{r + r^{-1}}\right)^2\left(\sin\theta\right)^2\right)^{-\gamma/2}\sdd\theta\\
&=2 \int_0^{\pi - \frac{r+r^{-1}}{r-r^{-1}}\delta} \left(\left(1 + \delta + \cos\theta\right)^2 + \left(\frac{r - r^{-1}}{r + r^{-1}}\right)^2\left(\sin\theta\right)^2\right)^{-\gamma/2}\sdd\theta.
\end{align*}

Define $\xi \coloneqq \theta-\pi$. Since $\abs{\xi} \le \pi$, we have
\begin{align*}
\left(1 + \delta + \cos\theta\right)^2 + \left(\frac{r - r^{-1}}{r + r^{-1}}\right)^2\left(\sin\theta\right)^2 &\ge \frac{4}{\pi^2}\left(\frac{r - r^{-1}}{r + r^{-1}}\right)^2\xi^2.
\end{align*}
This implies for the integral that
\begin{align*}
\int_0^{\pi - \frac{r+r^{-1}}{r-r^{-1}}\delta} &\left(\left(1 + \delta + \cos\theta\right)^2 + \left(\frac{r - r^{-1}}{r + r^{-1}}\right)^2\left(\sin\theta\right)^2\right)^{-\gamma/2}\sdd\theta\\
&\le \int_{\frac{r+r^{-1}}{r-r^{-1}}\delta}^{\pi} \frac{\pi^\gamma}{2^\gamma}\left(\frac{r - r^{-1}}{r + r^{-1}}\right)^{-\gamma}\xi^{-\gamma}\sdd\theta\\
&= \frac{\pi^\gamma}{2^{\gamma}} \left(\frac{r - r^{-1}}{r + r^{-1}}\right)^{-\gamma}\frac{1}{\gamma - 1}\left(\left(\frac{r+r^{-1}}{r-r^{-1}}\delta\right)^{1-\gamma} - \pi^{1-\gamma}\right)\\
&\le \frac{\pi^\gamma}{2^{\gamma}} \frac{1}{\gamma - 1}\left(\frac{r+r^{-1}}{r-r^{-1}}\right)\delta^{1-\gamma}.
\end{align*}

Putting these estimates together, we obtain
\begin{align*}
\abs{\frac{1}{\pi i} \int_{\abs{z} = 1}z^{-1-k} f\left(\frac{z + z^{-1}}{2}\right) \sdd z} &\le \frac{2r^{-k}}{\pi}\left(\frac{r + r^{-1}}{2}\right)^{-\gamma}\frac{r+r^{-1}}{r-r^{-1}} \delta^{1-\gamma}\\
&\qquad + \frac{2r^{-k}\pi^{\gamma-1}}{2^{\gamma}} \left(\frac{r + r^{-1}}{2}\right)^{-\gamma} \frac{1}{\gamma - 1}\left(\frac{r+r^{-1}}{r-r^{-1}}\right)\delta^{1-\gamma}\\
&\le \frac{2}{\pi}\left(\frac{r + r^{-1}}{2}\right)^{-\gamma}\frac{r+r^{-1}}{r-r^{-1}}\Bigg(1 + \frac{(\pi/2)^\gamma}{\gamma - 1}\Bigg)\delta^{1-\gamma}r^{-k}.
\end{align*}
Transforming $\delta$ back to $\eps$ yields the result.
\end{proof}

We can now proceed with the proof of Lemma \ref{lem:ImprovedGaussianErrorForSpecificPowerFunction}.

\begin{proof}[Proof of Lemma \ref{lem:ImprovedGaussianErrorForSpecificPowerFunction}]
For the proof we follow \cite[Theorem 19.3]{trefethen2019approximation} very closely. In fact, the main difference is that we use a sharper bound on the Chebyshev coefficients $a_k$ below.

Let $T_k$ be the Chebyshev polynomial of degree $k$. By \cite[Theorem 3.1]{trefethen2019approximation}, $f$ has a representation as a Chebyshev series $$f(x) = \sum_{k=0}^\infty a_k T_k(x).$$ Since Gaussian quadrature of level $m$ is exact for polynomials up to degree $2m-1$, and, by symmetry, furthermore exact for odd functions, and since the Chebyshev polynomials of odd degree are odd functions, we get the error bound 
\begin{align*}
\abs{\int_{-1}^1 f(x) \sdd x - \sum_{i=1}^m w_i f(x_i)} &\le \sum_{k=0}^\infty \abs{a_{2m + 2k}}\abs{\int_{-1}^1 T_{2m+2k}(x) \sdd x - \sum_{i=1}^m w_i T_{2m+2k}(x_i)}.
\end{align*}
Since $\sum_{i=1}^m w_i = 2$ and $\abs{T_k(x)} \le 1$ for $x\in[-1, 1]$, \cite[Theorem 19.2]{trefethen2019approximation} implies that
\begin{align*}
\abs{\int_{-1}^1 f(x) \sdd x - \sum_{i=1}^m w_i f(x_i)} &\le \sum_{k=0}^\infty \abs{a_{2m + 2k}}\left(2 + \frac{2}{(2m+2k)^2-1}\right).
\end{align*}

By \cite[Theorem 3.1]{trefethen2019approximation}, and in particular equation (3.13) therein, we have $$a_k = \frac{1}{\pi i}\int_{\abs{z}=1} z^{-1-k}f\left(\frac{z + z^{-1}}{2}\right) \sdd z.$$ By Lemma \ref{lem:BoundOfChebyshevCoefficient}, we have
\begin{align*}
\abs{a_k} &\le \frac{4/\pi}{r - r^{-1}}\Bigg(1 + \frac{(\pi/2)^\gamma}{\gamma - 1}\Bigg)\eps^{1-\gamma}r^{-k}.
\end{align*}
where $r > 1$ and $\eps > 0$ are chosen such that $c = \frac{r + r^{-1}}{2} + \eps.$ Therefore,
\begin{align*}
\abs{\int_{-1}^1 f(x) \sdd x - \sum_{i=1}^m w_i f(x_i)} &\le \sum_{k=0}^\infty \frac{4/\pi}{r - r^{-1}}\Bigg(1 + \frac{(\pi/2)^\gamma}{\gamma - 1}\Bigg)\eps^{1-\gamma}r^{-2m-2k}\left(2 + \frac{2}{4m^2-1}\right)\\
&= \frac{8}{\pi}\frac{4m^2}{4m^2-1}\frac{r}{(r - r^{-1})^2}\Bigg(1 + \frac{(\pi/2)^\gamma}{\gamma - 1}\Bigg)\eps^{1-\gamma}r^{-2m}.
\end{align*}

Note that that $r = c-\eps + \sqrt{(c-\eps)^2-1}.$ In line with Lemma \ref{lem:RemoveThatDeltaFromTheRate}, we define $\mu\coloneqq \frac{2m}{\gamma-1}$, and assume that $\mu \ge \frac{1}{c-1}\lor \frac{3c}{2\sqrt{c^2-1}}.$ We choose $$\eps = \frac{\sqrt{\mu^2(c^2-1)+1} - c}{\mu^2-1},$$ and get
\begin{align*}
\abs{\int_{-1}^1 f(x) \sdd x - \sum_{i=1}^m w_i f(x_i)} &\le \frac{8}{\pi}\frac{4m^2}{4m^2-1}\frac{r}{(r - r^{-1})^2}\Bigg(1 + \frac{(\pi/2)^\gamma}{\gamma - 1}\Bigg)e^{\gamma-1}\\
&\qquad \times \eps^{1-\gamma}\left(c + \sqrt{c^2-1}\right)^{-2m}.\qedhere
\end{align*}
\end{proof}

\section{Gaussian Approximations of the fractional kernel}\label{sec:GaussianApproximationsFractionalKernel}

This section is devoted to some simple results on both geometric (Section \ref{sec:GeometricGaussianApproximations}) and non-geometric (Section \ref{sec:NonGeometricGaussianApproximations}) Gaussian approximations of the fractional kernel $K$. All results in this section are simple corollaries of the following general error representation formula of Gaussian quadrature.

\begin{theorem}\label{thm:GaussGeneralErrorRepresentationFormula}\cite[Theorem 3.6.24]{bulirsch2002introduction}
If $f\in C^{2n}([a, b])$, and $(x_i)_{i=1}^n$ are the nodes and $(w_i)_{i=1}^n$ are the weights of Gaussian quadrature with respect to the weight function $w$, then $$\int_a^b f(x) w(x) \sdd x - \sum_{i=1}^n w_i f(x_i) = \frac{f^{(2n)}(\xi)}{(2n)!}\int_a^b w(x) p_n(x)^2 \sdd x,$$ where $\xi$ is some point with $\xi\in[a, b]$, and $p_n$ is a specific polynomial of degree $n$.
\end{theorem}

The following corollary will not be explicitly used in the proof of the convergence rates, but it is perhaps an interesting observation.

\begin{corollary}\label{cor:GaussianApproximationLowerBiased}
Let $K^N$ be an approximation of $K$ stemming from a geometric or non-geometric Gaussian rule. Then, 
\begin{align*}
K(t) \ge K^N(t).
\end{align*}
\end{corollary}

\begin{proof}
This follows immediately from Proposition \ref{prop:GaussianQuadratureUnderestimatesCompletelyMonotoneFunctions} once we note that both $x\mapsto e^{-xt}$ and $x\mapsto e^{-xt} x^{-H-1/2}$ are completely monotone.
\end{proof}

The next corollary gives us an error representation formula for Gaussian approximations. Note that this formula is not valid for general approximations. Also, we will not use it in the proof of the convergence rates. However, we use this corollary to compute the actual errors of Gaussian approximations in the numerical part.

\begin{corollary}\label{cor:L1GaussianApproximationErrorRepresentation}
Let $(x_i)_{i=1}^N$ be the nodes and $(w_i)_{i=1}^N$ be the weights of a geometric or non-geometric Gaussian rule, and let $K^N$ be the corresponding approximation. Then, 
\begin{align*}
\int_0^T \abs{K(t) - K^N(t)} \sdd t &= \frac{T^{H+1/2}}{\Gamma(H+3/2)} - \sum_{i=1}^N \frac{w_i}{x_i}\left(1 - e^{-x_iT}\right).
\end{align*}
\end{corollary}

\begin{proof}
By Corollary \ref{cor:GaussianApproximationLowerBiased}, we have
\begin{align*}
\int_0^T \abs{K(t) - K^N(t)} \sdd t &= \int_0^T\left(K(t) - K^N(t)\right) \sdd t\\
&= \int_0^T \frac{t^{H-1/2}}{\Gamma(H+1/2)}\sdd t - \sum_{i=1}^N w_i \int_0^T e^{-x_it} \sdd t\\
&= \frac{T^{H+1/2}}{\Gamma(H+3/2)} - \sum_{i=1}^N \frac{w_i}{x_i}\left(1 - e^{-x_iT}\right).\qedhere
\end{align*}
\end{proof}

\section{Computing $L^1$-errors}\label{sec:L1NormComputation}

Especially for the algorithm OL1 it is necessary to be able to compute the $L^1$ error between $K^N$ and $K$ quickly and with high accuracy. We remark that Corollary \ref{cor:L1GaussianApproximationErrorRepresentation} is not applicable since we are dealing with arbitrary approximations $K^N$ that do not stem from Gaussian rules. Simple quadrature rules like the trapezoidal rule prove unsuitable due to the singularity of the kernel $K(t)$ at $t=0$, as we will see below. Also, simple changes of variables to remove that singularity do not sufficiently solve this problem. Hence, we describe here an algorithm that allows us to compute the $L^1$-error within a few milliseconds for moderate values of $N$. More precisely, we give below three different algorithms in increasing sophistication and efficiency.
\begin{enumerate}
    \item \textbf{Trapezoidal:} This is just the trapezoidal rule, where we use the midpoint rule on the first interval due to the singularity in $K$.
    \item \textbf{Trapezoidal exact singularity:} We determine the first root $t_1 = t>0$ of $K(t) - K^N(t) = 0$. On $[0, t_1]$ we integrate exactly, on $[t_1, T]$ we use the trapezoidal rule.
    \item \textbf{Intersections:} We determine all roots $(t_i)_{i=0}^k$ when $K(t) - K^N(t) = 0$ (where $t_0 = 0$ and $t_k = T$), and compute the error exactly using that $$\int_0^T \abs{K(t) - K^N(t)} \sdd t = \sum_{i=0}^{k-1} \abs{\int_{t_i}^{t_{i+1}} \left(K(t) - K^N(t)\right) \sdd t},$$ and the fact that we can integrate $K$ and $K^N$ in closed form.
\end{enumerate}

Let us now describe how we find the roots of $K(t) - K^N(t)$, where we remark that there are at most $O(N)$ such roots. Assume that we are given a relative error tolerance $\tol$, and that we only have to compute the error up to this relative error tolerance.

To find the roots, we exploit that both $K$ and $K^N$ are completely monotone. In particular, this implies the inequalities $$K(t_0) + K'(t_0)(t-t_0) \le K(t) \le K(t_0),$$
$$K(t) \le K(t_0) + K'(t_0)(t-t_0) + \frac{1}{2}K''(t_0)(t-t_0)^2$$
for $t \ge t_0$, with similar inequalities for $K^N$.

We then find the crossings inductively, by moving from $t=0$ to $t=T$. First, obviously, $K(0) > K^N(0)$. Next, we solve $$K(t) = K^N(0).$$ This has an explicit solution $\widehat{t}>0$, and thanks to complete monotonicity, there was no crossing on $[0, \widehat{t})$.

Suppose now that we are currently at time $s\in(0,T)$, and we want to take the next step. Assume without loss of generality that $K(s) \ge K^N(s).$ In the other case, we just swap the roles of $K$ and $K^N$. We now differentiate between two cases.

\begin{enumerate}
    \item $\frac{K(s) - K^N(s)}{K(s)} > \tol.$ In this case, we are rather far away from the next crossing. We then take the next point $t$ as the solution $t>s$ of the equation $$K(s) + K'(s)(t-s) = K^N(s) + (K^N)'(s)(t-s) + \frac{1}{2}(K^N)''(s)(t-s)^2.$$ This quadratic equation has only one explicit solution $t>s$, and thanks to complete monotonicity, the kernels did not cross on $[s,t)$. 
    \item $\frac{K(s) - K^N(s)}{K(s)} \le \tol.$ In this case, we might be very close to a crossing. 
    If we proceeded exactly as in Case 1, the steps would become arbitrarily small if a crossing is ahead. To ensure that the algorithm terminates at some point, we always need to take a step that is lower bounded by some constant. We achieve this by choosing $t>s$ such that $$\sup_{\tau\in[s, t]}\abs{\frac{K(s) - K^N(s)}{K(s)} - \frac{K(\tau) - K^N(\tau)}{K(\tau)}} \le \textup{TOL}.$$ This in particular ensures that the relative error on $[s,t]$ is bounded by $2\textup{TOL}$, so we again do not overshoot. The above condition can be phrased as two inequalities without the absolute value sign. We then approximate $K(\tau)$ and $K^N(\tau)$ by their zeroth or first order Taylor polynomial (depending on the direction of the inequality). This yields two linear inequalities, and we choose $t$ to be the largest value satisfying both inequalities. Therefore, in the resulting interval $[s,t]$, the relative error between the kernels is bounded by $2\textup{TOL}$.
\end{enumerate}

We have thus given an algorithm on how to travel along the interval $[0, T]$.
It is not difficult to show that the algorithm reaches the final point $T$ in finitely many steps.
Next, we describe how we determine where the crossings of $K(t) - K^N(t)$ are.

Assume again that we are currently at some point $s\in[0, T)$, that $K(s) \ge K^N(s)$, and that the next point we step to is $t\in(s, T]$. We differentiate between two cases.

\begin{enumerate}
    \item If $K(t) > K^N(t)$, we say that there were no crossings on the interval $(s, t]$.
    \item If $K(t) \le K^N(t)$, we say that there is a crossing at $(s+t)/2$, an no other crossing in $(s, t]$.
\end{enumerate}

In Tables \ref{tab:CompareL1ComputationMethods2} we compare the three methods for computing the $L^1$-errors. We see that the ``Trapezoidal'' algorithm is basically useless, while ``Intersections'' clearly outperforms ``Trapezoidal exact singularity''.

\begin{table}[!htbp]
\centering
\resizebox{\textwidth}{!}{\begin{tabular}{c|c|c|c|c|c|c|c|c|c}
 & \multicolumn{3}{c|}{Trapezoidal} & \multicolumn{3}{c|}{Trapezoidal exact singularity} & \multicolumn{3}{c}{Intersections}\\ \hline
$\tol$ & Rel. error & $n$ & Time (ms) & Rel. error & $n$ & Time (ms) & Rel. error & $n$ & Time (ms)\\ \hline
1 & $8.99\cdot 10^{-1}$ & 100 & 0.35 & $3.59\cdot 10^{-1}$ & 100 & 1.30 & $8.52\cdot 10^{-2}$ & 25 & 2.81\\
$10^{-1}$ & $7.16\cdot 10^{-1}$ & 800 & 0.87 & $1.66\cdot 10^{-2}$ & $1\cdot 10^5$ & 34.3 & $5.74\cdot 10^{-3}$ & 73 & 6.38\\
$10^{-2}$ & $2.52\cdot 10^{-2}$ & $5\cdot 10^7$ & 14270 & $1.14\cdot 10^{-3}$ & $8\cdot 10^5$ & 220 & $1.17\cdot 10^{-4}$ & 346 & 20.6\\
$10^{-3}$ &&&& $7.24\cdot 10^{-5}$ & $3\cdot 10^6$ & 915 & $6.48\cdot 10^{-7}$ & 1990 & 119\\
$10^{-4}$ &&&& $2.27\cdot 10^{-5}$ & $7\cdot 10^6$ & 1774 & $9.42\cdot 10^{-9}$ & 7120 & 416\\
$10^{-5}$ &&&& $1.39\cdot 10^{-6}$ & $3\cdot 10^7$ & 6992 & $1.40\cdot 10^{-10}$ & 5358 & 321\\
$10^{-6}$ &&&& $8.73\cdot 10^{-8}$ & $1\cdot 10^8$ & 35983 & $1.65\cdot 10^{-12}$ & 5195 & 321\\
$10^{-7}$ &&&& $2.22\cdot 10^{-8}$ & $2\cdot 10^8$ & 122490 & $4.63\cdot 10^{-13}$ & 4264 & 421\\
$10^{-8}$ &&&&&&& $1.16\cdot 10^{-13}$ & 4124 & 260\\
$10^{-9}$ &&&&&&& $1.20\cdot 10^{-13}$ & 3704 & 230
\end{tabular}}
\caption{Relative errors, number of kernel evaluations $n$ and computational time in ms for computing the $L^1$-error for given relative error tolerances $\tol$. The Markovian approximation $K^N$ stems from the BL2 algorithm with $N=10$, $H=0.05$, and $T=1$. The reference $L^1$-error was computed using the intersections algorithm with $\tol = 10^{-10}$ and was about $0.26\%$.}
\label{tab:CompareL1ComputationMethods2}
\end{table}

\section{The algorithm BL2}\label{sec:BL2}

Here, we give a rough description on the algorithm BL2 which usually achieved the best results in Section \ref{sec:Numerics}. We remark that the actual implementation additionally contains many minor tweaks to ensure better numerical stability. The implementation can be found in \url{https://github.com/SimonBreneis/approximations_to_fractional_stochastic_volterra_equations} in the function \url{european_rule} in \url{RoughKernel.py}.

Comparing the algorithms OL1 and OL2, we note the following differences: OL1 has better convergence rates and smaller nodes, while OL2 is faster to compute (due to the explicit error formula in \cite[Proposition 2.11]{bayer2023markovian}) and may outperform OL1 for large $H$ (e.g. $H=0.1$) and small $N$, see e.g. Table \ref{tab:ErrorsIVSmiles}. The idea of BL2 is to combine the advantages of OL1 and OL2 into one algorithm.

The reason for the bad asymptotic performance and huge nodes of OL2 is the overemphasis of the singularity of $K(t)$ at $t=0$ due to the square (which is less relevant for large $H$ and small $N$, explaining the good results of OL2 in this setting). We may eliminate this problem by optimizing the $L^2$-norm under the conditions that the nodes $x$ stay bounded by some constant $L$, i.e. $x\le L$. This is also where the name BL2 (Bounded $L^2$) comes from. We denote by \url{opt}$(H, T, N, L)$ the algorithm that optimizes the $L^2$-norm of the kernel $K$ with Hurst parameter $H$ on the interval $[0, T]$, using $N$ nodes of size at most $L$. This algorithm \url{opt} returns the minimized $L^2$-approximation error and the corresponding quadrature rule.

\begin{algorithm}
\caption{BL2}
\begin{algorithmic}
    \Require $H\in (0, 1/2),\ T > 0,\ N\in \N$
    \Ensure A quadrature rule with $N$ nodes
    \If{$N=1$}
    \State $\textup{err}, \textup{rule} \gets \textup{opt}(H, T, N, \infty)$
    \State \Return rule
    \EndIf
    \State $L \gets 1$
    \State $\textup{err}_1,\textup{rule}_1 \gets \textup{opt}(H, T, N-1, L)$
    \State $\textup{err}_2,\textup{rule}_2 \gets \textup{opt}(H, T, N, L)$
    \While{$\textup{err}_2 > (1-\eps) \textup{err}_1$}
    \State $L \gets Lq$
    \State $\textup{err}_1,\textup{rule}_1 \gets \textup{opt}(H, T, N-1, L)$
    \State $\textup{err}_2,\textup{rule}_2 \gets \textup{opt}(H, T, N, L)$
    \EndWhile
    \State \Return $\textup{rule}_2$
\end{algorithmic}
\end{algorithm}

It remains to find a good bound $L$. The algorithm BL2 given above does this by choosing some small initial value of $L$ (say $L=1$), and iteratively comparing the errors when using $N$ mean-reversions in $[0, L]$, with using $N-1$ mean-reversions. If the improvement is too small (a relative improvement of less than $\eps$), we increase $L$ by a factor $q>1$. The first time the improvement is significantly large (a relative improvement bigger than $\eps$), we stop and return the current optimal quadrature rule.

It remains to choose the hyperparameters $q$ and $\eps$. After lots of trial and error, the choice $q\in[1.05, 1.15]$ and (surprisingly) $\eps = 0$ was numerically found to be optimal, though we remark that especially the choice of $\eps$ may depend on the optimization algorithm used in \url{opt}.

\printbibliography

\end{document}